\documentclass[]{imsart}

\RequirePackage[numbers]{natbib}

\usepackage{amsmath}
\usepackage{amsthm}
\usepackage{amssymb,xcolor,dsfont}
\usepackage{prodint,url}
\usepackage{graphicx}

\newtheorem{lemma}{Lemma}
\newtheorem{thm}{Theorem}

\newtheorem{cor}{Corollary}
\theoremstyle{remark}
\newtheorem{re}{Remark}
\newcommand{\p}{\mathbb{P}}
\newcommand{\1}{\mathds{1}}
\newcommand{\dd}{\mathrm{d}}
\newcommand{\R}{\mathbb{R}}
\newcommand{\E}{\mathbb{E}}
\def\argmax{\mathop{\mathrm{argmax}}}
\def\argmin{\mathop{\mathrm{argmin}}}

\startlocaldefs

\newcommand\blfootnote[1]{%
  \begingroup
  \renewcommand\thefootnote{}\footnote{#1}%
  \addtocounter{footnote}{-1}%
  \endgroup
}

	\renewcommand*{\thesection}{\arabic{section}}
\renewcommand*{\thesubsection}{\arabic{section}.\arabic{subsection}}
	\renewcommand*{\theequation}{I.\arabic{equation}}
	\renewcommand*{\thetable}{I.\arabic{table}}
 \renewcommand*{\thefigure}{I.\arabic{figure}}

\endlocaldefs

\begin{document}

\begin{frontmatter}

\title{A two-sample comparison of mean \\ \mbox{survival times of uncured sub-populations} -- Part I: Nonparametric analyses}
\runtitle{Mean survival of uncured patients in two samples -- Part I}


\author{\fnms{Dennis} \snm{Dobler${}^*$}\ead[label=e2]{d.dobler@vu.nl}}
\address{Vrije Universiteit Amsterdam \\ The Netherlands \\  \printead{e2}}
\and
\author{\fnms{Eni} \snm{Musta${}^*$}\ead[label=e1]{e.musta@uva.nl}}
\address{University of Amsterdam \\ The Netherlands \\ \printead{e1}}

\runauthor{D.\ Dobler and E.\ Musta}

\blfootnote{${}^*$: The authors contributed equally and are given in alphabetical order.}

\begin{abstract}
Comparing the survival times among two groups is a common problem in time-to-event analysis, for example if one would like to understand whether one medical treatment is superior to another. In the standard survival analysis setting, there has been a lot of discussion on how to quantify such difference and what can be an intuitive, easily interpretable, summary measure. In the presence of subjects that are immune to the event of interest (`cured'), we illustrate that it is not appropriate to just compare the overall survival functions. Instead, it is more informative to compare the cure fractions and the survival of the uncured sub-populations {separately from each other.} 
Our research is {mainly} driven by the question: if the cure fraction is similar for two available treatments, how else can we determine which is preferable? To this end, we estimate the
mean survival times in the uncured fractions of both treatment groups and develop permutation tests for inference. In this first out of two connected papers, we focus on nonparametric approaches. The methods are illustrated with medical data of leukemia patients.
\end{abstract}




\end{frontmatter}

	
	\section{Introduction}
	In many applications, it is of interest to compare survival probabilities among two different samples, e.g., two treatment arms. One common approach is to test for the equality of the survival functions although it does not provide information on the size of the difference. Alternatively, as a graphical tool, one could plot the difference between the two estimated survival curves together with confidence bands. However, in practice it is preferred to have a summary measure of such difference. This facilitates the understanding and interpretation of study results even though it provides limited information since no single metric can capture the entire profile of the difference between two survival curves.  The hazard ratio (HR) is commonly used to quantify this   difference under the assumption that the ratio of the two hazard functions remains constant over time; 
 for example, the proportionality of hazard rates is central to the famous semiparametric model by Cox \cite{cox1972}. However, such assumption is often not satisfied in practice and the use of the HR would be problematic. 

 One alternative approach is given by the restricted mean survival time (RMST).
 The difference between restricted mean survival times	for different groups has been advocated as a useful summary measure that offers clinically meaningful interpretation \cite{uno2014moving,royston2011use,royston2013restricted,ambrogi2022analyzing,zhao2016restricted}. The RMST is defined as the the expected lifetime truncated at a clinically relevant time point $\tau$.
 The restriction to $\tau$ is used to accommodate the limited study duration, as a result  of which the upper tail of the survival function cannot be estimated, unless one is willing to assume a specific parametric model for extrapolation beyond the range of the observed data. 
 
 Recently, \cite{horiguchi2020} and \cite{wolski2020} investigated a random permutation method for inference on the difference in restricted mean (net) survival times.
 While their test is finitely exact under exchangeable data, \cite{horiguchi2020} stated for the case of non-exchangeable data that ``Further research to develop methods for constructing confidence intervals for RMST
difference with a small sample data is warranted. It is quite challenging to construct an exact confidence interval for the
difference in RMST.''
\cite{ditzhaus2023}
continued in the direction of this remark and analyzed a studentized permutation version of the just-mentioned approach. 
Their resulting hypothesis test is exact under exchangeability and it even controls the type-I error probability asymptotically under non-exchangeability.
Because of this additional feature of exactness under exchangeability, permutation tests also enjoy great popularity in survival analytic applications beyond the RMST:
for instance, \cite{brendel2014weighted} and \cite{ditzhaus2020more} researched permutation-based weighted log-rank tests.

In this paper, we will consider the not unusual case that some of the subjects are immune to the event of interest (`cured') instead of the classical survival problem.
The challenge arises because, as a result of censoring, the cured subjects	(for which the event never takes place) cannot be distinguished from the susceptible ones. Cure rate models, which account for the presence of a cured sub-population have become increasingly popular particularly for analyzing cancer survival and evaluating the curative effect of treatments. More in general, they have found applications in many other domains including fertility studies, credit scoring, demographic and social studies analyzing among other things time until marriage, time until rearrest of released prisoners, time until one starts smoking. For a complete review on cure models methodology and applications, we refer the reader to \cite{AK2018,peng2021cure,legrand2019cure}. In presence of immune subjects, comparing survival between two samples becomes more complicated than in the standard setting since one can compare overall survivals, cure chances and survival probabilities for the uncured subjects. Several papers have focused on testing for differences in the cure rates in a nonparametric setting. \cite{klein2007analyzing,halpern1987cure,sposto1992comparison,laska1992nonparametric}
On the other hand, different methods have been proposed to test for equality of the survival functions among the uncured sub-populations. \cite{li2005nonparametric,tamura2000comparing,zhao2010empirical,broet2003score,broet2001semiparametric}

However, as in the standard survival analysis setting, just testing for equality of the two distributions would not be sufficient for many practical applications. Thus, apart from comparing the cure probabilities, it would be meaningful to compare relevant statistical summaries for the sub-population of uncured subjects.
To this end, we propose to analyze mean survival times of the uncured.
This has the advantage of being an easy-to-interpret extension of the RMST in the present context, whereas we will not impose a time restriction.
Hence, we will use the abbreviation MST for our purposes.
To the best of our knowledge, this has rarely been investigated in the literature so far. 
One notable exception is the recent paper \cite{chen2023mean} where a semiparametric proportional model for the mean (residual) life of the uncured was proposed and analyzed in a one sample context.
Their approach will be further discussed in Part~II of the present two companion papers.

To illustrate the above-mentioned concepts, 
Figure~\ref{fig:survival_curves} visualizes different  survival models with cure fractions.
These may be summarized with the following three important numbers:
the cure fraction $p$, i.e., the height of the plateau for late time points; the time point $\tau_0$ when the plateau is reached; the mean survival time of the uncured patients.
It is evident that that mean is potentially completely unrelated to the cure fraction.
On the other hand, the mean is of course affected by $\tau_0$.
In each panel of Figure~\ref{fig:survival_curves} except for the top-right one, the two survival curves reach their plateaus at different time points.
The top-left panel exhibits the same cure fraction in both populations but clearly different mean survival times (of the uncured).
In all other panels, the cure fractions differ.
In the top-right and the bottom-left panels,
the mean survival times of the uncured populations are equal.
This can be seen by affine-linearly transforming the vertical axis such that the thus transformed survival curves cover the whole range from 1 to 0;
then, the areas under the transformed curves are the same in each of both just-mentioned panels.
Lastly, the bottom-right panel shows two survival curves for which all parameters are different.

Measures other than the first moment could obviously also be used to summarize the survival curves of the uncured patients, e.g., the median or other quantiles of these proper survival functions, or other moments. 
In our opinion, however, the mean offers the easiest interpretation: 
how much (in absolute numbers) of the wholly available area of $\tau_0$ is below the survival curve?  
The more, the better.
This offers another means of comparing the usefulness of two (or more) treatments: which treatment prolongs the mean survival times most effectively, next to comparing the cure fractions?

\begin{figure}[ht]
    \centering
    \includegraphics[width=1\textwidth]{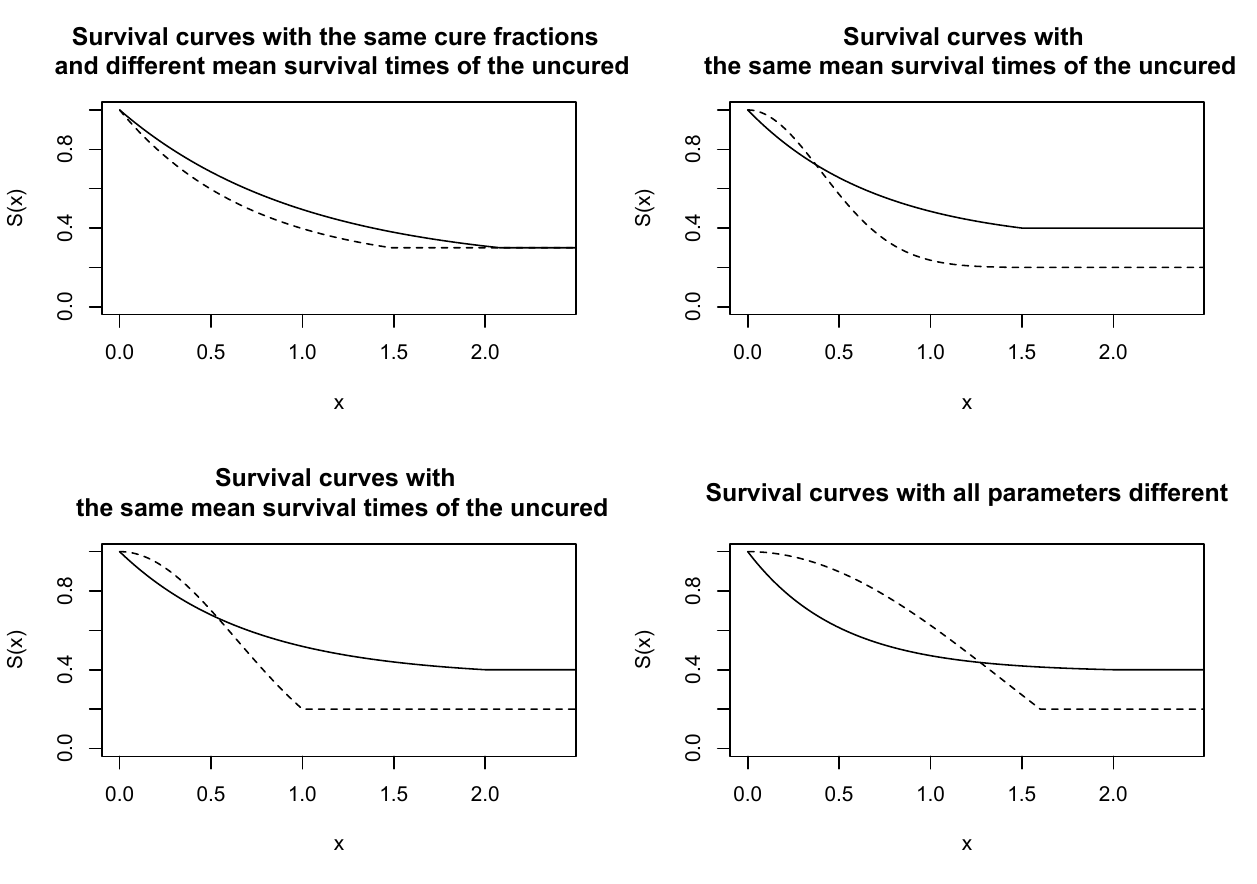}
    \caption{Different constellations of survival curves in two sample problems.}
    \label{fig:survival_curves}
\end{figure}

Let us now discuss the agenda for the present paper and the benefits of the proposed procedure:
\begin{itemize}
    \item We will propose a method for comparing the lifetimes of the uncured subject via mean survival times.
    \item This will allow for comparisons in two sample problems.
    \item Restrictions of time will not be necessary.
    \item Only weak assumption, e.g., for the sake of identifiability will be made.
    \item In particular, hazard rates are generally allowed to be discontinuous.
    \item Inference will be based on the random permutation method which gives rise to finitely exact hypothesis tests in the case of exchangeable samples and, otherwise, good small sample properties.
\end{itemize}
In this Part~I of our two companion papers, we will consider nonparametric models with a constant cure rate. 
In Part~II, we will extend the method to a semiparametric mixture cure model that allows for expressions of mean survival times conditionally on covariates. For the latter, we will assume a logistic-Cox model since it is the most widely used in practice.

This article is organized as follows.
The model and notation are introduced in Section~\ref{sec:model_I}.
In subsections therein, we also offer a toy example for a discussion about comparisons of survival in the two-sample problem and the formal definition of mean survival times of the uncured.
In Section~\ref{sec:stat_I}, we propose a Kaplan-Meier estimator-based estimation procedure, 
introduce the random permutation scheme, 
and present large sample properties which are crucial for inference.
Section~\ref{sec:simulations} contains a description of an extensive simulation study as well as the numerical results.
Data from a study on leukemia are illustrated and analyzed in the light of the proposed method in Section~\ref{sec:app_I}.
We conclude this Part~I with a discussion in Section~\ref{sec:disc_I} where we also offer an interim conclusion in preparation of Part~II.
All proofs are contained in Appendix~\ref{app}. The R code with an implementation of our methods is available in the GitHub repository \url{https://github.com/eni-musta/MST_uncured}.

	\section{Model and notation}
    \label{sec:model_I}
 
	We consider i.i.d.\ survival times $T_{11}, \dots, T_{1n_1}$ and  $T_{21}, \dots, T_{2n_2}$ from two independent groups $(i=1,2)$ that consist of a mixture of cured and uncured individuals, meaning that a fraction of the study population in each group (the cured ones) would not experience the event of interest.  
	We denote the event time of the cured individuals by $\infty$ and assume that, for the uncured ones, the event can happen on the interval $[0,\tau_{0,i}]$, $i=1,2$, respectively for each group. 
	We do not assume the cure threshold $\tau_{0,i}<\infty$ to be known in advance but depending on the application at hand one might have some information about it, for example in oncology based on the medical knowledge $\tau_{0,i}$ is expected to be somewhere between 5 or 10 years depending on the cancer type. 
	
	For the remainder of this paper, we denote by $(\Omega,\mathcal{A},\mathbb{P})$ the underlying probability space, $\mathbb{E}$ denotes expectation, $\stackrel d\to$ denotes convergence in distribution, and $\stackrel d=$ denotes equality in distribution.
	
	Let us denote by $F_i$ and $S_i=1-F_i$, $i=1,2$, the improper cumulative distribution function and the improper survival function of the time-to-event variables in the two treatment groups. 
 Let $F_{u,i}$ and $S_{u,i}=1-F_{u,i}$ be the proper cumulative distribution function and survival function for the uncured individuals in the two treatment groups. Let $$p_i=\p(T_{i1}>\tau_{0,i})=S_i(\tau_{0,i})\in(0,1)$$ denote the cure fractions in both groups. In particular, we have 
	\begin{equation}
		\label{eqn:distribution_survival}
		F_i(t)=(1-p_i)F_{u,i}(t),\qquad S_i(t)=p_i+(1-p_i)S_{u,i}(t)
	\end{equation}
	In the presence of right censoring, instead of the survival times, we observe the follow-up times $Y_{i1},\dots,Y_{in_i}$ and the censoring indicators $\Delta_{i1},\dots,\Delta_{in_i}$, where $Y_{ij}=\min\{T_{ij},C_{ij}\}$, $\Delta_{ij}=\1_{\{T_{ij}\leq C_{ij}\}}$ and $C_{ij}$ are the censoring times. We assume that censoring is independent of the survival times and has bounded support $[0,\tau_i]$ in each group. 
	In particular, because of the finite censoring times, all the cured individuals will be observed as censored. 
    In order to be able to identify the cure fraction, we need $\tau_{0,i}\leq\tau_i$  and $F_i$ continuous at $\tau_{i}$ in case $\tau_{0,i}=\tau_{i}$, which is known as the sufficient follow-up assumption \cite{maller1996survival}. 
    The idea is that, since the cure status is not observed and $F_{u,i}(\cdot)$ is left unspecified, if $\tau_i<\tau_{0,i}$ then the events $\{T_{i1}\in(\tau_i,\tau_{0,i}]\}$ and $\{T_{i1}=\infty\}$ would be indistinguishable.
    As a result, the cure rate could not be identified. A statistical test for this assumption is proposed in \cite{maller1996survival} but its practical behavior is not very satisfactory given the unstable behaviour of the Kaplan Meier estimator in the tail region. In practice, a long plateau of the Kaplan-Meier estimator, containing many censored observations, is considered to be an indication of sufficient follow-up.
		
\subsection*{Comparison of overall survival}
	We first illustrate why comparing overall survival functions is not appropriate in the presence of a cure fraction. The difference in overall survival combines together the difference in cure fractions and in the survival times of the uncured in a way that it is difficult to interpret. For example, if group one has a higher cure fraction but lower survival times for the uncured, the overall survival functions might cross and the difference between them would be a weighted combination of the two effects $$S_1(t)-S_2(t)=(p_1-p_2)\{1-S_{u,2}(t)\}+(1-p_1)\{S_{u,1}(t)-S_{u,2}(t)\}.$$ On the other hand, if the two groups have the same cure fraction $p$, then
	\[
	S_1(t)-S_2(t)=(1-p)\{S_{u,1}(t)-S_{u,2}(t)\}
	\]
	This means that, particularly for a large cure fraction, the observed difference in overall survival functions is much smaller than the actual difference of the survival functions for the uncured. 
	
	Using a one number summary of the difference in overall survival is even more problematic in the presence of a cure fraction. First, the proportional hazard assumption is clearly violated on the level of the whole population and, as a result, the hazard ratio cannot be used. For the Mann-Whitney effect, what counts are the chances of having  longer  survival times for one group compared to the other, but the actual difference between these times does not matter. So one cannot distinguish between having a larger cure fraction or just slightly longer survival. 
 
 Consider for example the following hypothetical scenario:  patients receiving treatment A have a 20\% chance of being cured while with treatment B there is no cure chance; a random person receiving treatment B lives several months longer compared to an uncured patient who received treatment A.
    Let $T_1$ and $T_2$ represent the random lifetimes of patients receiving treatments A and B, respectively.
    According to the above description, the Mann-Whitney effect would then be 
	\[
	\p(T_1>T_2)+\frac12\p(T_1=T_2)=0.2<0.5,
	\] 
	leading to the conclusion that treatment B should be preferred. This is counter-intuitive because, given the small difference in survival times of the uncured, in practice one would probably prefer the treatment that offers some chance of getting cured. On the other hand, if one would use the restricted mean survival time  as a summary measure, the actual survival times matter. However, because of the restriction  to a specified point $\tau$ (duration of the study), there would still be no distinction between the cured individuals and those who survive more than $\tau$
	\[
	RMST^{(i)}_\tau=\E[\min(T_i,\tau)]=(1-p_i)\E[\min(T_i,\tau)| T_i < \tau]+\tau p_i, \quad i=1,2.
	\] To illustrate this, consider the following example: patients receiving treatment A have 20\% chance of being cured, while with treatment B there is no cure chance; a random person receiving treatment B lives on average 60 months, while an uncured patient who received treatment A lives on average 24 months. Let us assume that $\tau_{0,1}=\tau_{0,2}=120$ months. If the duration of the study was also 120 months (sufficient follow-up), we would obtain
	\[
		RMST^{(1)}_{120}-RMST^{(2)}_{120}=0.8\cdot \E[T_1|T_1 < 120]+120\cdot0.2-\E[T_2]=-16.8,
	\]
	 leading to the conclusion that treatment B should be preferred. However, if the study had continued for longer, e.g., 240 months, we would obtain
	 	\[
	 RMST^{(1)}_{240}-RMST^{(2)}_{240}=0.8\cdot\E[T_1|T_1<240]+240\cdot0.2-\E[T_2]=7.2,
	 \]
	suggesting that treatment A is better, which contradicts the previous conclusion.
	
	For these reasons, we think that in the presence of a cure fraction, it is more informative to compare separately the cure fractions and the survival functions of the uncured. In practice, one can then make a personalized decision by choosing to put more weight to one component compared to the other, based on the life expectancy if uncured and the risks one is willing to take. For example, for children there is an essential difference between cure and 10 year survival, while such difference might be less significant for elderly patients. 
	
	\subsection*{Mean survival time of the uncured}
The problem of comparing cure fractions has already been considered in the literature. Here, we focus on comparing the survival times of the uncured.
In particular, we propose the mean survival time as a summary measure.

	We are interested in the difference of mean survival times of the uncured individuals among the two groups 
	\[
	MST_{u,1}-MST_{u,2}=\E[T_{11} \mid T_{11} < \infty]-\E[T_{21} \mid T_{21} < \infty].
	\] 
	In combination with the cure fractions, such mean survival times provide useful summaries of the improper survival curves.
	Using the relations in \eqref{eqn:distribution_survival}, we obtain the following expression for the mean survival times
	\[
	MST_{u,i}= \int_0^{\tau_{0,i}}S_{u,i}(u)\,\dd u=\int_0^{\tau_{0,i}}\frac{{S}_{i}(u)-{p}_i}{1-{p}_i}\,\dd u, \qquad i=1,2.
	\]

	\section{Estimation, asymptotics, and random permutation}
 \label{sec:stat_I}

\subsection*{Estimators and their large sample properties}
 
	Estimation of the cure rate and of the nonparametric survival function in this setting has been considered in \cite{maller1992,maller1996survival} and is based on the Kaplan-Meier (KM) estimator. In particular, we can estimate $S_i$ and $p_i$ by 
	\[
	\hat{S}_i(t)=\prod_{u\in(0,t]}\left(1-\frac{\dd N_i(u)}{R_i(u)}\right)\quad \text{and}\quad \hat{p}_i=\hat{S}_i(Y_{i,(m_i)}),
	\]
	respectively, where $ Y_{i,(m_i)}$ is the largest observed event time in group $i$, $N_i(u)=\sum_{j=1}^{n_i}\1_{\{Y_{ij}\leq u,\Delta_{ij}=1\}}$ counts the number of observed events up to time $u$ and $R_i(u)=\sum_{j=1}^{n_i}\1_{\{Y_{ij}\geq u\}}$ counts the numbers of individuals at risk at time $u$. 
	By a plug-in method, we estimate the mean survival time of the uncured by 
	\[
	\widehat{MST}_{u,i}=\int_0^{Y_{i,(m_i)}}\frac{\hat{S}_{i}(u)-\hat{p}_i}{1-\hat{p}_i}\,\dd u
	\]

	Let $\hat{m}=	\widehat{MST}_{u,1}-	\widehat{MST}_{u,2}$ be an estimator of $m={MST}_{u,1}-{MST}_{u,2}$. 
	Using the asymptotic properties of the KM estimator,
	we obtain the following result for the mean survival times, for which we define 
	\begin{equation}
		\label{eqn:v_i}
	v_i(t)=\int_0^t\frac{\dd F_i(u)}{S_i(u)\{1-H_i(u-)\}},
	\end{equation}
$H_i(t)=\p(Y_{i1}\leq t)=\{1-F_i(t)\}\{1-G_i(t)\}$, and $G_i(t)=\p(C_{i1}\leq t)$ denotes the distribution of the censoring times.
	\begin{thm}
		\label{theo:MST_homog}
	 For $i=1,2$, assume $p_i\in(0,1)$  
	and that one of the following conditions holds:
	\begin{itemize}
		\item[a)] $\tau_{0,i}<\tau_i$
  \item[b)] $\tau_{0,i}=\tau_i$, $F_i$ is continuous and 
		\begin{equation}
	\label{eqn:cond_integral_KM}
		\int_0^{\tau_i}\frac{\dd F_i(t)}{1-G_i(t-)}<\infty;
		\end{equation}
    		\item[c)] $\tau_{0,i}=\tau_i$, $F_i$ is continuous at ${\tau_{i}}$, $\lim_{t\uparrow{\tau_{i}}}\{F_i({\tau_{i}})-F_i(t)\}^2v_i(t)=0$, and
		\[
		\lim_{t\uparrow{\tau_{i}}}\int_t^{{\tau_{i}}}\frac{\1_{\{0\leq G_i(u-)<1\}}S_i(u)}{\{1-G_i(u-)\}S_i(u-)}\,\dd F_i(u)=0.
		\]
	\end{itemize}
	 Then the variable $\sqrt{n_i}(\widehat{MST}_{u,i}-{MST}_{u,i})$ is asymptotically normally distributed, i.e. $$\sqrt{n_i}(\widehat{MST}_{u,i}-{MST}_{u,i})\xrightarrow{d}N\sim\mathcal{N}(0,\sigma_i^2)$$ as $n\to\infty$. The limit variance $\sigma_i^2$ is defined in \eqref{def:sigma-i} in the appendix. 
	\end{thm}

The extra technical conditions in b) and c) of the previous theorem are the conditions needed to obtain weak convergence of the normalized Kaplan-Meier estimator to a Gaussian process \cite{gill1980censoring,ying1989note}. Note that in the particular case $\tau_{0,i} < \tau_i$,  the conditions unrelated to the continuity of $F_i$ are automatically satisfied, leading to no extra requirements for case~a).
Continuity of $F_i$ is nowhere needed in case~a).

	From the independence assumption between the two groups and Theorem~\ref{theo:MST_homog},
	we obtain the following result for which we define $a_n=\sqrt{n_1n_2/(n_1+n_2)}$.
	\begin{cor}
		\label{cor:MST_homog}
		Assume that $n_1/(n_1+n_2)\to\kappa\in(0,1)$ {as $\min(n_1, n_2) \to \infty$}. Under any of the two conditions in Theorem 1, we have that $a_n(\hat{m}-m)$ is asymptotically normally distributed with mean zero and variance $$\sigma^2=(1-\kappa)\sigma_1^2+\kappa\sigma_2^2.$$
	\end{cor}
	The canonical plug-in estimator of $\sigma^2$, say
 \begin{equation}
	    \label{eqn:sigma_hat}
     \hat{\sigma}^2 = \frac{n_2}{n_1 + n_2} \hat \sigma_1^2 + \frac{n_1}{n_1+n_2} \hat \sigma_2^2
	\end{equation}  is obviously consistent.
	Hence, the combination of Theorem~\ref{theo:MST_homog} and Corollary~\ref{cor:MST_homog} could be used to justify inference methods for ${MST}_{u,1} - {MST}_{u,2}$ based on the asymptotic normal approximation.
	However, such inference procedures can usually be made more reliable by means of resampling methods.

\subsection*{Inference via random permutation across samples}
 
	We propose random permutation to construct inference methods for $m$.
	To introduce this procedure, let $\pi=(\pi_1, \dots, \pi_{n_1+n_2})$ be any permutation of $(1,2, \dots, n_1+n_2)$.
	When applied to the pooled sample, say, $(Y_{1}, \Delta_{1}), \dots, (Y_{{n_1+n_2}}, \Delta_{{n_1+n_2}})$, this permutation leads to the permuted samples $(Y_{\pi_1}, \Delta_{\pi_1}), \dots, (Y_{\pi_{n_1}}, \Delta_{\pi_{n_1}})$ and $(Y_{\pi_{n_1+1}}, \Delta_{\pi_{n_1+1}}), \dots, (Y_{\pi_{n_1+n_2}}, \Delta_{\pi_{n_1+n_2}})$.
	
	In the special case of exchangeability, i.e.\ $(Y_{1j}, \Delta_{1j}) \stackrel{d}{=} (Y_{2j}, \Delta_{2j})$, 
	$\hat m = \widehat{MST}_{u,1} - \widehat{MST}_{u,2}$ would have the same distribution as $\hat m^{\pi} = \widehat{MST}_{u,1}^\pi - \widehat{MST}_{u,2}^\pi$. Here, $\widehat{MST}_{u,i}^\pi$ are the estimators of the mean survival times, just based on the $i$-th permuted sample.
	So, under a sharp null hypothesis of exchangeability, a test for the equality of mean survival times would reject the null hypothesis if $\hat m$ belongs to the $\alpha \times 100\%$ most extreme values of $\hat m^\pi$ across all $(n_1+n_2)!$ permutations.
	
	However, under the weak null hypothesis of equal mean survival times, $H: m=0$, the samples are in general not exchangeable.
	As a consequence, the asymptotic variances of $\hat m$ and $\hat m^{\pi}$ cannot be assumed equal and hence they must be studentized.
	We will thus focus on $\hat m^\pi/\hat{\sigma}^{\pi}$ as the permutation version of $\hat m/\hat \sigma$, where
	\begin{equation}
	    \label{eqn:hat_sigma_pi}
     \hat{\sigma}^{\pi2} = \frac{n_2}{n_1 + n_2} \hat \sigma_1^{\pi2} + \frac{n_1}{n_1+n_2} \hat \sigma_2^{\pi2},
	\end{equation} and $\sigma_i^{\pi2}$ is the plug-in variance estimator based on the $i$-th permuted sample.
	Consequently, our aim is to compare $\hat m/\hat \sigma$ to the conditional distribution of $\hat m^\pi/\hat{\sigma}^{\pi}$ given the data to reach a test conclusion.
	
	Because it is computationally infeasible to realize $\hat m^\pi/\hat{\sigma}^{\pi}$ for all $(n_1+n_2)!$ permutations, we will realize a relatively large number $B$ of random permutations $\pi$ and approximate the conditional distribution by the collection of the realized $\hat m^\pi/\hat{\sigma}^{\pi}$ of size $B$.
	
	In the following, we will discuss the asymptotic behaviour of $\hat m^\pi/\hat{\sigma}^{\pi}$ to justify the validity of the resulting inference procedures.
	From now on, we understand the weak convergence of conditional distributions (in probability) as the convergence of these distributions to another with respect to e.g.\ the bounded Lipschitz metric (in probability); see e.g.\ Theorem~1.12.4 in \cite{VW96}.
	\begin{thm}
		\label{theo:MST_homog_perm}    
		Assume that $\tau_{0,i}<\tau_i$ and $p_i\in(0,1)$ 
  for $i=1,2$. Then, as $n_1, n_2\to\infty$ with $n_1/(n_1+n_2) \to \kappa \in (0,1)$, the conditional distribution of 
  $a_n\hat m^\pi$
  given the data converges weakly in probability to the zero-mean normal distribution with variance given in \eqref{def:sigma} in the appendix.
	\end{thm}
Under continuity assumptions similar to those in Theorem~\ref{theo:MST_homog} and additional assumptions on the censoring distributions, we conjecture that a similar weak convergence result holds for the case of $\tau_{0,i}=\tau_i$.
This could potentially be shown by extending the results of \cite{dobler19} to the random permutation method instead of the classical bootstrap.
Deriving such results, however, is beyond the scope of the present paper.
 
	The structure of asymptotic variance in the previous theorem motivates a canonical permutation-type variance estimator $\hat \sigma^{\pi2}$, that is, the plug-in estimator based on the pooled sample.
	Due to the obvious consistency of this estimator, we arrive at the following main result on the permuted studentized mean survival time:
	\begin{cor}
		\label{cor:MST_homog_perm}
		Assume that $\tau_{0,i}<\tau_i$ and $p_i\in(0,1)$  and $F_i$ is continuous for $i=1,2$. Then, as $n_1, n_2\to\infty$ with $0 < \liminf n_1/(n_1+n_2) \leq \limsup n_1/(n_1+n_2) < 1$, the (conditional) distributions of 
  $a_n \hat m^\pi/\hat \sigma^\pi$
  (given the data) and $a_n  (\hat m - m)/\hat \sigma$ converge weakly (in probability) to the same limit distribution which is standard normal.
	\end{cor}
	
	{We conclude this section with a remark on the inference procedures deduced from the random permutation approach.
    \begin{re}
    Corollary~\ref{cor:MST_homog_perm} gives rise to asymptotically exact 1- and 2-sided tests for the null hypotheses $H_0^{(1)}: m \leq 0$, $H_0^{(2)}: m \geq 0$, or $H_0^{(3)}: m\neq 0$ against the respective complementary alternative hypotheses: comparing $\hat m/\hat \sigma$ with data-dependent critical value(s) obtained from the collection of realized $\hat m^\pi/\hat \sigma^\pi$ (for fixed data) allows for controlling the chosen significance level $\alpha \in (0,1)$ as $n_1, n_2 \to \infty$ under the assumptions made above.
    A similar remark holds true for more general null hypotheses in which $m$ is compared to some hypothetical value $m_0 \in \R$.
    In addition, as a well-known property of permutation tests based on studentized test statistics, the just-mentioned tests are exact in the special case of exchangeability between both sample groups for which $m_0=0$ is automatically fulfilled.
    Similarly, by inverting hypothesis tests into confidence intervals, asymptotically exact confidence intervals for $m$ can be constructed; see the subsequent section for details.
    \end{re}
    }

\section{Simulation study}
\label{sec:simulations}

In this section, we study the finite sample performance of both the asymptotic and the permutation approach when constructing confidence intervals for $m$ and testing the one sided hypothesis $H_0:m\leq 0$ versus $H_1: m>0.$ In order
to cover a wide range of scenarios, we consider nine settings as described below. Settings 1-4 correspond to having two samples with the same mean survival time for the uncured ($m=0$) but possibly different distributions, while settings 5-9 correspond to having two samples with different mean survival times for the uncured ($m\neq 0$) with different magnitudes and signs for $m$. The cure and censoring rates also vary across the settings. In all of the following settings, the survival times for the uncured are truncated at $\tau_{0,i}$ equal to the $99\%$ quantile of their distribution in order to satisfy the assumption of compact support.
The censoring times are generated independently from an exponential distribution with parameter $\lambda_{C,i}$ and are truncated at  $\tau_{i}=\tau_{0,i}+2$. The truncation of the censoring times is done only to reflect the bounded follow-up but $\tau_{i}$ does not play any role  apart from the fact that $\tau_i>\tau_{0,i}.$ Note also that the reported censoring rate includes the cured subjects, which are always observed as censored, hence it is larger than the cure rate. 

\subsection*{Simulation settings}

\textit{Setting 1.} The two samples are exchangeable ($m=0$): cure rate $40\%$,  Weibull distribution for the uncured with shape and scale parameters $0.75$ and $1.5$ respectively,  censoring rate around $50\%$ ($\lambda_{C,i}=0.3$, $i=1,2$). 
\\[0.2cm]
\textit{Setting 2.} The uncured have the same Weibull distribution in both samples ($m=0$) with shape and scale parameters 0.75 and 1.5 respectively. The cure rate is $20\%$ in sample 1 and $60\%$ in sample 2, the censoring rate is around $30\%$ and $70\%$ in sample 1 and 2 respectively ($\lambda_{C,1}=0.25$, $\lambda_{C,2}=0.5$). 
\\[0.2cm]
\textit{Setting 3.} The uncured have the same Weibull distribution in both samples ($m=0$) with shape and scale parameters 0.75 and 1.5 respectively. The cure rate is samples 1 and 2 is $60\%$ and $20\%$ respectively; censoring rate is around $65\%$ and $25\%$ respectively ($\lambda_{C,1}=0.3$, $\lambda_{C,2}=0.1$).
\\[0.2cm]
\textit{Setting 4.} The uncured in sample 1 follow a  Weibull distribution with shape and scale parameters 0.75 and 1 respectively, while the uncured in sample 2 follow a Gompertz distribution with scale parameter 1 and shape parameter 0.327. The parameters are chosen such that the two groups have the same mean survival time (so $m=0$). The cure rate is $40\%$ in both samples; censoring rate is around $50\%$ in both samples ($\lambda_{C,1}=0.2$, $\lambda_{C,2}=0.15$).
\\[0.2cm]
\textit{Setting 5.} The uncured in the two samples have different Gompertz distributions with the same scale parameter 1 and shape parameters 0.1 and 0.5 respectively. The  difference of mean survival times is $m=1.09$. For this choice of parameters, the supports of the event times in the two samples are $[0,3.8]$ and $[0,2.3]$ respectively. The cure rate is $40\%$ in both samples; censoring rate is around $65\%$ in sample 1 and $45\%$ in sample 2 ($\lambda_{C,1}=0.3$, $\lambda_{C,2}=0.1$).
\\[0.2cm]
\textit{Setting 6.} This is the same as Setting 5 but the two groups are exchanged, i.e. $m=-1.09$. 
\\[0.2cm]
\textit{Setting 7.} The uncured in the two samples have different Gompertz distributions with the same scale parameter 1 and shape parameters 0.08 and 0.1 respectively. The  difference of mean survival times is $m=0.18$ For this choice of parameters, the supports of the event times in the two samples are $[0,4.1]$ and $[0,3.9]$ respectively.  The cure rate is $60\%$ in sample 1 and  $20\%$ in sample 2; censoring rate is around $70\%$ in sample 1 and $40\%$ in sample 2 ($\lambda_{C,1}=0.2$, $\lambda_{C,2}=0.15$).
\\[0.2cm]
\textit{Setting 8.} The survival distributions of the uncured are as in Setting 7, i.e. $m=0.18$.  The cure rate is samples 1 and 2 is $30\%$ and $20\%$ respectively; censoring rate is around $44\%$ in sample 1 and $34\%$ in sample 2 ($\lambda_{C,1}=0.1$, $\lambda_{C,2}=0.1$).
\\[0.2cm]
\textit{Setting 9.} The event times of the uncured in the sample 1 follow a Gompertz distribution with scale parameter 1 and shape parameter 0.08, while in sample 2 they follow a Weibull distribution with shape ans scale parameters 2 and 0.28 respectively. Both  distributions have support $[0,4.1]$ but different mean survival times ($m=0.52$).  The cure rate is $40\%$ in both samples; censoring rate is around $50\%$ in both samples ($\lambda_{C,1}=0.1$, $\lambda_{C,2}=0.1$).

\subsection*{Simulation results}

First, considering different sample sizes $n_1=2n_2\in\{50,200\}$ or $n_1=n_2\in\{100,200\}$, $95\%$ confidence intervals for $m$ are constructed based on both the asymptotic approximation and the permutation approach:
\[
I_n=[\hat{m}\mp q_{1-\alpha/2}\hat\sigma/a_n],\qquad I^\pi_n=[\hat{m}- q^\pi_{1-\alpha/2}\hat\sigma,\hat{m}- q^\pi_{\alpha/2}\hat\sigma],
\]
where $\hat\sigma$ is given in \eqref{eqn:sigma_hat}, $q_{1-\alpha}$ denotes the $100(1-\alpha)\%$-quantile of the standard normal distribution, $\alpha=0.05$ and $q^\pi_{1-\alpha}$ denotes the $100(1-\alpha)\%$-quantile of the conditional distribution of $\hat{m}^\pi/\hat\sigma^\pi$ for the permutation approach. 
We take $B=500$ random permutations, which seemed to be sufficient since increasing $B$ to 1000 did not have much effect in the results. Average length and coverage probabilities over 1,000 repetitions are reported in Table~\ref{tab:1}. 
The coverage rates closest to 95\% among both types of confidence intervals is printed in bold-type.
\begin{table}[h]
	\centering
	\begin{tabular}{c|c|cc|cc|cc|cc}
		Sett.& &\multicolumn{2}{c|}{$n_1=2n_2=50$}& \multicolumn{2}{c|}{$n_1=n_2=100$}& \multicolumn{2}{c|}{$n_1=2n_2=200$} &\multicolumn{2}{c}{$n_1=n_2=200$}\\
		& & M1 & M2 & M1 & M2 & M1 & M2 & M1 & M2\\
		\hline
		\hline
		1 & L & 1.34 & 1.15& 0.83&0.88 & 0.74&0.80 &0.62 &0.64\\
		& CP & 88.4  & \textbf{93.1}& 93.0 & \textbf{94.3} &93.3 & \textbf{95.4}& 94.1&\textbf{94.2}\\
		\hline
		2 & L & 1.04 & 1.22&0.91 &1.00 & 0.84& 0.92& 0.73&0.78\\
		& CP & 76.0 & \textbf{87.2}&86.6 & \textbf{90.0}&83.5 &\textbf{87.6} &89.4 &\textbf{90.6}\\
		\hline
		3 & L & 1.18 & 1.37& 0.82&0.87 &0.66 & 0.71&0.60 &0.62\\
		& CP & 89.8 & \textbf{93.5} &88.8 & \textbf{90.1} &\textbf{93.3}& 92.9&92.0 &\textbf{92.2}\\
		\hline
		4 & L & 1.44 &1.64 & 1.08& 1.12& 0.87& 0.90&0.83&0.84\\
		& CP &86.3 &\textbf{87.2} & 84.2& \textbf{86.5} & \textbf{89.5} &88.9 &89.2 & \textbf{89.3}\\
		\hline
		5 & L & 1.02 &1.22 &0.67 & 0.69&0.53&0.54 &0.48 & 0.48\\
		& CP &93.1  & \textbf{95.9} &93.6 & \textbf{94.6} & \textbf{95.0}& 94.8& \textbf{95.2} &95.4 \\
		\hline
		6 & L &  1.18& 1.34& 0.67& 0.69& 0.64& 0.66&0.48 &0.48\\
		& CP &86.4 &\textbf{89.2} &93.6 &\textbf{94.7} &\textbf{93.3} & \textbf{93.3}&\textbf{95.0} &94.7\\
		\hline
  7 & L & 1.31 & 1.52&0.83 & 0.85&0.67 & 0.69&0.60&0.60\\
		& CP & 92.3 &\textbf{95.7} & 94.3&\textbf{94.5} &\textbf{93.3}&92.9 & \textbf{93.7}&93.6\\
		\hline
		8 & L &  1.07& 1.15& 0.64& 0.65& 0.55& 0.55&0.45 &0.45\\
		& CP &94.1 &\textbf{95.6} &\textbf{93.3} &\textbf{93.3} &\textbf{94.8}& 94.7&\textbf{93.9} &93.5\\
		\hline
		9 & L &  1.18& 1.30& 0.71& 0.72& 0.61& 0.61&0.50 &0.50\\
		& CP &92.5 &\textbf{94.1} &94.2 &\textbf{94.8}&94.2& \textbf{94.6}&94.4 &\textbf{94.6}\\
		\hline
	\end{tabular}
	\caption{\label{tab:1}Coverage probabilities (CP) in $\%$  and average length (L) of 95\% confidence intervals using the asymptotic approach (M1) and the permutation approach (M2) for different sample sizes.}
\end{table}

We observe that the confidence intervals based on the permutation approach are in general slightly wider and have better coverage, particularly for small sample sizes. As the sample sizes increase, the two approaches give more comparable results. For some settings, much larger sample sizes are needed to have coverage close to the nominal level but, for most of them, coverage is close to 95\%. 
When the sample sizes are the same, settings 2 and 3 are almost the same, with setting~3 having less censoring, leading to shorter confidence intervals and better coverage. When $n_1=2n_2$, setting 2 is more difficult because the smaller sample has a very large cure and censoring rate, leading to worse coverage probabilities. Similarly, when the sample sizes are the same, setting 5 and 6 are the same, leading to same length confidence intervals and approximately same coverage (due to sampling variation). When $n_1=2n_2$, setting 6 is more difficult because has higher censoring rate in the smaller sample.
As a result, we observe longer confidence intervals and worse coverage probabilities.  Setting 8 is similar to setting 7 but the first sample has lower cure rate, leading to shorter confidence intervals. 
When the two samples are not comparable in terms of cure and censoring rate, increasing the sample size of the sample in which it is easier to estimate ${MST}_{u,i}$, does not usually lead to better coverage (compare settings 2 and 3, 5 and 6). On the other hand, increasing the sample size of the sample in which estimation of ${MST}_{u,i}$ is more difficult usually leads to better coverage.

We further investigate settings 2, 3, 4 which exhibit the worst performance in terms of coverage probabilities. In setting 2, as estimation in the second sample is more difficult (because of higher cure and censoring rates), the performance of both the asymptotic and permutation approach in worse when the size of sample 1 is larger than the size of sample 2.  
In setting 3, estimation of the first sample is more challenging and we observed that the coverage is better when the size of sample 1 is larger. In setting 4, both samples have the same censoring and cure rate but the coverage seems to be worse when the sample sizes are the same. Results for larger sample sizes under the most difficult scenarios for each of these three settings are reported in Table~\ref{tab:5}. They show that, as expected, the coverage probabilities for both approaches converge to the nominal level.
\begin{table}[h]
	\centering
	\begin{tabular}{c|c|cc|cc|cc|cc}
		Sett.& &\multicolumn{2}{c|}{$n_1=2n_2=600$}& \multicolumn{2}{c|}{$n_1=2n_2=1200$}& \multicolumn{2}{c|}{$n_1=2n_2=4000$} &\multicolumn{2}{c}{$n_1=2n_2=10000$}\\
		& & M1 & M2 & M1 & M2 & M1 & M2 & M1 & M2\\
		\hline
		\hline
		2 & L & 0.57 & 0.60& 0.44&0.45 & 0.26&0.26 &0.12 &0.12\\
		& CP & 84.4 & \textbf{87.6}&88.4 & \textbf{89.9}&\textbf{92.1} & 92.0& \textbf{94.5}&94.2\\
	\cline{1-10}
 \multicolumn{10}{c}{}\\
 \multicolumn{2}{c|}{}&\multicolumn{2}{c|}{$n_1=n_2=500$}& \multicolumn{2}{c|}{$n_1=n_2=1000$}& \multicolumn{2}{c|}{$n_1=n_2=2000$}&\multicolumn{2}{c}{}\\
	\multicolumn{2}{c|}{}& M1 & M2 & M1 & M2 & M1 & M2 &\multicolumn{2}{c}{} \\
  \hline
  \hline
  	3 & L & 0.40& 0.40& 0.28&0.28 & 0.20& 0.20&&\\
		& CP & \textbf{93.8} & 93.6&\textbf{92.5}& 91.9&\textbf{94.8}&94.4&&\\
  \hline
		4 & L & 0.55 & 0.56& 0.39& 0.40 &0.28&0.28&&\\
		& CP & \textbf{92.3} & \textbf{92.3} &\textbf{94.2} &  \textbf{94.2}&94.1&\textbf{94.5}&& \\
		\hline
	\end{tabular}
	\caption{\label{tab:5}Coverage probabilities (CP) in $\%$ and average length (L) of 95\% confidence intervals for $m$ using the asymptotic approach (M1) and the permutation approach (M2) for different sample sizes.}
\end{table}

In addition, we selected 2 of the settings (setting 2 and 9) and further investigated the effect of the censoring and cure rates for sample sizes $100-100$ and $200-100$. First, we keep the cure rate fixed at $40\%$ (moderate) and consider 3 censoring levels: $45\%$ (low), $50\%$ (moderate) and $60\%$ (high). Secondly, we vary the cure rate: $20\%$ (low), $40\%$ (moderate) and $60\%$ (high), while maintaining the same moderate censoring level equal to the cure rate plus $10\%$. Average length and coverage probabilities over 1,000 repetitions are reported in Table~\ref{tab:3}. 
\begin{table}[h]
	\centering
	\begin{tabular}{c|c|c|cc|cc|cc}
		\multicolumn{3}{c}{}&\multicolumn{6}{c}{censoring rate}\\
		\multicolumn{3}{c|}{}&\multicolumn{2}{c|}{low}&\multicolumn{2}{c|}{moderate}&\multicolumn{2}{c}{high}\\
		Sett. &$n_1/n_2$&& M1 & M2 & M1 & M2 & M1 & M2 \\
		\hline
		\hline
		4& 100/100&L& 0.87&0.89&1.08&1.12&1.12&1.23\\
		&&CP& 92.4&\textbf{92.6}&84.0&\textbf{86.5}&68.7&\textbf{72.2}\\
		&200/100&L& 0.67&0.71&0.87&0.90&1.06&1.17\\
		&&CP& \textbf{93.8}&92.9&\textbf{89.5}&88.9&76.8&\textbf{77.3}\\
		\hline
		9&100/100&L& 0.67&0.68&0.71&0.72&0.79&0.82\\
		&&CP&\textbf{94.1} &93.9&94.2&\textbf{94.8}&92.9&\textbf{94.3}\\
		&200/100&L&0.58&0.58&0.61&0.61&0.69&0.71\\
		&&CP&\textbf{94.5}&94.2&94.2&\textbf{94.6}&\textbf{94.7}&94.6\\
		\hline
		\hline
		\multicolumn{3}{c|}{}&\multicolumn{6}{c}{cure rate}\\
		\multicolumn{3}{c|}{}&\multicolumn{2}{c|}{low}&\multicolumn{2}{c|}{moderate}&\multicolumn{2}{c}{high}\\
		Sett. &$n_1/n_2$&& M1 & M2 & M1 & M2 & M1 & M2 \\
		\hline\hline
		4& 100/100&L&0.85 &0.92&1.08&1.12&1.21&1.21\\
		&&CP& 90.6&\textbf{91.4}&84.0&\textbf{86.5}&88.1&\textbf{89.0}\\
		&200/100&L&0.66 &0.72&0.87&0.90&1.20&1.31\\
		&&CP& \textbf{93.0}&\textbf{93.0}&\textbf{89.5}&88.9&80.6&\textbf{80.7}\\
		\hline
		9	&100/100&L& 0.59&0.60&0.71&0.72&0.92&0.96\\
		&&CP& \textbf{94.3}&\textbf{94.3}&94.2&\textbf{94.8}&94.4&\textbf{95.5}\\
		&200/100&L&0.51 &0.51&0.61&0.61&0.80&0.82\\
		&&CP& 95.8&\textbf{95.6}&94.2&\textbf{94.6}&93.2&\textbf{93.8}\\
		\hline
	\end{tabular}
	\caption{\label{tab:3}Coverage probabilities (CP) in $\%$  and average length (L) of 95\% confidence intervals using the asymptotic approach (M1) and the permutation approach (M2) for different sample sizes, censoring and cure rates.}
\end{table}	
As expected, we observe that, as the censoring or cure rate increases, the length of the confidence intervals increases. In setting 4 the coverage deteriorates significantly for a high censoring rate, while in setting 9 the coverage remains stable and close to the nominal value throughout all scenarios. 

Next, we consider a one-sided hypothesis test for $H_0: m\leq 0$ versus $H_1: m> 0$ at level $5\%$.  The rejection rates for the test  are reported in Table~\ref{tab:2}. Looking at settings 1-4 and 6  for which $H_0$ is true (with $m=0$ for settings 1-4), we observe that most of the time the rejection rate is lower or close to $5\%$ for both methods. Setting 2 is again the most problematic one with rejection rate higher than the level of the test.  This might be because in setting 2 the second sample has a high cure (and censoring) rate, which might lead to underestimation of the mean survival times for the uncured in sample 2 and as a result an overestimation of $m$. For settings 2, 3, and 4 we also considered larger sample sizes.
The results are reported in Table~\ref{tab:6}. In particular, we observe that the rejection rates in setting 2 decrease and approaches the significance level as the sample size increases. 
In terms of power, as $m $ or the sample size increase, the power increases. Both methods are comparable, with the permutation approach usually leading to slightly lower rejection rate under both hypothesis. 
Again,  in settings 4 and 9 we also investigate the effect of the censoring and cure rate as above.  Results are given in Table~\ref{tab:4}.
 As the censoring or cure rate increases, the power of the test decreases, while the rejection rate in setting 4 ($H_0$ is true) remains below the $5\%$ level throughout all scenarios.
 
 {Finally, to acknowledge that the case of unbalanced sample sizes where the smaller sample meets the higher censoring rate, we would like to point to \cite{ditzhaus2023}.
 For very small sample sizes, their studentized permutation test about the RMST also exhibited the worst control of the type-I error rate in this challenging context; see Table~1 therein, and also Tables~S.1 and~S.2 in the supplementary material accompanying that paper.
 Of note, their proposed permutation test is still quite accurate with a size not exceeding 6.8\% even in the most challenging setting.}

\begin{table}[h]
	\centering
	\begin{tabular}{c|cc|cc|cc|cc}
		Sett. &\multicolumn{2}{c|}{$n_1=2n_2=50$}& \multicolumn{2}{c|}{$n_1=n_2=100$}& \multicolumn{2}{c|}{$n_1=2n_2=200$} &\multicolumn{2}{c}{$n_1=n_2=200$}\\
		&  M1 & M2 & M1 & M2 & M1 & M2 & M1 & M2\\
		\hline
		\hline
		1 &  5.6 &5.0 &6.9 & 4.9&5.9& 5.4&6.2 &5.5\\
		2  & 20.9  & 19.0& 17.5& 14.0&18.0 &17.2&14.3 &12.1\\
		3  &0.9& 0.8& 2.3& 1.7&0.6& 0.5& 2.4&2.2\\
		4 & 0.6&0.2 &2.3& 0.6&0.3& 0.2&1.9 &0.7\\
		5  &  96.5&94.1 &100 &100 &100 & 100& 100&100\\
		6  & 0.0&0.0 &0.0&0.0& 0.0& 0.0& 0.0&0.0\\
  7 & 11.7 & 9.5& 21.3&21.0&23.0 &22.1 &30.5 &30.3\\
		8  & 15.1 &13.0 & 30.3&29.3 & 33.3& 33.0& 45.0&45.0\\
		9  & 55.4 &50.0 & 87.4&86.4 & 94.8&94.4& 97.3&98.7\\
		\hline
	\end{tabular}
	\caption{\label{tab:2}Rejection rate in $\%$  for testing the hypothesis $H_0: m\leq 0$ versus $H_1: m> 0$ at level $5\%$ using the asymptotic approach (M1) and the permutation approach (M2) for different sample sizes.}
\end{table}
\begin{table}[h]
	\centering
	\begin{tabular}{c|cc|cc|cc|cc}
		Sett. &\multicolumn{2}{c|}{$n_1=2n_2=600$}& \multicolumn{2}{c|}{$n_1=2n_2=$ 1,200}& \multicolumn{2}{c|}{$n_1=2n_2=$ 4,000} &\multicolumn{2}{c}{$n_1=2n_2=$ 10,000}\\
		 & M1 & M2 & M1 & M2 & M1 & M2 & M1 & M2\\
		\hline
		\hline
		2  &18.8&15.0&15.5&13.2&7.7&10.1&3.4&7.0\\
	\cline{1-9}
 \multicolumn{9}{c}{}\\
&\multicolumn{2}{c|}{$n_1=n_2=500$}& \multicolumn{2}{c|}{$n_1=n_2=$ 1,000}& \multicolumn{2}{c|}{$n_1=n_2=$ 2,000}&\multicolumn{2}{c}{}\\
	& M1 & M2 & M1 & M2 & M1 & M2 &\multicolumn{2}{c}{} \\
  \hline
  \hline
  	3  &2.8&2.7&3.6&3.8&4.1&4.2&&\\
  \hline
		4  &2.0&1.7&2.3&1.9&3.4&3.2&&\\
		\hline
	\end{tabular}
	\caption{\label{tab:6}Rejection rates  in $\%$  for testing the hypothesis $H_0: m\leq 0$ versus $H_1: m> 0$ at level $5\%$  using the asymptotic approach (M1) and the permutation approach (M2) for different sample sizes.}
\end{table}
\begin{table}[h]
	\centering
	\begin{tabular}{c|c|cc|cc|cc}
		\multicolumn{2}{c}{}&\multicolumn{6}{c}{censoring rate}\\
		\multicolumn{2}{c|}{}&\multicolumn{2}{c|}{low}&\multicolumn{2}{c|}{moderate}&\multicolumn{2}{c}{high}\\
		Sett. &$n_1/n_2$& M1 & M2 & M1 & M2 & M1 & M2 \\
		\hline
		\hline
		4& 100/100& 3.0&1.6&2.3&0.6&1.7&0.2\\
		&200/100&2.9&1.0&0.3&0.2&2.4&0.0\\
		\hline
		9&100/100&91.9&91.6&87.4&86.4&79.7&79.1\\
		&200/100&95.9&95.0&94.8&94.4&89.0&87.5\\
		\hline
		\hline
		\multicolumn{2}{c|}{}&\multicolumn{6}{c}{cure rate}\\
		\multicolumn{2}{c|}{}&\multicolumn{2}{c|}{low}&\multicolumn{2}{c|}{moderate}&\multicolumn{2}{c}{high}\\
		Sett. &$n_1/n_2$& M1 & M2 & M1 & M2 & M1 & M2 \\
		\hline\hline
		4& 100/100&2.8 &1.0&2.3&0.6&2.4&0.8\\
		&200/100&3.6&1.1&0.3&0.2&1.9&0.1\\
		\hline
		9	&100/100&95.7&95.5&87.4&86.4&73.1&72.0\\
		&200/100&99.2&99.0&94.8&94.4&79.5&77.6\\
		\hline
	\end{tabular}
	\caption{\label{tab:4}Rejection rate in $\%$  for testing the hypothesis $H_1: m>0$ at level 5\%  using the asymptotic approach (M1) and the permutation approach (M2) for different sample sizes, censoring and cure rates.}
\end{table}

\section{Application}
\label{sec:app_I}

In this section, we apply the developed methods to a real data set from research on leukemia \cite{kersey87}; the study ran from March 1982 to May 1987.
91 patients were treated with high-dose chemoradiotherapy, followed by a bone marrow transplant.
$n_1=46$ patients received allogeneic marrow from a matched donor and $n_2=45$ patients without a matched donor received autologous marrow, i.e.\ their own. 
They were followed for 1.4 to 5 years.
For other details such as additional patient characteristics and the frequency of the graft-versus-host disease among allogeneically transplanted patients we refer to the original study \cite{kersey87}.

In our analysis, we are going to re-analyze relapse-free survival.
The data sets are available in the monograph \cite{peng2021cure}.
They contain the (potentially right-censored) times to relapse or death (in days), together with the censoring status.
It was argued in both \cite{peng2021cure} and \cite{kersey87} that a cure model is appropriate for the data. The authors of the just-mentioned book first fit parametric accelerated failure time mixture cure models (see Section 2.6) and did not find a significant difference in either cure rates or survival times for the uncured. 
Additionally, in Section 3.6 of \cite{peng2021cure} they fit semiparametric logistic-Cox and logistic-AFT models.  Under the logistic-Cox model they did find a significant effect for the survival times of the uncured between both treatment groups. In particular, they concluded that the Autologous group has significantly higher hazard than the Allogeneic group with an estimated hazard ratio of 1.88,  p-value 0.04 and 95\% confidence interval (1.03,3.45). On the other hand, with the semiparametric logistic-AFT model, no statistically significant difference is detected between the two groups.

Let us briefly summarize the data sets.
The censoring percentages amounted to 28\% and 20\%
in the allogeneic and autologous groups, respectively, i.e.\ 13 and 9 patients in absolute numbers.
The percentages of data points in plateaus were 15\% and 16\% 
and the estimated cure fractions $\hat p_i = \hat S_i(\tau_{0,i})$ were 26\% and 19\%, respectively.
The test of \cite{klein2007analyzing} for the equality of cure fractions resulted in a non-significant $p$-value of 0.453.

\begin{figure}[ht]
    \centering
    \includegraphics[width=0.8\textwidth]{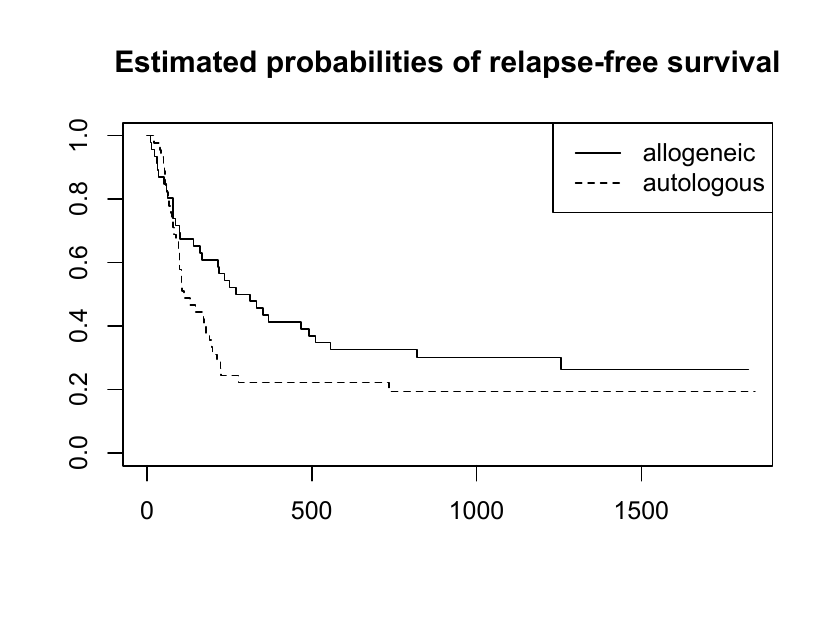}
    \caption{Kaplan-Meier curves for relapse-free survival.}
    \label{fig:kmes_I}
\end{figure}

Figure~\ref{fig:kmes_I} shows an illustration of the Kaplan-Meier curves.
It shows that the curves are crossing and very close to each other in during the first weeks.
After that, the curve for the autologous group clearly stays below the one for the allogeneic group.
However, that discrepancy melts down as time progresses, hence the non-significant $p$-value for the equality of cure fractions.

On the other hand, the difference in estimated mean survival times for the uncured patients amounts to $m=129$.
The $95\%$ confidence intervals for those mean survival differences were [3, 255] (asymptotic) and [1, 255] (permutation).
The two-sided hypothesis tests for equal mean survival times of the uncured resulted in the $p$-values $0.045$ (asymptotic) and $0.046$ (permutation), i.e.\ just significant at the significance level $\alpha=5\%$.
The random permutation-based inference methods have been run with 5,000 iterations.
The asymptotic and the permutation-based inference methods thus agreed on their outcomes, despite the rather small sample sizes.
In view of the rather low censoring rates and the simulation results for moderate censoring settings presented in the previous section, we deem all applied inference methods reliable.

\section{Discussion and interim conclusion}
\label{sec:disc_I}
In this article, we considered a two-sample comparison of survival data in the presence of a cure fraction. In such situations, instead of just looking at the overall survival function, it is more informative to compare the  cured fractions and the survival of the uncured sub-populations. We propose the use of the mean survival time as a summary measure of the survival curve for the uncured subjects since it is a model-free parameter and easy to interpret. We introduced a nonparametric estimator of the mean survival time for the uncured  
and developed both an asymptotic and a permutation-based method for inference on the difference between the $MST_u$. 
Based on our simulation results, both methods were quite reliable, with the permutation approach being recommended particularly for small sample sizes. However, more caution is required when applying the methods in the presence of a high censoring or cure rate, for which larger sample sizes are needed to obtain reliable results.
	
 The $MST_u$ is useful in assessing whether there is a difference in the survival times of the uncured among the two groups. However, we would like to point out that, even in the context of a randomized clinical trial, $MST_u$ 
 does not necessarily have a causal interpretation as a direct effect of the treatment on survival of the uncured. This is because of the conditioning on being uncured;
 the uncured sub-populations between the two treatment arms might, in general, fail to be comparable. For example, it might be that treatment~A is more beneficial in curing patients compared to treatment~B and those who do not get cured with treatment A are the patients with worse condition. As a result, the survival of the uncured for treatment~A might be worse compared to treatment~B, but that does not mean that treatment A shortens the survival of the uncured. However, the fact the $MST_u$ does not have an interpretation as a direct causal effect on the uncured is not a problem when the goal is to choose which treatment should be preferred. Randomized clinical trial data allow us to understand the causal effect of the treatment on the joint distribution of $\1_{T=\infty}$ and $T\1_{T<\infty}$.  One can then in practice define a utility function, for example $w_11_{T=\infty}+w_2T1_{T<\infty}$ for certain weights $w_1,w_2$ that represent whether the curative effect or the life prolonging one is more important. Maximizing the expected utility reduces to choosing the treatment that maximizes $w_1\p(T=\infty)+w_2\p(T<\infty)MST_u$. On the other hand, if one is interested only in the causal effect of the treatment on the uncured subpopulation, determining the relevant quantity is not straightforward and extra caution is required since we are conditioning on a post-treatment variable, which cannot be directly intervened upon.  If all the possible variables that can affect both the cure status and the survival time of the uncured were observed, one could condition on those to estimate the effect of treatment on the uncured but this is a quite unrealistic scenario. 

{We would also like to point out that, instead of comparing the $MST_u$'s over some time horizon $[0, \tau_0]$, trivial adjustments of the methods and proofs would also allow for comparing the mean residual survival times.
Similarly as in \cite{chen2023mean}, these are defined as $MRST_{u,t} = \E[T - t \ | \ t < T < \infty ] $.}

Despite the advantage of not requiring any model assumptions, in several situations, one might want to account for covariate effects when comparing the survival times of the uncured among two groups. This will be the focus of our study in Part II, where the results of this article will be further extended to the conditional $MST_u$ 
estimated under the assumption that each of the groups follows a semiparametric logistic-Cox mixture cure model.
\begin{appendix}

\setcounter{section}{0}

\section{Proofs}\label{app}

	\begin{proof}[Proof of Theorem \ref{theo:MST_homog}]
		Since the KM estimator remains constant after the last observed event time, we have $\hat{p}_i=\hat{S}_i(\tau_{0,i})$ and 
		\[
		\widehat{MST}_{u,i}=\int_0^{\tau_{0,i}}\frac{\hat{S}_{i}(u)-\hat{S}_{i}(\tau_{0,i})}{1-\hat{S}_{i}(\tau_{0,i})}\,\dd u.
		\]
		Hence $\widehat{MST}_{u,i}-{MST}_{u,i}=\psi(\hat{S}_i)-\psi(S_i)$, where
		\[
		\psi: \tilde{D}[0,\tau_{0,i}]\to \R\qquad \psi(\theta)=\int_0^{\tau_{0,i}}\frac{\theta(u)-\theta(\tau_{0,i})}{1-\theta(\tau_{0,i})}\,\dd u
		\]
		and
		\[
		\tilde{D}[0,\tau_{0,i}]=\{f\in{D}[0,\tau_{0,i}]\mid \sup_{t\in[0,\tau_{0,i}]}|f(t)|<1\}.
		\]
		We will derive the asymptotic distribution using the functional delta method based on standard limit results for the KM estimator.
		
		First, {we consider for simplicity the case of a continuous distribution $F_i$ and assume that either condition a) or b)} of the theorem holds. From results of \cite{gill1983,ying1989note}, we have that the stochastic process
		$\sqrt{n_i}\{\hat{S}_i(t)-S_i(t)\}$ 
		converges weakly in $D[0,\tau_{0,i}]$ to the process
		$S_i(t)\cdot B(v_i(t))$
		where $B$ is a standard Brownian motion and $v_i$ is defined as in \eqref{eqn:v_i}.
		Returning to the delta method, the function $\psi$ is Hadamard-differentiable tangentially to $C[0,\tau_{0,i}]$ with derivative $\dd \psi_\theta$ given by
		\[
		\dd \psi_\theta\cdot h=h(\tau_{0,i})\int_0^{\tau_{0,i}}\frac{\theta(u)-\theta(\tau_{0,i})}{\{1-\theta(\tau_{0,i})\}^2}\,\dd u+\int_0^{\tau_{0,i}}\frac{h(u)-h(\tau_{0,i})}{1-\theta(\tau_{0,i})}\,\dd u
		\]
		By Theorem 3.9.4 in \cite{VW96}, we conclude that $ \sqrt{n_i}\{\widehat{MST}_{u,i}-\widehat{MST}_{u,i}\}$ converges weakly to 
		\[
		\begin{aligned}
			N&= S_i(\tau_{0,i})B(v_i(\tau_{0,i}))\int_0^{\tau_{0,i}}\frac{S_i(u)-S_i(\tau_{0,i})}{\{1-S_i(\tau_{0,i})\}^2}\,\dd u \\
			&\qquad+\int_0^{\tau_{0,i}}\frac{S_i(u)B(v_i(u))-S_i(\tau_{0,i})B(v(\tau_{0,i}))}{1-S_i(\tau_{0,i})}\,\dd u\\
			&=p_iB(v_i(\tau_{0,i}))\int_0^{\tau_{0,i}}\frac{S_i(u)-p_i}{(1-p_i)^2}\,\dd u+\int_0^{\tau_{0,i}}\frac{S_i(u)B(v(u))-p_iB(v_i(\tau_{0,i}))}{1-p_i}\,\dd u\\
		\end{aligned}
		\]
		The variable $N$ is normally distributed with mean zero and variance
		\begin{equation}
			\label{def:sigma-i}
			\begin{aligned}
				\sigma_i^2&=\int_0^{\tau_{0,i}}\int_0^{\tau_{0,i}}\frac{S_i(u)S_i(t)}{(1-p_i)^2}v_i(u\wedge t)\,\dd u\,\dd t+\frac{p_i^2}{(1-p_i)^2}({MST}_{u,i}-\tau_{0,i})^2v_i(\tau_{0,i})\\
				&\quad+2\frac{p_i}{(1-p_i)^2}({MST}_{u,i}-\tau_{0,i})\int_0^{\tau_{0,i}}S_i(u)v_i(u)\,\dd u.\\
			\end{aligned}
		\end{equation}
	
{If  $F_i$ is not continuous and either condition a) or  c) of the theorem is satisfied}, we only have weak convergence in $D[0,\tau_{i}]$  of the stopped process 
		\[
		\begin{aligned}
			\hat{L}_i(t)&=\sqrt{n_i}\frac{S_i(t)}{S_i(t\wedge Y_{i,(n)})}\left\{\hat{S}_i(t\wedge Y_{i,(n)})-S_i(t\wedge Y_{i,(n)})\right\}\\
			&=\sqrt{n_i}\left\{\frac{S_i(t)\hat{S}_i(t\wedge Y_{i,(n)})}{{S}_i(t\wedge Y_{i,(n)})}-S_i(t)\right\}
		\end{aligned}
		\]
		where $Y_{i,(n)}$ denotes the largest observation in group $i$ (see for example Theorem 3.14 in \cite{maller1996survival}). {If $\tau_{0,i}<\tau_i$, then $Y_{i,(n)}>\tau_{0,i}$ with probability converging to one, which leads to the uniform convergence of $\sqrt{n_i}\{\hat{S}_i(t)-S_i(t)\}$ 
		on $D[0,\tau_{0,i}]$ as in part a). Otherwise, if $\tau_{0,i}=\tau_i$, } applying the Delta method as before, we would obtain the limit distribution of \[
		\sqrt{n_i}\left\{\psi(\hat{Q}_i)-\psi(S_i)\right\},\qquad Q(t)=\frac{S_i(t)\hat{S}_i(t\wedge Y_{i,(m_i)})}{{S}_i(t\wedge Y_{i,(m_i)})}
		\] 
		It then remains to deal with the difference $\sqrt{n_i}\left\{ \psi(\hat{S}_i)-\psi(\hat{Q}_i)\right\}$ and show that it converges to zero. For this one can use Lemma 3 in \cite{ying1989note}, which does not require continuity of $F_i$. 
	\end{proof}
	
	\begin{proof}[Proof of Theorem~\ref{theo:MST_homog_perm}]
		We begin by analyzing the asymptotic behaviour of $$ (W_1^\pi, W_2^\pi) = \sqrt{\frac{n_1n_2}{n_1+n_2}} (\hat S_1^\pi - \hat S , \hat S_2^\pi - \hat S)$$
		as a random element of $(D[0,\tau_0])^2 $, where $\tau_0 = \max(\tau_{0,1}, \tau_{0,2})$ and $\hat S$ denotes the Kaplan-Meier estimator based on the pooled sample.
		We refer to Lemma~2 in the online supplementary material to \cite{dobler18} for the result that, as $\min(n_1, n_2)\to\infty$,
		the conditional distribution of $(W_1^\pi, W_2^\pi)$ converges weakly on $(D[0,\tau_0])^2$ in probability to the distribution of
		the Gaussian process 
		$$ ((1-\kappa) S(\cdot) B(v(\cdot)),  -\kappa S(\cdot) B(v(\cdot))). $$
		Here, $B$ again denotes a standard Brownian motion, 
		$$S(t) = \exp\Big(- \int_0^t \frac{\kappa (1-G_1(u-)) \dd F_1(u) + (1-\kappa) (1-G_2(u-)) \dd F_2(u)}{\kappa (1-G_1(u-)) S_1(u-) + (1-\kappa) (1-G_2(u-)) S_2(u-)}\Big)$$
		is the limit of the pooled Kaplan-Meier estimator, and $$v(t) = \int_0^t \frac{\kappa (1-G_1(u-)) \dd F_1(u) + (1-\kappa) (1-G_2(u-)) \dd F_2(u)}{\{\kappa (1-G_1(u-))  S_1(u-) + (1-\kappa) (1-G_2(u-)) S_2(u-)\}^2}.$$
		
		Now, write $\widehat{MST}_{u,i}^\pi = \psi(\hat S^\pi_i)$ and $\widehat{MST}_{u} = \psi(\hat S)$ for the estimated mean survival time based on the pooled sample.
		We pointed out the Hadamard-differentiability of $\psi$ in the proof of Theorem~\ref{theo:MST_homog} above.
		Thus, Theorem~3.9.11 in \cite{VW96} applies and it follows that, conditionally on the data, 
		$$\sqrt{n_1 n_2/(n_1 + n_2)}(\widehat{MST}_{u,1}^\pi-\widehat{MST}_{u}, \widehat{MST}_{u,2}^\pi-\widehat{MST}_{u}) $$ converges weakly to the following two-dimensional Gaussian random vector, say, $(N_1, N_2)$ in probability:
		\begin{align*}
			& \Big[ S(\tau_{0})B(v(\tau_{0}))\int_0^{\tau_{0}}\frac{S(u)-S(\tau_{0})}{\{1-S(\tau_{0})\}^2}\,\dd u
			\\
			& \quad +\int_0^{\tau_{0}}\frac{S(u)B(v(u))-S(\tau_{0})B(v(\tau_{0}))}{1-S(\tau_{0})}\,\dd u \Big]
			\cdot (1-\kappa, -\kappa) \\
			& = \Big[ p B(v(\tau_{0}))\int_0^{\tau_{0}}\frac{S(u)-p}{\{1-p\}^2}\,\dd u\\
			& \quad +\int_0^{\tau_{0}}\frac{S(u)B(v(u))-pB(v(\tau_{0}))}{1-p}\,\dd u \Big]
			\cdot (1-\kappa, -\kappa),
		\end{align*}
		where $p = S(\tau_{0})$.
		Finally, we take the difference of both entries of the pair to conclude that the weak limit of the conditional distribution of $\sqrt{n_1 n_2/(n_1 + n_2)}(\widehat{MST}_{u,1}^\pi-\widehat{MST}_{u,2}^\pi) $ is normal with mean zero and variance 
		\begin{equation}
			\label{def:sigma}
			\begin{aligned}
				\sigma^{\pi2}&=\int_0^{\tau_{0}}\int_0^{\tau_{0}}\frac{S(u)S(t)}{(1-p)^2}v(u\wedge t)\,\dd u\,\dd t+\frac{p^2}{(1-p)^2}({MST}_{u}-\tau_{0})^2v(\tau_{0})\\
				&\quad+2\frac{p}{(1-p)^2}({MST}_{u}-\tau_{0})\int_0^{\tau_{0}}S(u)v(u)\,\dd u.\\
			\end{aligned}
		\end{equation}
		As is apparent from a comparison of \eqref{def:sigma-i} and \eqref{def:sigma}, both limit variances coincide up to sample size-related factors if both samples are exchangeable.
	\end{proof}

\end{appendix}

 \begin{acks}[Acknowledgments]
 The authors would like to thank Joris Mooij for insightful comments regarding the causal interpretation and the decision theoretic perspective on choosing between two treatments.
 \end{acks}


\bibliographystyle{imsart-number} 
\bibliography{literature}       


	\newpage
	\setcounter{section}{0}
	\renewcommand*{\thesection}{\arabic{section}}
\renewcommand*{\thesubsection}{\arabic{section}.\arabic{subsection}}
	\setcounter{page}{1}
	\renewcommand*{\theequation}{II.\arabic{equation}}
	\renewcommand*{\thetable}{II.\arabic{table}}
 \renewcommand*{\thefigure}{II.\arabic{figure}}
 
	\begin{frontmatter}
		
		\title{A two-sample comparison of mean \\ \mbox{survival times of uncured sub-populations} -- Part II: Semiparametric analyses}
		\runtitle{Mean survival of uncured patients in two samples -- Part II}
		
		
	\author{\fnms{Dennis} \snm{Dobler${}^*$}\ead[label=e2]{d.dobler@vu.nl}}
\address{Vrije Universiteit Amsterdam \\ The Netherlands \\  \printead{e2}}
\and
\author{\fnms{Eni} \snm{Musta${}^*$}\ead[label=e1]{e.musta@uva.nl}}
\address{University of Amsterdam \\ The Netherlands \\ \printead{e1}}

		\runauthor{D.\ Dobler and E.\ Musta}

\blfootnote{${}^*$: The authors contributed equally and are given in alphabetical order.}
  
		\begin{abstract}
	The restricted mean survival time has been recommended as a useful summary measure to compare lifetimes between two groups, avoiding the assumption of proportional hazards between the groups. In the presence of cure chances, meaning that some of the subjects will never experience the event of interest, it has been illustrated in Part I that it is important to separately compare the cure chances  and the survival of the uncured. In this study, we adjust the mean survival time of the uncured  ($MST_u$)  for potential confounders, which is crucial in observational settings. For each group, we employ the widely used logistic-Cox mixture cure model and estimate the  $MST_u$ conditionally on a given covariate value. An asymptotic and a permutation-based approach have been developed for making inference on the difference of conditional $MST_u$'s between two groups. Contrarily to available results in the literature, in the simulation study we do not observe a clear advantage of the permutation method over the asymptotic one to justify its increased computational cost.  The methods are illustrated through a practical application to breast cancer data. 
		\end{abstract}
		
		\begin{keyword}
			\kwd{asymptotic statistics}
			\kwd{cure models}
		\kwd{random permutation}
  \kwd{inference}
 \kwd{right-censoring}
 \kwd{logistic-Cox mixture model}
 \end{keyword}
		
		
		
	\end{frontmatter}
	
	\section{Introduction}
	Comparing lifetimes between two groups is a common problem in survival analysis. The survival function carries all the necessary information for making such comparison.
 In practice, however, to facilitate interpretations and decision making, it is important to have a single (or at least some few) metric summary measure(s) that quantifies the difference between the two groups. Moving away from the routine use of hazard ratios, the restricted mean survival time (RMST) has been recently advocated as a better and safer alternative, which does not rely on the proportional hazards assumption \cite{royston2011use,royston2013restricted,uno2014moving,dormuth22}. 
 In observational studies, however, it is important to also account for possible confounders.
	One can then adjust for imbalances in the baseline covariates between the two groups by using regression-based methods to estimate the RMST; see \cite{conner2019adjusted,conner2021estimation,ambrogi2022analyzing} and references therein.

	In the presence of a cure fraction, some of the subjects do not experience the event of interest and are referred to as `cured'. In this context,   it was illustrated in Part I that it is more appropriate and informative to separately compare the cure fraction and the survival times of the uncured subpopulations, for example by means of the expected survival times, $MST_u$.
	Despite the clear advantage of using a nonparametric, model-free approach, in the presence of covariate information semiparametric models are usually preferred in terms of interpretability and  balance between complexity and flexibility. Several semiparametric mixture cure models have been introduced in the literature to account for the possibility of cure; see, e.g., \cite{ST2000,peng2000,li2002semi,zhang2007new,lu2004semiparametric}. Among them, the most commonly used in practice is the logistic-Cox mixture model, which considers a logistic model for the cure probability given the covariates and a Cox proportional hazards (PH) model for the survival times of the uncured subjects. For this reason, we will employ a logistic-Cox model for both groups in this paper. Note, however, that the baseline hazards for the uncured subpopulations in the two groups are in general different, leading to non-proportional hazards. Based on the maximum likelihood estimates of the components of a logistic-Cox model (computed using the \texttt{smcure} package in R), we propose an estimator for the conditional mean survival time of the uncured given the covariates. Despite this model choice, the estimation procedure and the results could be similarly extended to other semiparametric mixture cure models. 
	
	Recently, the problem of estimating the conditional mean survival time for the uncured in a one-sample context has been considered in \cite{chen2023mean}. They assume  a semiparametric proportional model for the mean (residual) life of the uncured and  propose an estimation method via inverse-probability-of-censoring weighting and estimating equations. As a consequence, their approach requires also estimation of the censoring distribution. 
		
	In addition to the estimation of the mean survival time for the uncured, we also analyze both the asymptotic and the permutation-based approach for inference on the difference between the conditional $MST_u$ of the two groups, for a fixed covariate value.  In the nonparametric setting, it has been observed that the permutation approach improves upon the asymptotic method for small sample sizes, while maintaining good behavior asymptotically \cite{horiguchi2020,wolski2020,ditzhaus2023}. 
 To the best of our knowledge, the permutation approach has not been used before for semiparametric models in a similar context of maximum likelihood-based estimators that we will pursue. One notable permutation-based inference approach in the survival literature concerns the weighted logrank test \cite{brendel2014weighted}; there, the semiparametric model arises from the form of the null hypothesis which is a cone or subspace of hazard derivatives.
 
 Given also the complexity of our model, several challenges arise from both the computational and theoretical point of view. In order to obtain results on the asymptotic validity of the permutation approach, we first derive a general Donsker-type theorem for permutation based Z-estimators.  Secondly, when fitting the model to the permuted sample, there is an issue of model mispecification. This leads to the convergence of the estimators to a minimizer of the Kullback-Leibler divergence. Thirdly, since the variance estimators need to be computed via a bootstrap procedure, combination of bootstrap and permutation becomes computationally expensive. Hence, it is of interest to investigate whether the gain in accuracy is sufficient to compensate for the increased computation cost compared to the asymptotic approach.
	
	The article is organized as follows.
	The model and notation are introduced in Section~\ref{sec:model_II}.
In Section~\ref{sec:stat_II}, we propose a semi-parametric estimator for the difference in conditional mean survival time of the uncured and derive its asymptotic distribution. Additionally, we also introduce a random permutation approach   for inference and justify its asymptotic validity. The behavior of the asymptotic and permutation based methods for finite sample sizes is evaluated through a simulation study in 
	Section~\ref{sec:simulations_II}. Practical application of the methods is illustrated through a study of breast cancer in Section~\ref{sec:application_II}.
Finally, we conclude  with a discussion in Section~\ref{sec:disc_II}.
	All proofs are contained in Appendix~\ref{sec:app_II}, while the general Donsker-type theorem for permutation based Z-estimators is given in  Appendix~\ref{sec:app_Z-est}. The R code with an implementation of our methods is available in the GitHub repository \url{https://github.com/eni-musta/MST_uncured}.
 Appendix~\ref{sec:app_R} contains a few lines of R code for accessing the breast cancer data analyzed in this paper.

		\section{Model and notation}
	\label{sec:model_II}
	We consider i.i.d.\ survival times $T_{11}, \dots, T_{1n_1}$ and  $T_{21}, \dots, T_{2n_2}$ from two independent groups $(i=1,2)$, each comprising a mixture of cured (immune to the event of interest) and uncured subjects.
 Note that the distributions of the survival times are allowed to differ between groups. For mathematical convenience, the event time of the cured subjects is set to $\infty$, signifying that the event never actually occurs.  Conversely, a finite survival time implies that subjects are susceptible and will experience the event at some point. However, because of censoring, the cure status is only partially known. Specifically, instead of the actual survival times, we observe the follow-up times $Y_{i1},\dots,Y_{in_i}$ and the censoring indicators $\Delta_{i1},\dots,\Delta_{in_i}$, where $Y_{ij}=\min\{T_{ij},C_{ij}\}$, $\Delta_{ij}=\1_{\{T_{ij}\leq C_{ij}\}}$ and $C_{ij}$ are the censoring times.
		Assume that, for each individual in both groups, we observe two covariate vectors
	$X_{ij}\in\R^p$ and $Z_{ij}\in\R^q$, $i=1,2$, $j=1,\dots, n_i$, representing the variables that affect the probability of being susceptible (incidence) and the survival of the uncured (latency). In this way, we allow for these two components of the model to be affected by different variables. However, we do not exclude situations in which the two vectors $X$ and $Z$ are exactly the same or share some components. 
	
 Using the framework of mixture cure models, the relations in \eqref{eqn:distribution_survival} of Part I now hold conditionally on the covariates. In particular, we have that the survival function of $T_{i1}$ given $X_{i1}$ and  $Z_{i1}$ is given by 
	\begin{equation}
		\label{eqn:surv_cov}
		S_i(t|x,z)=\p(T_{i1}>t|{X_{i1}=x, Z_{i1}=z})=p_i(x)+(1-p_i(x))S_{u,i}(t|z),
	\end{equation}
	where $S_{u,i}(t|z)=\p(T_{i1}>t | Z_{i1}=z, T_{i1}<\infty)$ is the conditional survival function of the susceptibles and $p_i(x)=\p(T_{i1}=\infty | X_{i1}=x)$ denotes the conditional cure probability in group $i$. Instead of independent censoring, now we assume that censoring is independent of the survival times conditionally on the covariates: $T\perp C\mid (X,Z)$.
	
	Among various modeling approaches for the incidence and the latency, the most common choice in practice is a parametric model, such as logistic regression, for the incidence and a semiparametric model, such as Cox proportional hazards, for the latency   \cite{legrand2019cure,yilmaz2013insights,stringer2016cure,wycinka2017}. The popularity of such choice  is primarily due to simplicity and interpretability, particularly when dealing with multiple covariates. We focus on this type of models and assume that 
	\[
	1-p_i(x)=\phi(\gamma_i^Tx),
	\]
	where $\phi: \R\to[0,1]$ is a known function, $\gamma_i\in \R^{p+1}$ and  $\gamma^T_i$ denotes the transpose of the vector $\gamma_i$. Here, the first component of $x$ is taken to be equal to one and the first component of $\gamma_i$ corresponds to the intercept. In particular, for the logistic model, we have 
	\begin{equation}
		\label{eqn:logistic}
		\phi(u)=\frac{e^u}{1+e^{u}}.
	\end{equation}
	One can in principle allow also for a different function $\phi$ in the two groups but for simplicity we assume that to be the same. 	
	For the latency, we assume a semiparametric model $S_{u,i}(t|z) =S_{u,i}(t|z;\beta_i,\Lambda_i) $ depending on
	a finite-dimensional parameter $\beta_i\in \mathcal \R^q$, and a function $\Lambda_i$. For example, for the Cox proportional hazards  model, we have
	\begin{equation} 
		\label{eqn:SuCox}
		S_{u,i}(t|z)
		= \exp\{-\Lambda_i(t)\exp(\beta^T_iz)\},
	\end{equation}
	where  $\Lambda_i$ is the baseline cumulative hazard in group $i$.
	
 	One challenge with mixture cure models is model identifiability, 
	i.e., ensuring that different parameter values lead to different distributions of the observed variables. General identifiability conditions for semiparametric mixture cure models were derived by \cite{Parsa}. In the particular case of the logistic-Cox model the conditions are:
	\begin{itemize}
		\item[(I1)]  for all $x$, $ 0 <\phi(\gamma_i^Tx) < 1$,
		\item[(I2)] the function $S_{u,i}$ has support $[0; \tau_{0,i}]$ for some $\tau_{0,i}<\infty$,
		\item[(I3)] $ P(C_{i1} >\tau_{0,i}|X_{i1};Z_{i1}) > 0$ for almost all $X_{i1}$ and $Z_{i,1}$,
		\item[(I4)]  the matrices $Var(X_{i1})$ and $Var(Z_{i1})$ are positive definite,
	\end{itemize}
	Condition I3 corresponds again to the assumption of sufficient follow-up. In practice, this can be evaluated based on the plateau of the Kaplan-Meier estimator and the expert (medical) knowledge. 	
	
	In the presence of covariate information, we are now interested in the difference of mean survival times of the uncured individuals among the two groups conditional on the covatiates: 
	\[
	MST_{u,1,z}-MST_{u,2,z}=\E[T_{11} \mid T_{11} < \infty,Z_{11}=z]-\E[T_{21} \mid T_{21} < \infty,Z_{21}=z].
	\] 
	Note that we use only the covariate $Z$ because that affects the survival of the uncured individuals. In combination with the conditional cure probabilities $p_i(x)$, such conditional mean survival times provide useful
summaries of the conditional survival curves.

	\section{Estimation, asymptotics, and random permutation}
 \label{sec:stat_II}
 The conditional mean survival time can be written as
	\[
	MST_{u,i,z}= \int_0^{\tau_{0,i}}S_{u,i}(u|z)\,\dd u\qquad i=1,2.
	\]
	This leads to the following estimator 
	\[
	\widehat{MST}_{u,i,z}= \int_0^{Y_{i,(m_i)}}\hat{S}_{u,i}(u|z)\,\dd u,
	\]
	where $\hat{S}_{u,i}(\cdot|z)$ is an estimate of the conditional survival function for the uncured and $ Y_{i,(m_i)}$ is the largest observed event time in group $i$. 
	
	Next we focus on the logistic-Cox mixture model, given by \eqref{eqn:logistic}-\eqref{eqn:SuCox}, and consider the plug-in estimate 
	\[
	\hat{S}_{u,i}(t|z)=\exp\left(-\hat\Lambda_i(t)e^{\hat\beta^T_iz}\right),
	\]
	where $\hat\Lambda_i$, $\hat{\beta}_i$ are the maximum likelihood estimates of $\Lambda_i$ and $\beta_i$ respectively. Maximum likelihood estimation in the logistic-Cox model was initially proposed by \cite{ST2000,peng2000} and is carried out via the EM algorithm. The procedure is implemented in the R package \texttt{smcure} \cite{cai_smcure}. In practice, the survival $\hat{S}_{u,i}(t|z)$ is forced to be equal to zero beyond the last event $Y_{i,(m_i)}$, meaning that the observations in the plateau are considered as cured. This is known as the zero-tail constraint as suggested in \cite{ST2000,taylor95} and 
 is reasonable under the assumption of sufficient follow-up: $\tau_{0,i}<\tau_i$, which follows from (I3). 
	
	The asymptotic properties of the maximum likelihood estimates $\hat\Lambda_i$, $\hat{\beta}_i$, $\hat\gamma_i$ were derived in \cite{Lu2008} under the following assumptions:
	\begin{itemize}
		\item[(A1)] The function $\Lambda_i(t)$ is strictly increasing, continuously differentiable on $[0,\tau_{0,i})$ and $\Lambda_i(\tau_{0,i}):=\lim_{t\to\tau_{0,i}}\Lambda_i(t)<\infty$.
		\item[(A2)] $\gamma_i,\beta_i$ lie in the interiors of compact sets and the covariate vectors
		$Z_{ij}$ and $X_{ij}$ have compact support: there exist $m_i >0$ such that:\\ 
  $\p (\Vert Z_{ij}\Vert < m_i \text{ and } \Vert X_{ij}\Vert <
		m_i) = 1 $.
		\item[(A3)] There exists a constant $\epsilon>0$ such that
		$\p(T_{i1} =\tau_{0,i}\mid T_{i1}<\infty,Z_{i1})>\epsilon$ with probability one. 
		\item[(A4)] $\p(Y_{i1}\geq t\mid Z_{i1},X_{i1})$ is continuous in $t \leq \tau_{0,i}$.
	\end{itemize} 
	Assumptions (A1),(A3) are formulated slightly in a different way in \cite{Lu2008} but, given the identifiability constraints (I1)-(I4), they reduce to the ones stated above. 
	In particular, assuming that the survival distribution for the uncured has a positive mass at the end point of the support (A3) is a technical condition needed to guaranty that $\Lambda_i$ stays bounded on $[0,\tau_{0,i}]$ while ensuring the identifiability of the model. In the Cox model without cure fraction, one does not encounter this problem because the support of the event times is larger than the follow-up of the study. Even though (A3) might seem not realistic, one can think of such assumption being satisfied with a very small $\epsilon$. In such case it is unlikely to observe events at $\tau_{0,i}$ as we see in real-life scenarios. If instead of the maximum likelihood estimation, one considers estimation via presmoothing as proposed in \cite{musta2020presmoothing}, this condition can be avoided at the price of additional technicalities. This is because the conditional probability of $\{T=\infty\}$ is identified beforehand by means of a nonparametric smooth estimator. As a result, when estimating $\Lambda_i$ in the second step, one could restrict to a smaller interval $[0,\tau^*]\subset[0,\tau_{0,i}]$; see the discussion in Section 5.1 of \cite{musta2020presmoothing}.

	Let $\hat{m}_z=	\widehat{MST}_{u,1,z}-\widehat{MST}_{u,2,z}$ be an estimator of $m_z=MST_{u,1,z}-MST_{u,2,z}$. 
	Using the large sample properties of the estimators $\hat\Lambda_i$ and $\hat\beta_i$ from \cite{Lu2008}, we first derive the limit distribution of the process $\sqrt{n_i}\{\hat{S}_{u,i}(\cdot|z)-S_{u,i}(\cdot|z)\}$ and then obtain the following result for the conditional mean survival time. 
	\begin{thm}
		\label{theo:MST_cov}
		Assume that the identifiability conditions (I1)-(I4) and the assumptions (A1)-(A4) are satisfied. Then, for any $z\in\mathcal{Z}$, the variable $\sqrt{n_i}(\hat{E}_{i,z}-E_{i,z})$ is asymptotically normally distributed, i.e., $$\sqrt{n_i}(\widehat{MST}_{u,i,z}-MST_{u,i,z})\xrightarrow{d}N\sim\mathcal{N}(0,\sigma_{i,z}^2)$$ as $n_i\to\infty$. The limit variance $\sigma_{i,z}^2$ is defined in \eqref{eqn:sigma_cov} in the Appendix~\ref{sec:app_II}.
	\end{thm}
 From the independence assumption between the two groups and Theorem~\ref{theo:MST_cov},
we obtain the following result for which we define $a_n=\sqrt{n_1n_2/(n_1+n_2)}$.
	\begin{cor}
		\label{cor:MST_cov}
		Assume that $n_1/(n_1+n_2)\to\kappa\in(0,1)$ as $\min(n_1,n_2) \to \infty$. Then $a_n(\hat{m}_z-m_z)$ is asymptotically normally distributed with mean zero and variance $$\sigma^2_z=(1-\kappa)\sigma_{1,z}^2+\kappa\sigma_{2,z}^2.$$
	\end{cor}
	We restrict for simplicity to the logistic-Cox model and the maximum likelihood estimation method but the previous results can be generalized to other estimation methods or other semiparametric mixture cure models. For example, if the presmoothing approach introduced in \cite{musta2020presmoothing} is used instead of the MLE, then the asymptotic properties could be derived in the same way  using Theorem 4 in \cite{musta2020presmoothing}. We also note that because of their complicated expressions, the variances of the estimators in the semiparametric mixture cure model are estimated via a bootstrap procedure \cite{cai_smcure}. As a result, we will also use the bootstrap to estimate~$\sigma_z.$
\\[0.2cm]

 As in Part I, we would like to investigate whether the asymptotic inference can be made more reliable by means of a permutation approach. Again $\pi=(\pi_1, \dots, \pi_{n_1+n_2})$ denotes any permutation of $(1,2, \dots, n_1+n_2)$.
{Write 
  $ (Y_j, \Delta_j,X_j,Z_j), j=1,\dots, n_1+n_2$ for the pooled sample which consists of the data points of the first group $(j \leq n_1)$ and those of the second group $(j> n_1)$.}	
 Applying $\pi$ to the pooled sample leads to the permuted samples $$(Y_{\pi_1}, \Delta_{\pi_1},X_{\pi_1},Z_{\pi_1}), \dots,  
   (Y_{\pi_{n_1}}, \Delta_{\pi_{n_1}},X_{\pi_{n_1}},Z_{\pi_{n_1}}) \quad \text{and}$$ 
   $$(Y_{\pi_{n_1+1}}, \Delta_{\pi_{n_1+1}},X_{\pi_{n_1+1}},Z_{\pi_{n_1+1}}), \dots, (Y_{\pi_{n_1+n_2}}, \Delta_{\pi_{n_1+n_2}},X_{\pi_{n_1+n_2}},Z_{\pi_{n_1+n_2}}).$$
		In the special case of exchangeability, for any $z$, 
	$\hat m_z $ would have the same distribution as $\hat m^{\pi}_z = \widehat{MST}_{u,1,z}^\pi - \widehat{MST}_{u,2,z}^\pi$. Here, $\widehat{MST}_{u,i,z}^\pi$ are the estimators of the conditional mean survival times, just based on the $i$-th permuted sample.
	
Since the samples are in general not exchangeable and, as a consequence, the asymptotic variances of $\hat m_z$ and $\hat m^{\pi}_z$ cannot be assumed equal, we use their studentized version.
	We will thus focus on $\hat m^\pi_z/\hat{\sigma}^{\pi}_z$ as the permutation version of $\hat m_z/\hat \sigma_z$, where
	$\hat{\sigma}_z^{\pi2}$ is the estimated variance of  $\hat m^\pi_z$, estimated via bootstrap. 
	Because it is computationally infeasible to realize $\hat m^\pi_z/\hat{\sigma}^{\pi}_z$ for all $(n_1+n_2)!$ permutations, we will realize a relatively large number $B$ of random permutations $\pi$ and approximate the conditional distribution by the collection of the realized $\hat m^\pi_z/\hat{\sigma}^{\pi}_z$ of size $B$.
		In the following, we will discuss the asymptotic behaviour of $\hat m^\pi_z/\hat{\sigma}^{\pi}_z$ to justify the validity of the resulting inference procedures.

 One challenge that arises in this setting is that, since we are assuming a semiparametric model, the permuted samples will in general not follow the same model. Hence, when we fit the logistic-Cox model to obtain the estimates in the permuted samples, the model is misspecified.  Hence, we first show in a series of lemmas in  Appendix~\ref{sec:app_II} that the maximum likelihood estimators converge to the parameters of a logistic-Cox likelihood that minimize the Kullback-Leibler divergence from the true distribution of the pooled data. We can indeed argue that such minimizer exists and we assume that it is unique. In case of non-uniqueness, we expect that the results can be extended and, depending on the starting point of the algorithm, the estimates would converge to one of such minimizers. However, such extension is beyond the scope of the current paper. In practice, we observed that the EM algorithm converges and the limit was stable with respect to the initial point, which might indicate that the minimizer was indeed unique.

 Secondly, to obtain the asymptotic distribution of the permuted estimators, we first obtain a general Donsker-type theorem for permutation based
Z-estimators (see Appendix~\ref{sec:app_Z-est}).
That result holds in a great generality, so it would also apply to countless other two sample problems. Thus, it is of interest of its own. But let us first return to the main result about the permuted estimators in the present context:
	\begin{thm}
		\label{thm:perm_cov}
		Assume that the identifiability conditions (I1)-(I4) and (A1)-(A4) hold. Assume also that the minimizer of the KL divergence definied in \eqref{eqn:bar_parameters} in Appendix~\ref{sec:app_II} is unique. Then, for any $z\in\mathcal{Z}$, as
$\min(n_1, n_2) \to\infty$ with $n_1/(n_1 + n_2) \to\kappa\in(0, 1)$, the conditional distribution of 
$a_n\hat m_z^\pi$ given the data converges weekly in probability to the zero-mean normal distribution with variance $\sigma^{\pi2}_z$ given in \eqref{eqn:sigma_cov2} in Appendix~\ref{sec:app_II}.
	\end{thm}
	\begin{cor}
	\label{cor:MST_cov_perm}
	Under the assumptions of the previous theorem,  for any $z\in\mathcal{Z}$, as $\min(n_1, n_2)\to\infty$ with $0 < \liminf n_1/(n_1+n_2) \leq \limsup n_1/(n_1+n_2) < 1$, the (conditional) distributions of 
  $a_n\hat m_z^\pi/\hat \sigma^{\pi}_z$ (given the data) and 
  $a_n(\hat m_z - m_z)/\hat \sigma_z$ converge weakly (in probability) to the same limit distribution which is standard normal.
	\end{cor}

	\section{Simulation study}
 \label{sec:simulations_II}
	In this section, we investigate the practical performance of the permutation approach and of the asymptotic method when comparing mean survival times for the uncured sub-population in a two-sample problem. We consider two samples of size 200 and 100, respectively, from logistic-Cox mixture cure models. Note that, in practice, semiparametric cure models are usually not used for sample sizes much smaller than these because of their complexity (more parameters need to be estimated compared to standard Cox model for example) and the need to observe a long plateau with a considerable amount of censored observations (as a confirmation of the sufficiently long follow-up assumption). Since the permutation approach is computationally intensive and asymptotically we expect the behavior of the two methods to be more similar, we also did not consider larger sample sizes. Instead, we focus on three different scenarios as described below by varying the distributions of the uncured subjects, the cure proportions and the censoring rates. For simplicity, we also consider the same covariates in the incidence and latency components, i.e., $X=Z$. 

 \section*{Simulation settings}
 
	\textit{Setting 1.} 
	Both samples are generated from the logistic-Cox mixture cure model with Weibull baseline distribution with shape parameter $0.75$ and scale parameters $1.5$ and $2$, respectively.  The survival times of the uncured subjects are truncated at $\tau_{0,i}$, $i=1,2$ equal to the $99\%$ quantile of the corresponding baseline Weibull distribution in order to have finite supports. 
	We consider two independent covariates $Z_1$ and $Z_2$, which affect both the cure probability and the survival of the uncured. In the first sample,  $Z_1\sim N(0,1)$, $Z_2\sim Bern(0.4)$ while in the second sample $Z_1\sim N(1,1)$, $Z_2\sim Bernoulli(0.6)$. The regression coefficients are $\gamma_1=(0,0.5,0.8)$, $\beta_1=(0.3,0.5)$, $\gamma_2=(0.1,1,0.6)$, $\beta_2=(0.3+\log(0.75),0.5)$. This corresponds to having around $43\%$ and $24\%$ cured subjects in each sample.  The censoring times are generated independently of the other variables from exponential distributions with parameters $0.4$ and $0.2$, respectively. They are truncated at $\tau_i=\tau_{0,i}+2$, $i=1,2$ to reflect the limited length of studies in practice. The censoring rate in sample 1 is $52\%$, while in sample 2 it is $28\%$. In both cases, we have around $15\%$ of the observations in the plateau. In this setting, the covariate distributions, the cure and censoring rates, and the survival distributions of the uncured are different among the two samples. However, depending on the covariate values, the conditional mean survival times of the uncured can be the same, i.e., $m_z=0$. We consider a range of possible covariates, see Table~\ref{tab:z}, including some extreme and unlikely values in order to get more different mean survival times between the two groups. \\[0.2cm]
	\textit{Setting 2.} 	Both samples are generated from logistic-Cox mixture cure models with Gompertz baseline distribution with shape parameter $1$ and rate parameters $0.1$, $0.3$, respectively. The survival times of the uncured subjects are truncated at $\tau_{0,i}$, $i=1,2$, equal to the $99\%$ quantiles of the corresponding baseline distributions in order to have finite supports. We consider two independent covariates, $Z_1\sim N(0,1)$ and $Z_2\sim Unif(-1,1)$ with the same distribution in both samples. The regression coefficients are $\gamma_1=\gamma_2=(0.8,-1,1)$, $\beta_1=(-0.6,0.5)$, $\gamma_2=(0.1,1,0.6)$, $\beta_2=(-0.05,0.4)$. This corresponds to having around $35\%$ cured subjects in each sample. The censoring times are generated independently of the other variables from exponential distributions with parameters, $0.1$ and $0.2$, respectively. They are truncated at $\tau_i=\tau_{0,i}+2$, $i=1,2$, to reflect the limited length of studies in practice. The censoring rate in sample 1 is $46\%$, while in sample 2 it is $48\%$. In both cases, we have around $20\%$ of the observations in the plateau. In this setting, the covariate distributions, the cure rates, and the censoring rates are the same for both samples. The survival distributions of the uncured are different but again, for certain values of the covariates, the conditional mean survival times of the uncured are the same. 	We consider different covariate values as in Table~\ref{tab:z}. In particular, $z_8$ is a very extreme and unlikely value but it was considered in order to have a case with larger negative value for $m_z$.\\[0.2cm]
	\textit{Setting 3.} 	
	Both samples  are generated from the same distribution as for sample~1 in Setting~1. This means that the two samples are exchangeable and for any covariate value we have $m_z=0$. 
	\begin{table}[h]
		\begin{tabular}{c|ccccccc}
			\hline\noalign{\smallskip}
			Setting 1 & $z_1=(0,1)$& & $z_2=(-1,0)$ && $z_3=(1,0)$ && $z_4=(1,1)$ \\
			& $m_{1}=0.11$&& $m_{2}=0.5$&& $m_{3}=0$&& $m_4=0$\\
			\cline{2-8}\noalign{\smallskip}
			& $z_5=(2,1)$ &&$z_6=(4,0)$ && $z_7=(-4,0)$ & & $z_8=(-3,1)$\\
			& $m_5=-0.07$ &&$m_6=-0.3$&& $m_7=1.68$&&$m_8=0.83$\\
			\hline\noalign{\smallskip}
			Setting 2 & $z_1=(-2,0)$&& $z_2=(-1.85,0.8)$&& $z_3=(-2.16,-0.8)$&& $z_4=(0,0)$\\
			&$m_{1}=0$&& $m_{2}=0$&&$m_{3}=0$&& $m_4=0.79$\\	
			\cline{2-8} \noalign{\smallskip}
			& $z_5=(1,0.5)$&& $z_6=(-1,-0.5)$&& $z_7=(-3,0.5)$&& $z_8=(-6,0)$\\
			& $m_5=1.16$& &$m_6=0.43$& &$m_7=-0.31$&& $m_8=-0.82$\\
			\cline{2-8}\noalign{\smallskip}
			& $z_9=(2,0)$&& $z_{10}=(-1.5,0)$&& $z_{11}=(-2.5,0)$&&\\
			&$m_{9}=1.65$&&$m_{10}=0.18$&& $m_{11}=-0.16$&&\\
			\hline	
		\end{tabular}
		\caption{\label{tab:z}Covariate values $z$ and corresponding difference in conditional mean survival times for the uncured $m_z$ for Settings 1 and 2.}
	\end{table}
	
\section*{Simulation results}
	For each setting and covariate value, we construct $1-\alpha=95\%$ confidence intervals for $m_z$ based on the asymptotic and the permutation approach 
	\[
	I_z^*=[\hat{m}_z\mp q_{1-\alpha/2}\hat\sigma_{z}/a_n],\qquad I_z^\pi=[\hat{m}_z-q^\pi_{1-\alpha/2},\hat{m}_z- q^\pi_{\alpha/2}\hat\sigma_{z}]
	\]
	where $\hat{m}_z$ is computed as in Section \ref{sec:stat_II}, $\hat\sigma_{z}^2$ is the variance of $a_n\hat{m}_z$ estimated via the bootstrap,  $q_{1-\alpha/2}$ denotes the quantile of the standard normal distribution and $q^\pi_{1-\alpha/2}$ is the quantile of the distribution of $\hat{m}_z^\pi/\hat\sigma_z^\pi$.
 Due to the computational cost, 
 we use 100 bootstrap samples 
 combined with 
 $500$ random permutation samples. This procedure was repeated 1,000 times. The lengths and coverage probabilities of the confidence intervals are given in Tables \ref{tab:conf_int}-\ref{tab:conf_int3}. The coverage rates
closest to 95\% among both types of confidence intervals is printed in bold-type. For the exchangeable Setting~3, since the permutation confidence intervals are exact, we only provide the results of the asymptotic approach. 
	\begin{table}[h]
		\centering
		\begin{tabular}{c|cc|cc|cc|cc}
			&\multicolumn{2}{c|}{$z_1$}& \multicolumn{2}{c|}{$z_2$}& \multicolumn{2}{c}{$z_3$}&\multicolumn{2}{c}{$z_4$}\\
			&M1 & M2 &M1 & M2&M1 & M2  &M1 & M2 \\
   \hline
   \hline
			L &0.498&0.476&1.294&1.385&0.703&0.716&0.372&0.353\\
			CP &96.5&\textbf{95.3}&\textbf{90.6}&89.8&94.5&\textbf{95.5}&97.4&\textbf{96.3}\\
			&&&&&&&&\\
			&\multicolumn{2}{c|}{$z_5$}& \multicolumn{2}{c|}{$z_6$}& \multicolumn{2}{c|}{$z_7$}&\multicolumn{2}{c}{$z_8$}\\
			&M1 & M2  &M1 & M2 &M1 & M2 &M1 & M2 \\
     \hline
   \hline
			L &0.429&0.393&1.168&1.012&2.606&3.119&1.846&1.925\\
			CP &98.3&\textbf{96.4}&\textbf{95.6}&91.4&68.6&\textbf{72.5}&\textbf{85.9}&85.6\\
		\end{tabular}
		\caption{	\label{tab:conf_int}Coverage probabilities (CP) in $\%$ and length (L) of 95\% confidence intervals using asymptotic approach (M1) and the permutation approach (M2)  for different choices of $z$ in Setting 1.}
	\end{table}
	\begin{table}[h]
		\centering
		\begin{tabular}{c|cc|cc|cc|cc}
			&\multicolumn{2}{c|}{$z_1$}& \multicolumn{2}{c|}{$z_2$}& \multicolumn{2}{c}{$z_3$}&\multicolumn{2}{c}{$z_4$}\\
			&M1 & M2 &M1 & M2&M1 & M2  &M1 & M2 \\
     \hline
   \hline
			L &0.970&0.952&1.003&0.985&1.308&1.291&0.641&0.64\\
			CP &\textbf{93.9}&93.3&\textbf{92.9}&92.8&\textbf{94.1}&93.4&\textbf{97.0}&\textbf{97.0}\\
			&&&&&&&&\\
			&\multicolumn{2}{c|}{$z_5$}& \multicolumn{2}{c|}{$z_6$}& \multicolumn{2}{c|}{$z_7$}&\multicolumn{2}{c}{$z_8$}\\
			&M1 & M2  &M1 & M2 &M1 & M2 &M1 & M2 \\
     \hline
   \hline
			L &1.038&1.060&0.851&0.845&1.311&1.263&2.097&1.669\\
			CP &\textbf{97.1}&97.3&\textbf{94.3}&94.2&\textbf{92.7}&91.9&\textbf{89.9}&81.5\\
			&&&&&&&&\\
			&\multicolumn{2}{c|}{$z_9$}& \multicolumn{2}{c|}{$z_{10}$}& \multicolumn{2}{c|}{$z_{11}$}&\multicolumn{2}{c}{}\\
			&M1 & M2  &M1 & M2 &M1 & M2 & & \\
     \hline
   \hline
			L &1.526&1.571&0.785&0.775&1.158&1.130&&\\
			CP &98.1&\textbf{97.5}&\textbf{94.0}&93.3&\textbf{93.2}&93.1&&\\
		\end{tabular}
		\caption{	\label{tab:conf_int2}Coverage probabilities (CP) in $\%$ and length (L) of 95\% confidence intervals using the asymptotic approach (M1) and the permutation approach (M2)   for different choices of $z$ in Setting 2.}
	\end{table}
	\begin{table}[h]
		\centering
		\begin{tabular}{c|c|c|c|c|c|c|c|c}
			&	$z_1$& $z_2$&$z_3$	&{$z_4$}& {$z_5$}&
			{$z_6$}& {$z_7$}& $z_8$\\
     \hline
   \hline
			L &0.707&0.795&0.974&0.587&0.680&1.338&3.512&2.726\\
			CP &94.7&91.0&95.3&97.0&98.6&99.0&83.5&90.1\\
		\end{tabular}
		\caption{	\label{tab:conf_int3}Coverage probabilities (CP) in $\%$ and length (L) of 95\% confidence intervals using the asymptotic method for different choices of $z$ in Setting 3.}
	\end{table}
	
	We observe that the coverage of both confidence intervals is very low for some covariate values. That happens mainly when $m_z$ is large in absolute value (either  positive or negative depending on the setting). This seems to be because, for certain covariate values, the errors that we make in the estimation of the coefficients and baseline survival get amplified when computing the survival function conditional on $z$, resulting in a biased estimate for $m_z$. Much larger sample sizes would be needed to get a good estimate of $m_z$ for such covariates~$z$. In the other cases, the coverage probabilities are close to the nominal value and the two methods are comparable. For some $z$, the permutation approach does slightly better than the asymptotic one, but vice versa for other choices of~$z$.
	
	Next, we consider testing the hypothesis $H_0: m_z=0$ against $H_1: m_z\neq 0$ at level $\alpha=5\%$. Again 
	the variance of $\hat{m}_z$ is estimated via the bootstrap with 100 bootstrap samples and the quantiles of $\hat{m}_z/\hat\sigma_z$ are estimated via $500$  permutation samples. The procedure was repeated 1,000 times. 
	In  Tables \ref{tab:test}-\ref{tab:test3}, we report the percentages of the cases in which $H_0$ was  rejected.
	
	\begin{table}[h]
		\centering
		\begin{tabular}{c|cc|cc|cc|cc}
			&\multicolumn{2}{c|}{$z_1$}& \multicolumn{2}{c|}{$z_2$}&
			\multicolumn{2}{c|}{$z_3$}& \multicolumn{2}{c}{$z_4$}\\
			&M1 & M2  &M1 & M2  &M1 & M2 &M1 & M2 \\
     \hline
   \hline
			Rejection rate &14.1&11.8&22.9&15.2&5.5&4.5&2.6&3.7\\
			&&&&&&&&\\
			&\multicolumn{2}{c|}{$z_5$}& \multicolumn{2}{c|}{$z_6$}&
			\multicolumn{2}{c|}{$z_7$}& \multicolumn{2}{c}{$z_8$}\\
			&M1 & M2  &M1 & M2  &M1 & M2 &M1 & M2 \\
     \hline
   \hline
			Rejection rate &6.9&13.6&8.5&17.8&33.6&19.5&26.7&22.1\\
		\end{tabular}
		\caption{	\label{tab:test}Rejection rates of $H_0$ in $\%$ for the asymptotic approach (M1) and the permutation approach (M2) and different choices of $z$ in Setting 1.
 {$H_0$ is true only for $z_3$ and $z_4$.}}
	\end{table}
	
	\begin{table}[h]
		\centering
		\begin{tabular}{c|cc|cc|cc|cc}
			&\multicolumn{2}{c}{$z_1$}& \multicolumn{2}{c}{$z_2$}&
			\multicolumn{2}{c|}{$z_3$}& \multicolumn{2}{c}{$z_4$}\\
			&M1 & M2  &M1 & M2  &M1 & M2 &M1 & M2 \\
     \hline
   \hline
			Rejection rate &6.1&6.7&7.1&7.2&5.1&6.6&99.9&100.0\\
			&&&&&&&&\\
			&\multicolumn{2}{c|}{$z_5$}& \multicolumn{2}{c|}{$z_6$}&
			\multicolumn{2}{c|}{$z_7$}& \multicolumn{2}{c}{$z_8$}\\
			&M1 & M2  &M1 & M2  &M1 & M2 &M1 & M2 \\
     \hline
   \hline
			Rejection rate &99.5&99.3&48.3&49.3&15.4&19.9&30.6&52.7\\
			&&&&&&&&\\
			&\multicolumn{2}{c|}{$z_9$}& \multicolumn{2}{c|}{$z_{10}$}& \multicolumn{2}{c|}{$z_{11}$}&\multicolumn{2}{c}{}\\
			&M1 & M2  &M1 & M2 &M1 & M2 & & \\
     \hline
   \hline
			Rejection rate &99.2&99.0&14.1&15.1&9.7&11.8&&\\
		\end{tabular}
		\caption{\label{tab:test2}Rejection rates of $H_0$ in $\%$ for the asymptotic approach (M1) and the permutation approach (M2)  and different choices of $z$ in Setting 2.
   {$H_0$ is true only for $z_1$, $z_2$, and $z_3$.}}
	\end{table}
	\begin{table}[h]
		\centering
		\begin{tabular}{c|c|c|c|c|c|c|c|c}
			&	$z_1$& $z_2$&$z_3$	&{$z_4$}& {$z_5$}&
			{$z_6$}& {$z_7$}& $z_8$\\
     \hline
   \hline
			Rejection rate &5.3&9.0&4.7&3.0&1.4&1.0&16.5&9.9\\
		\end{tabular}
		\caption{	\label{tab:test3}Rejection rates of $H_0$ in $\%$ for the asymptotic method and different choices of $z$ in Setting~3.
   {$H_0$ is true for all choices of $z$.}}
	\end{table}
	
	In Settings 1 and 2, the levels of the test seem to be close to the nominal level and the power is larger when $|m_z|$ is larger, even though it does not only depend on $|m_z|$ but also on the sign of $m_z$ (deviations in conditional mean survival times might be easier to detect in one direction compared to the other).  
	In Setting~1, the asymptotic method has more power when $m>0$ ($m_1, m_2,m_7,m_8$),  while the permutation approach has more power when $m<0$ ($m_5, m_6$). 	In Setting 2, we observe that the results for both methods are comparable when $m>0$ but the permutation approach  has  more power when $m<0$.  For the exchangeable setting, the level of the asymptotic test is larger than the nominal value for some of the covariate values (the ones for which the coverage probabilities were anti-conservative; see above).

	Overall, we conclude that, unless the two samples are exchangeable, there is no clear advantage of using the permutation approach to justify its much higher computational cost. This is different from what is observed previously in the literature and it might be related to the fact that the logistic-Cox model is misspecified in the permutation samples. Computationally, the EM algorithm still converges and is stable with respect to the initial estimates. This suggest that the problem should not be about the existence of a unique maximizer of the likelihood for the misspecified model.

	\section{Application}
	\label{sec:application_II}
	In this section, we illustrate the methods developed in this paper by analyzing a data set about breast cancer.
    In Appendix~\ref{sec:app_R}, we provide the R code for accessing this freely available data set.
	The data come from an observational study that included 286 lymph-node-negative breast cancer patients collected between 1980 and 1995. 
	Thereof, 209 patients were oestrogen-receptor-positive (ER+) and 77 were ER-negative (ER-).
	These two will later form the subgroups to be analyzed in a two-sample inference problem.
	As additional covariates, we consider the patients' age (ranging from 26 to 83 with a median of 52 years) and a tumour size score which is an integer number between 1 and 4.
	We refer to \cite{wang05} for a more complete description of the study and other specifics of the dataset. 
	Additionally, \cite{amico2019} compared this dataset in the light of two competing models and corresponding statistical methods: a Cox-logistic cure model versus a Single-Index/Cox model.
	
	We, on the other hand, do not model the ER-status semiparametrically but nonparametrically by means of two subgroups.
	Our aim is to conduct a regression analysis to investigate differences in disease progression expectations for ER+/- patients while taking the covariates tumour size (ordinal) and age into account.
	The outcomes of this two-sample analysis could be used to justify why the two groups should not be pooled, and how or how not to model the ER-status semiparametrically.
	These questions are relevant if one wishes to make predictions, e.g., for the remaining expected lifetime of a patient.
	
	From a technical point of view, we consider the composite endpoint of relapse-free survival (measured in months), which here means that deaths and the occurrence of distant metastases are combined into one event of interest.
	We excluded those 8 patients from our analysis who had a tumour size  exceeding 2.
	This results in two samples of sizes $n_1=203$ (ER+) and $n_2= 75$ (ER-).
	
	Let us briefly summarize the data: in the ER+ subgroup, about 55\% have a tumour size score of 1, as opposed to 47\% in the ER- subgroup.
	The age distributions in the four subgroups (ER+/-, tumour size score 1/2) are generally similar (rather symmetric, no outliers), although the patients with ER- and smaller tumour sizes exhibit a smaller dispersion in age; see Figure~\ref{fig:box}.
	For both groups, the latest uncensored events were observed after 80 and 48 months, respectively.
	The censoring rates amount to 62\% and 64\%, respectively, and the majority of censorings occurred in the plateau.
	Thus, there is sufficient follow-up.
	
	\begin{figure}[ht]
		\centering
		\includegraphics[width=0.8\textwidth]{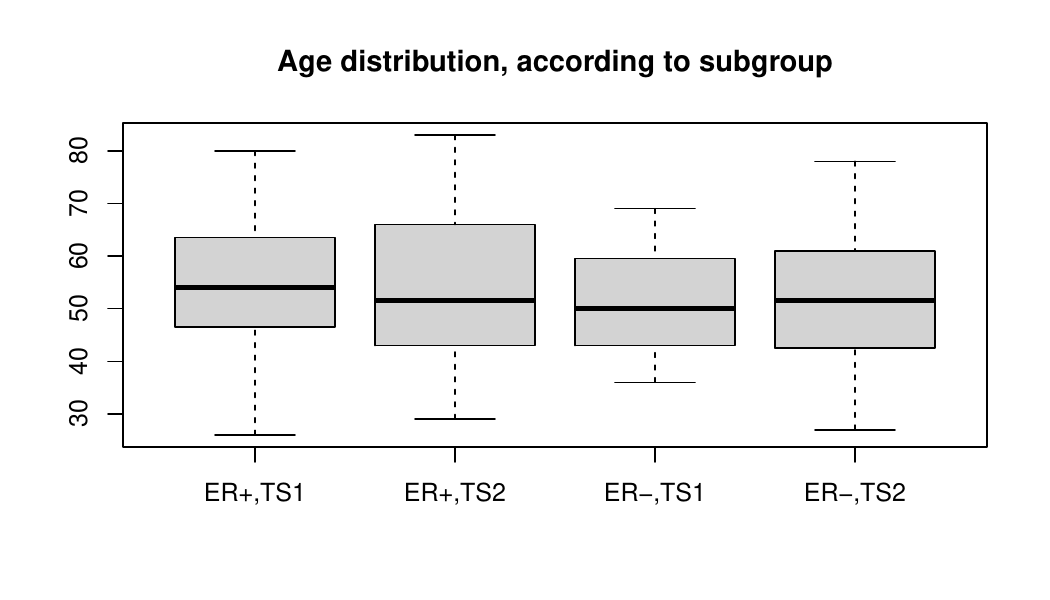}
		\caption{Boxplots summarizing the age distributions in the four relevant subgroups. TS is short for tumour size.}
		\label{fig:box}
	\end{figure}
	
	Figure~\ref{fig:kmes} shows the nonparametric Kaplan-Meier estimates for relapse-free survival for the subsamples of ER+ and ER- patients. 
	These two curves are crossing twice: once, but insignificantly, soon after the time origin, and once again after the last observed event in the ER- subgroup.
	These crossings underline that the classical proportional hazards model \citep{cox1972} might not be appropriate for a combined modeling of all these data within a single, extended Cox-logistic cure model:
	such a model would contradict crossing Kaplan-Meier curves as seen in Figure~\ref{fig:kmes} (after rescaling both curves to exhibit the same cure rate). 
	
	\begin{figure}[ht]
		\centering
		\includegraphics[width=0.8\textwidth]{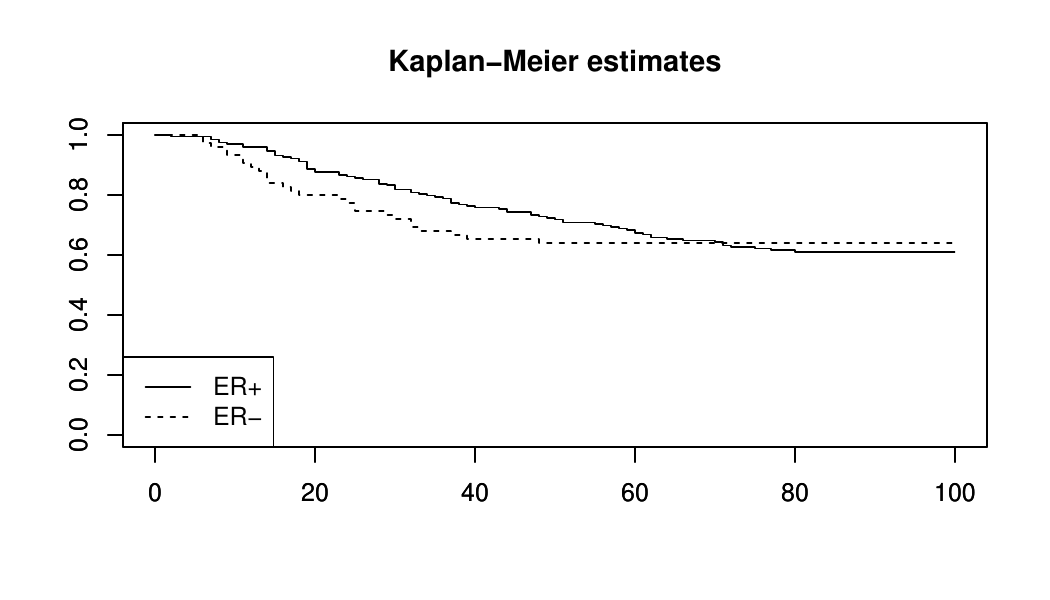}
		\caption{Nonparametric Kaplan-Meier curves for the ER+ and ER- subgroups (time in months).}
		\label{fig:kmes}
	\end{figure}
	
	We begin our inferential data analysis by observing a non-significant difference in cure rates between both groups (ER+/-): the nonparametric test of \cite{klein2007analyzing} resulted in a $p$-value of $p=0.659$.
	This motivated us to check for differences in the residual mean survival times for uncured patients.
	
	Let us now turn to the hypothesis tests for equal differences in mean survival times, for $H_0: m=0 $ versus $H_a: m\neq 0$, at level $\alpha$ and $(1-\alpha)$-confidence intervals for the difference $m$.
	Throughout, we have chosen $\alpha=5\%$, 500 permutations, and 100 bootstrap iterations.
	Starting with our nonparametric test of equal residual mean survival times for the uncured, we first of all obtained a point estimate of $\hat m=16$.
	This resulted in $p$-values of $p=0$ for both the asymptotic and the permutation approach.
	The corresponding 95\%-confidence intervals are $[10,22]$ 
	and $[9,22]$,  
	respectively.

	Our present approach is to compare the outcomes of two independently fitted Cox-logistic models for both sample groups ER+/-.
	We have first checked the proportional hazards assumption for the latency parts of the model by means of the test proposed in \cite{peng2017residual}.
	The resulting $p$-values are 1 and 0.96 respectively for the covariates age and tumour size in sample 1; 1 and 0.81 respectively for the covariates age and tumour size in sample 2.
	Thus, there is no reason to reject the proportional hazards assumption.
	We would also like to point out that we do not rely on the proportional hazards assumption for ER status, which is nonparametrically modeled in terms of two separate subgroups.

	Thus continuing with the two independently fitted Cox-logistic models, Table~\ref{tab:param} contains all point estimates of the parametric model components.
    Note that none of the covariates was found to have a significant influence on any of the two models.
	From the point estimates, we also see that for ER+ patients age seems protective for both, incidence ($\gamma_1 <0$) and latency ($\beta_1<0$), and  generally harmful for ER- ($\gamma_1, \beta_1> 0$).
	Also, a bigger tumour size is generally harmful for ER+ patients ($\gamma_2, \beta_2>0$) but, for ER- patients, it has a protective influence on incidence ($\gamma_2 <0$) although a harmful influence on the latency ($\beta_2>0$). 
    Testing whether there is a significant difference of any of these parameters from 0 could be easily achieved by studentizing the parameter point estimates and comparing the results with quantiles from the standard normal distribution or those from a corresponding permutation version of the studentized parameter estimates. Since this was not our main focus, these results are not shown.
	
	\begin{table}[ht]
		\centering
		\begin{tabular}{c|ccccc}
			ER & $\gamma_0$ & $\gamma_1$ & $\gamma_2$ & $\beta_1$ & $\beta_2$ \\ \hline
			+ & $-0.838(0.438)$ & $-0.018(0.013)$ & $0.272(0.296)$ & $-0.014(0.011)$ & $0.482(0.260)$ \\
            & $p=0.056$ & $p=0.158$ & $p=0.359$ & $p=0.206$ & $p=0.063$ \\\hline
			- & $0.303(0.754)$ & $0.011(0.021)$ & $-0.568(0.505)$ & $0.010(0.026)$ & $0.437(0.505)$ \\
            & $p=0.688$ & $p=0.591$ & $p=0.261$ & $p=0.708$ & $p=0.387$ 
		\end{tabular}
		\caption{Point estimates, standard errors, and $p$-values of the parametric model components rounded to three decimal places; values obtained from the \texttt{smcure} package in R.}
		\label{tab:param}
	\end{table}

	Just like for the nonparametric test, the results from the semiparametric model were unambiguous (see Table~\ref{tab:breast_p_CI}): 
	for all considered covariate combinations, uncured patients with ER+ have, at level $\alpha=5\%$, a significantly larger expected mean survival time than uncured patients with ER-.
	This difference seems to grow with progressing age, especially for patients with tumour size 1.
	This is in line with the parameter estimates related to age and latency.
	The influence of the tumour size on the mean survival difference is not that obvious.
	All in all, both the asymptotic and the permutation-based method resulted in very similar outcomes.

	\begin{table}[ht]
		\centering
		\begin{tabular}{ll|c|c|cc}
			&&  point estimate & hypothesis test & \multicolumn{2}{|c}{confidence interval} \\
			$X=Z$ & method & $m_z$ & $p$-value & lower  & upper   \\ \hline 
			TS1, $-10$y & asymptotic & 15 & $0.019$ & $2$ & $27$ \\
			& permutation & 15 & $0.020$ & $2$ & $27$ \\ \hline
			TS1, mean age & asymptotic& 18 & $<0.001$ & $9$ & $27$ \\
			& permutation & 18 & $<0.001$ & $9$ & $27$ \\ \hline
			TS1, $+10$y & asymptotic & 22 & $<0.001$ & $13$ & $31$ \\
			& permutation & 22 &  $<0.001$ & $13$ & $31$  \\ \hline\hline
			TS2, $-10$y & asymptotic & 11 & $0.014$ & 2 & 19 \\
			& permutation & 11 & $0.032$ &  $0$ & $19$ \\ \hline
			TS2, mean age & asymptotic & 14 &  $0.001$ & 6 & 23 \\
			& permutation & 14 & $0.008$ & $4$ & $22$ \\ \hline
			TS2, $+10$y & asymptotic& 17 & $0.002$ & $6$ & $28$ \\
			& permutation & 17 & $0.012$ & $5$ & $28$ 
		\end{tabular}
		\caption{$p$-values for two-sided testing of equal mean survival times, $H_0: m_z = 0$, and corresponding 95\%-confidence intervals for the differences (rounded to full months), for different covariate combinations. TS is short for tumour size and $\pm10$y refers to the mean age 53.9 years plus or minus ten years.}
		\label{tab:breast_p_CI}
	\end{table}
	
	\section{Discussion and future research}
 \label{sec:disc_II}
		In this article, we focused on the comparison of survival times for the uncured subpopulations in a two-sample problem by means of the mean survival time adjusted for adjusted for potential confounders. We proposed a semiparametric estimation method for the $MST_u$ conditionally on a covariate value and developed  an asymptotic and a permutation-based approach for inference. We encountered several theoretical and computational challenges related to the permutation approach in a semi-parametric model. Moreover, based on our simulation study, we did not observe a clear advantage of the permutation method over the asymptotic one, contrarily to existing findings in the literature for the nonparametric setting. 
	
	We employed a logistic-Cox mixture cure model because it is the most widely used semiparametric cure model. However, the logistic model for the incidence could be replaced by any other parametric model that might be more suitable in specific applications. Also, the Cox model for the latency component could for example  be replaced by the accelerated failure time model. In such cases, both the estimators and the theory would need some adjustments but we expect that the same challenges would arise and  the practical performance would be similar to the current model. Another possibility would be to allow for covariates without imposing a specific model on the latency. For example, \cite{patilea2020general} focused on estimating the cure rate, while leaving the distribution of the uncured unspecified, and they obtained a nonparametric estimator for the conditional survival function of the uncured. Another nonparametric estimator was proposed in \cite{lopez2017nonparametric}, relying on the Beran estimator for the conditional survival function. It would be of interest to further extend our  method to these nonparametric settings. 

    One important extension of the present methodology would include a rectification for causal interpretations, as motivated in the discussion section of Part~I.
    For example, the combination of an inverse-probability-of-treatment weighted version of the $MST_{u,i,z}$'s with the cure rates would allow for causal insights and a more complete picture in comparisons of two groups with cure fractions.
    If the data did not originate from randomized clinical trials in the first place, also the estimator for the cure rate would require an adjustment for enabling causal interpretations.
    These considerations are clearly beyond the scope of the present paper and are thus left for future research.
 
	Apart from the permutation approach for inference, one could also consider a pooled bootstrap approach; cf.\ Section~3.7.2 of \cite{VW96}. 	
	However, even with the pooled bootstrap we would encounter the same theoretical and computational problems and we expect that the practical behavior would be similar  to the permutation approach. 
    Also, one would lose the benefit of the finite sample exactness of the permutation-based inference procedures under exchangeability.
	
	Finally, to avoid the model misspecification issue for the permutation approach in the  semiparametric setting, one could estimate a mixed model instead of a logistic-Cox model in the permutation samples, i.e., fitting the correct model of the pooled data, which is a linear combination of two logistic-Cox models. Afterwards, we could estimate the MST as usual. Then one would still need to develop the asymptotic results theoretically for the resulting estimator; these asymptotics will not be the same as in the standard logistic-Cox. In practice, we do not expect this approach to behave better since more parameters need to be estimated for each permuted sample.
	
\begin{appendix}

\setcounter{section}{0}

	\section{Proofs}
 \label{sec:app_II}

		\begin{proof}[Proof of Theorem~\ref{theo:MST_cov}]
		Let $\theta_i=(\gamma_i,\beta_i)$ denote the vector of parameters of the semiparametric mixture cure model. Consider the space $\mathcal{H}_m=\{h=(h_1,h_2)\in BV[0,\tau_{0,i}]\times\R^{p+q}: \Vert h_1\Vert_v+\Vert h_2\Vert_1\leq m\}$, where $m<\infty$ and $\Vert h_1\Vert_v$ is the absolute value of $h_1(0)$ plus the total variation of $h_1$ on the interval $[0, \tau_{0,i} ]$. From Theorem 3 in \cite{Lu2008} we have that the process 
		\begin{equation}
			\label{eqn:process_h}
			\langle n_i^{1/2}(\hat{\Lambda}_i-{\Lambda}_i),n_i^{1/2}(\hat{\theta}_i-{\theta}_i)\rangle(h)=n_i^{1/2}\int_0^{\tau_{0,i}}h_1(s)\,\dd (\hat{\Lambda}_i-{\Lambda}_i)(s)+n_i^{1/2}h^T_2(\hat{\theta}_i-{\theta}_i)
		\end{equation}
		indexed by $h\in\mathcal{H}_m$ converges weakly in $l^\infty(\mathcal{H}_m)$ to a tight Gaussian process $G_i$ in   $l^\infty(\mathcal{H}_m)$ with mean zero and covariance process
		\[
		Cov(G_i(h),G_i(h^*))=\int_0^{\tau_{0,i}}h_1(s)\sigma_{(1),i}^{-1}(h^*)(s)\dd\Lambda_i(s)+h^T_2\sigma_{(2),i}^{-1}(h^*)
		\]
		where
		\begin{equation}
			\label{eqn:sigma_1}
			\begin{split}
				&	\sigma_{(1),i}(h)(s)=\E\left[\1_{\{Y_{i1}\geq s\}}V_i(s;\theta_i,\Lambda_i)(h)g_i(s;\theta_i,\Lambda_i)e^{\beta^T_iZ_{i1}}\right]\\
				&-\E\left[\int_s^{\tau_{0,i}}\1_{\{Y_{i1}\geq u\}}V_i(s;\theta_i,\Lambda_i)(h)g_i(u;\theta_i,\Lambda_i)\{1-g_i(u;\theta_i,\Lambda_i)\}e^{2\beta^T_iZ_{i1}}\dd\Lambda_i(u)\right],
			\end{split}
		\end{equation}
		\begin{equation}
			\label{eqn:sigma_2}
			\sigma_{(2),i}(h)=\E\left[\int_0^{\tau_{0,i}}\1_{\{Y_{i1}\geq s\}}W_i(s;\theta_i,\Lambda_i)V_i(s;\theta_i,\Lambda_i)(h)g_i(s;\theta_i,\Lambda_i)e^{\beta^T_iZ_{i1}}\dd\Lambda_i(s)\right]
		\end{equation}
		and
		\begin{equation*}
			\label{def:g}
			g_i(s;\theta,\Lambda)=\frac{\phi(\gamma^TX_{i1})\exp\left(-\Lambda(s) \exp\left(\beta^T Z_{i1}\right)\right)}{1-\phi(\gamma^TX_{i1})+\phi(\gamma^TX_{i1})\exp\left(-\Lambda(s) \exp\left(\beta^T Z_{i1}\right)\right)},
		\end{equation*}
		\[
		V_i(s;\theta_i,\Lambda_i)(h)=h_1(s)-\left\{1-g_i(s;\theta_i,\Lambda_i)\right\}e^{\beta^T_iZ_{i1}}\int_0^sh_1(u)\dd\Lambda_i(u)+h^T_2W_i(s;\theta_i,\Lambda_i)
		\]
		\[
		W_i(s;\theta_i,\Lambda_i)=\left(\left\{1-g_i(s;\theta_i,\Lambda_i)\right\}X^T_{i1},\left[1-\left\{1-g_i(s;\theta_i,\Lambda_i)\right\}e^{\beta^T_iZ_{i1}}\Lambda_i(s)\right]Z^T_{i1}\right)^T.
		\]
		Using this result we first obtain weak convergence of the process $\sqrt{n}\{\hat{S}_{u,i}(t|z)-S_{u,i}(t|z)\}$ for fixed $z\in\mathcal{Z}$. By definition and a series of Taylor expansions we have
		\[
		\begin{aligned}
			&\sqrt{n_i}\{\hat{S}_{u,i}(t|z)-S_{u,i}(t|z)\\
			&=\sqrt{n_i}\left\{\exp\left(-\hat{\Lambda}_i(t)e^{\hat\beta^T_iz}\right)-\exp\left(-{\Lambda}_i(t)e^{\beta^T_iz}\right)\right\}\\
			&=-\sqrt{n_i}\left\{\hat{\Lambda}_i(t)e^{\hat\beta^T_iz}-{\Lambda}_i(t)e^{\beta^T_iz}\right\}\exp\left(-{\Lambda}_i(t)e^{\beta^T_iz}\right)+R_1\\
			&=-\sqrt{n_i}\left\{\hat{\Lambda}_i(t)-{\Lambda}_i(t)\right\}e^{\beta^T_iz} \exp\left(-{\Lambda}_i(t)e^{\beta^T_iz}\right)\\
			&\quad- \sqrt{n_i}\left\{e^{\hat\beta^T_iz}-e^{\beta^T_iz}\right\}\Lambda_i(t)\exp\left(-{\Lambda}_i(t)e^{\beta^T_iz}\right)+R_1+R_2\\
			&=-\sqrt{n_i}\left\{\hat{\Lambda}_i(t)-{\Lambda}_i(t)\right\}e^{\beta^T_iz} \exp\left(-{\Lambda}_i(t)e^{\beta^T_iz}\right)\\
			&\quad- \sqrt{n_i}\left\{{\hat\beta_i}-{\beta_i}\right\}^Tze^{\beta^T_iz}\Lambda_i(t)\exp\left(-{\Lambda}_i(t)e^{\beta^T_iz}\right)+R_1+R_2+R_3\\
		\end{aligned}
		\]
		where the remainder terms $R_1,R_2,R_3$ converge uniformly to zero in probability because of the boundedness of  $\Lambda_i$, $\beta_i$ and Theorems 2-3 in \cite{Lu2008}. Note that, as in assumption (A1), we denote by $S_{u,i}(\tau_{0,i}|z)$ the left limit $S_{u,i}(\tau_{0,i}-|z)$ so that $S_{u,i}(\cdot|z)$ is a continuous function, bounded from below away from zero. This modification does not influence the value of $MST_{u,i,z}$.  
		
		\noindent Consider the functions $h\in\mathcal{H}_m$ of the form
		\begin{equation}
			\label{eqn:h}
			\begin{aligned}
				h_{t,z}&=(h_{1;t,z},h_{2;t,z})\\
				&=\left(\1_{[0,t]}(\cdot)e^{\beta^T_iz} \exp\left(-{\Lambda}_i(t)e^{\beta^T_iz}\right),\left(\mathbf{0}_p,ze^{\beta^T_iz}\Lambda_i(t)\exp\left(-{\Lambda}_i(t)e^{\beta^T_iz}\right)\right)\right),
			\end{aligned}
		\end{equation}
		where $\mathbf{0}_p$ denotes a zero vector in $\R^p$ since we are not interested in the $\gamma$ component. For an appropriate choice of $m$ and any $t\in[0,\tau_{0,i}]$, such functions belong to $\mathcal{H}_m$ because of assumptions (A1)-(A2). For these functions we have
		\[
		\sqrt{n_i}\{\hat{S}_{u,i}(t|z)-S_{u,i}(t|z)\}=\langle n_i^{1/2}(\hat{\Lambda}_i-{\Lambda}_i),n^{1/2}(\hat{\theta}_i-{\theta}_i)\rangle(h_{t,z})+o_P(1)
		\]
		from which we conclude that the process $\sqrt{n}\{\hat{S}_{u,i}(t|z)-S_{u,i}(t|z)$ converges weakly in $D^\infty([0,\tau_{0,i}])$ to a mean zero Gaussian process $G_{i,z}^*$ with covariance 
		\begin{equation}
			\label{eqn:rho}
			\begin{aligned}
				\rho_{i,z}(s,t)&=Cov(G_{i,z}^*(t),G_{i,z}^*(t^*))\\
				&=\int_0^{\tau_{0,i}}h_{1;t,z}(s)\sigma_{(1),i}^{-1}(h_{t^*,z})(s)\dd\Lambda_i(s)+h^T_{2;t,z}\sigma_{(2),i}^{-1}(h_{t^*,z}),
			\end{aligned}
		\end{equation}
		where $h_{t,z}$ is as in \eqref{eqn:h} and $\sigma_{(1),i},\sigma_{(2),i}$ as in \eqref{eqn:sigma_1},\eqref{eqn:sigma_2}.
		Next we use the delta method to obtain the asymptotic distribution of the conditional mean survival time. We have 
		\begin{equation}
			\label{eqn:hat_E-E}
			\begin{aligned}
				\sqrt{n_i}(\widehat{MST}_{u,i,z}-MST_{u,i,z})&=\sqrt{n_i}\int_0^{\tau_{0,i}}\{\hat{S}_{u,i}(t|z)-S_{u,i}(t|z)\}\dd u\\
				&\qquad-\sqrt{n_i}\{\tau_{0,i}-Y_{i,(m_1)}\}\hat{S}_{u,i}(\tau_{0,i}).		\end{aligned}
		\end{equation}
		The first term in the previous equation is equal to $\sqrt{n_i}\{\psi(\hat{S}_{u,i})-\psi(S_{u,i})\}$ where
		\[
		\psi:{D}[0,\tau_{0,i}]\to \R\qquad \psi(\xi)=\int_0^{\tau_{0,i}}\xi(u)\,\dd u
		\]
		The function $\psi$ is Hadamard-differentiable with derivative
		\[
		\dd\psi_\xi\cdot h=\int_0^{\tau_{0,i}}h(u)\dd u.
		\]
		By Theorem 3.9.4 in \cite{VW96} it follows that $\sqrt{n_i}\{\psi(\hat{S}_{u,i})-\psi(S_{u,i})\}$ converges weakly to 
		\[
		N=\int_0^{\tau_{0,i}}G_{i,z}^*(u)\dd u
		\]
		which is normal distributed with mean zero and variance 
		\begin{equation}
			\label{eqn:sigma_cov}
			\sigma^2_{i,z}=\int_0^{\tau_{0,i}}\int_0^{\tau_{0,i}}\rho_{i,z}(s,t)\,\dd s\, \dd t,
		\end{equation}
		where $\rho_{i,z}$ is defined in \eqref{eqn:rho}.
		Next we show that the second term on the right hand side of \eqref{eqn:hat_E-E} converges to zero in probability, from which we can conclude that the asymptotic distribution of $\sqrt{n_i}(\widehat{MST}_{u,i,z}-MST_{u,i,z})$ is determined by that of the first term. 
		Since $\hat{S}_{u,i}(\tau_{0,i})$ converges to ${S}_{u,i}(\tau_{0,i})$, it is sufficient to show that $\sqrt{n_i}\{\tau_{0,i}-Y_{i,(m_1)}\}=o_P(1)$.  For any $\delta>0$ we have
		\[
		\begin{aligned}
			&\p(\sqrt{n_i}\{\tau_{0,i}-Y_{i,(m_1)}\}>\delta)\\
			&=\p\left(Y_{i,(m_1)}<\tau_{0,i}-\frac{\delta}{\sqrt{n_i}}\right)\\
			&=\p\left(\Delta_{i1}Y_{i1}<\tau_{0,i}-\frac{\delta}{\sqrt{n_i}}\right)^n\\
			&=\left[1-\p\left(\Delta_{i1}Y_{i1}\geq\tau_{0,i}-\frac{\delta}{\sqrt{n_i}}\right)\right]^n\\
			&=\left[1-\int_{\mathcal{Z}\times\mathcal{X}}\int_{\tau_{0,i}-\frac{\delta}{\sqrt{n_i}}}^{\tau_{0,i}}\p\left(C_{i1}\geq t|x,z\right)\dd F_{T_{i1}|X_{i1},Z_{i1}}(t|x,z)\,\dd F_{(Z_{i1},X_{i1})}(z,x)\right]^{n_i}\\
			&\leq\left[1-\int_{\mathcal{Z}\times\mathcal{X}}\p\left(C_{i1}> \tau_0|x,z\right)\phi(\gamma^T_ix)\p(T_{i1}=\tau_{0,i}|T_{i1}<\infty, z)\,\dd F_{(Z_{i1},X_{i1})}(z,x)\right]^{n_i}.
		\end{aligned}
		\]
		Because of assumptions (I1),(I3) and (A3), for some $K>0$ we have
		\[
		\begin{aligned}
			&\p(\sqrt{n_i}\{\tau_{0,i}-Y_{i,(m_1)}\}>\delta)\\
			&\leq\left[1-K\epsilon\right]^{n_i}\to 0.
		\end{aligned}
		\]
		This concludes the proof.
	\end{proof}
	
	In order to prove our main Theorem~\ref{thm:perm_cov}, we first need some preliminary results, which are provided in the following lemmas. In what follows, we assume that conditions (I1)-(I4) and (A1)-(A4) are satisfied.
		
		Denote by $\hat{\Lambda}^\pi_{n_i,i}$ and $\hat{\theta}_{n_i,i}^\pi$ the estimators of $\Lambda$ and $\theta=(\gamma,\beta)$ obtained by fitting a logistic-Cox model to the i-th permuted sample, which will be denoted for notational convenience as $(\Delta_{i1}^\pi,Y_{i1}^\pi,X_{i1}^\pi,Z_{i1}^\pi),\dots,$ $(\Delta_{in_i}^\pi,Y_{in_i}^\pi,X_{in_i}^\pi,Z_{in_i}^\pi) $.
  Let $\bar{\Lambda}_{n_1+n_2}$ and  $\bar{\theta}_{n_1+n_2}$ denote the estimators of $\Lambda$ and $\theta$ based on the pooled sample ${(\Delta_{1},{Y}_{1},{X}_{1},{Z}_{1}),\dots,(\Delta_{n_1+n_2},{Y}_{n_1+n_2},{X}_{n_1+n_2},{Z}_{n_1+n_2})}.$
		Note that the true distribution of the pooled sample is $\bar{\p}=\kappa\p_1+(1-\kappa)\p_2$, where $\p_i$ denotes the distribution of the $i$th sample. In particular, $\p$ does not correspond to a logistic-Cox model. Let $Q_{\Lambda,\theta}$ be the corresponding distribution of a logistic-Cox model with parameters $(\Lambda,\theta)$ and corresponding log-likelihood 
		\[
		\begin{aligned}
			l(\delta,y,x,z;\Lambda,\theta)&=\delta\left\{\log\phi(\gamma^Tx)+\log f_u(y|z;\Lambda,\beta) \right\}\\
			&\qquad+(1-\delta)\log\left\{1-\phi(\gamma^Tx)+\phi(\gamma^Tx)S_u(y|z;\Lambda,\beta) \right\}
		\end{aligned}
		\]
		By assumption (I2), the event times on the pooled sample happen on $[0,\bar\tau_0]$, where $\bar\tau_0=\max\{\tau_{0,1},\tau_{0,2}\}$. Hence $S_u(t|z;\Lambda,\beta)=0$ for $t\geq\bar\tau_0$ which corresponds to $\Lambda$ being defined on $[0,\bar\tau_0)$. In addition, $\bar{\p}({\Delta}=1, {Y}=\bar\tau_0)>0$ by assumptions (A3) and (I3). Hence, we can restrict on distributions $Q_{\Lambda,\theta}$ that have a positive mass at $\bar\tau_0$, meaning that $\lim_{t\to\bar\tau_0}\Lambda(t)<+\infty$ and we can denote the limit by $\Lambda(\bar\tau_0)$. 
		To reflect the existence of the jump in the likelihood, for the terms with $\Delta=1$ and ${Y}=\bar\tau$ we have $f_u(\bar\tau|z;\Lambda,\beta)=S_u(\bar\tau-|z;\Lambda,\beta)=\exp(-\Lambda(\bar\tau_0)e^{\beta^Tz})$, instead of the usual expression $f_u(t|z;\Lambda,\beta)={\lambda_u(t|z;\beta)}S_u(\bar\tau-|z;\Lambda,\beta). $ Here $\lambda_u$ denotes the hazard function corresponding to $S_u$ and the baseline hazard function corresponding to $\Lambda$ will be denoted by $\lambda$. 
	Define
\begin{equation}
	\label{eqn:bar_parameters}
(\bar\Lambda,\bar\theta)=\argmax_{\Lambda,\theta}\E_{\bar\p}[l(\Delta,{Y},{X},{Z};\Lambda,\theta)]=\argmin_{\Lambda,\theta}\mathrm{KL}(\bar\p|Q_{\Lambda,\theta})
\end{equation}
where $KL(\cdot|\cdot)$ denotes the Kullback-Leibler divergence between two distributions.
\begin{lemma}
The argmax defined in \eqref{eqn:bar_parameters} exists. 
\end{lemma}
\begin{proof}
We show that the argmax can be restricted on a bounded set, from which the existence follows because of continuity. In the three steps below we deal consequently with $\beta$, $\Lambda$  and $\gamma$. 

\textit{Step 1}. First we show that for any $K>0$ there exists $\bar{c}>0$ such that for any $c\geq\bar{c}$ we have $\inf_{\tilde{\beta}\in S^{q-1}}\bar\p(c\,|\tilde{\beta}^T{Z}|>K)>0$, where $S^{q-1}$ is the unit circle in $\R^q$.
	Suppose by contradiction that there exists $K$ such that for any $c$ we have  $\inf_{\tilde{\beta}\in S^{q-1}}\bar\p(c\,|\tilde{\beta}^T{Z}|>K)=0$. Note that the infimum is actually a minimum because $S^{q-1}$ is compact and the function is continuous. Hence, it means that for any $c$ there exists $\tilde{\beta}\in S^{q-1}$ for which  $\bar\p(c\,|\tilde{\beta}^T{Z}|\leq K)=1$. Equivalently, for any $\epsilon>0$, there exists $\tilde{\beta}\in S^{q-1}$ for which  $\bar\p(|\tilde{\beta}^T{Z}|\leq \epsilon)=1$. The closed subsets of $S^{q-1}$ defined by $B_m=\{\tilde{\beta}\in S^{q-1}\,\big|\,\bar\p(|\tilde{\beta}^T{Z}|\leq \frac1m)=1\}$ are non-empty for all $m$ and $B_m\downarrow B=\cap_m B_m$. $B$ cannot be empty because then $(B_m^c)_m$ form an open covering of the compact $S^{q-1}$ and there would exists a finite sub-covering, which is impossible since all $B_m$ are non-empty. It follows that $B$ is not empty, which is equivalent to saying that there exists $\tilde{\beta}\in S^{q-1}$ for which  $\bar\p(|\tilde{\beta}^T{Z}|=0)=1$. This contradicts the assumption that $Var({Z})$ has full rank. \\
	Next, let $\eta=\frac12\Lambda({\bar\tau_0})\inf_z\bar\p({Y}=\bar\tau_0|{\Delta}=1,{Z}=z)$ and choose $K$ such that $x\leq \eta e^x$ for all $x\geq K$. Let $\beta=c{\tilde\beta}$ with $\tilde{\beta}\in S^{q-1}$, $c>\bar{c}$. We will show that, as $c$ increases, the expectation in \eqref{eqn:bar_parameters} becomes arbitrarily small. For fixed $\gamma$ and $\Lambda$, we can write
	\[
	\begin{split}
		&\E_{\bar\p}[l(\Delta,{Y},{X},{Z};\Lambda,\theta)]\\
		&=\E_{\bar\p}[{\Delta}\beta^T{Z}{\1_{\{{Y}<\bar\tau_0\}}}-{\Delta}\Lambda({Y}) e^{\beta^T{Z}}]+R_1\\
		&=\E_{\bar\p}\left[\left\{{\Delta}c\tilde\beta^T{Z}{\1_{\{{Y}<\bar\tau_0\}}}-{\Delta}\Lambda({Y}) e^{c\tilde\beta^T{Z}}\right\}\1_{\{\tilde\beta^T{Z}>0\}}\1_{\{c\,|\tilde{\beta}^T{Z}|>K\}}\right]\\
		&+\E_{\bar\p}\left[\left\{{\Delta}c\tilde\beta^T{Z}{\1_{\{{Y}<\bar\tau_0\}}}-{\Delta}\Lambda({Y}) e^{c\tilde\beta^T{Z}}\right\}\1_{\{\tilde\beta^T{Z}<0\}}\1_{\{c\,|\tilde{\beta}^T{Z}|>K\}}\right]+R_2,
	\end{split}
\]
	where $R_1$ and $R_2$ denote terms that are bounded in absolute value. Using  $\E_{\bar\p}[\Lambda({Y})|{Z},{\Delta}=1]>\Lambda({\bar\tau_0})\inf_z\bar\p({Y}=\bar\tau_0|{\Delta}=1,{Z}=z)=2\eta>0$, we obtain the following bound 
\[
\begin{split}
	\E_{\bar\p}[l(\Delta,{Y},{X},{Z};\Lambda,\theta)]&\leq \E_{\bar\p}\left[\left\{{\Delta}c\tilde\beta^T{Z}-2{\Delta}\eta e^{c\tilde\beta^T{Z}}\right\}\1_{\{\tilde\beta^T{Z}>0\}}\1_{\{c\,|\tilde{\beta}^T{Z}|>K\}}\right]\\
		&+c\E_{\bar\p}\left[{\Delta}\tilde\beta^T{Z}{\1_{\{{Y}<\bar\tau_0\}}}\1_{\{\tilde\beta^T{Z}<0\}}\1_{\{c\,|\tilde{\beta}^T{Z}|>K\}}\right]+R_2\\
		&\leq -\eta\E_{\bar\p}\left[{\Delta} e^{c\tilde\beta^T{Z}}\1_{\{\tilde\beta^T{Z}>0\}}\1_{\{c\,|\tilde{\beta}^T{Z}|>K\}}\right]\\
		&-c\E_{\bar\p}\left[{\Delta}|\tilde\beta^T{Z}|{\1_{\{{Y}<\bar\tau_0\}}}\1_{\{\tilde\beta^T{Z}<0\}}\1_{\{c\,|\tilde{\beta}^T{Z}|>K\}}\right]+R_2\\
		&\leq -\eta\E_{\bar\p}\left[{\Delta} e^{c\tilde\beta^T{Z}}\1_{\{\tilde\beta^T{Z}>0\}}\1_{\{\bar{c}\,|\tilde{\beta}^T{Z}|>K\}}\right]\\
		&-c\E_{\bar\p}\left[{\Delta}|\tilde\beta^T{Z}|{\1_{\{{Y}<\bar\tau_0\}}}\1_{\{\tilde\beta^T{Z}<0\}}\1_{\{\bar{c}\,|\tilde{\beta}^T{Z}|>K\}}\right]+R_2\\
		\end{split}
\]
This further leads to 
\[
\begin{split}
&	\E_{\bar\p}[l(\Delta,{Y},{X},{Z};\Lambda,\theta)]\\	&\leq -\eta e^{cK/\bar{c}}\bar\p \left({\Delta}=1, \tilde\beta^T{Z}>0,\bar{c}\,|\tilde{\beta}^T{Z}|>K\right)\\
		&-c\frac{K}{\bar{c}}\bar\p\left({\Delta}=1,{Y}<\bar\tau_0,\tilde\beta^T{Z}<0,\bar{c}\,|\tilde{\beta}^TZ|>K\right)+R_2\\
		&\leq -c\frac{K}{\bar{c}}\bar\p\left({\Delta}=1,{Y}<\bar\tau_0,\bar{c}\,|\tilde{\beta}^T{Z}|>K\right)+R_2\\
		&\leq -c\frac{K}{\bar{c}}\inf_z\bar{\p}(\Delta=1,{Y}<\bar\tau_0|{Z}=z)\inf_{\tilde{\beta}\in S^{q-1}}\bar{\p}\left(\bar{c}\,|\tilde{\beta}^T{Z}|>K\right)+R_2
	\end{split}
	\]
Since both infimums are strictly positive, $\E_{\bar\p}[l(\Delta,{Y},{X},{Z};\Lambda,\theta)]$ can be made arbitrarily small for $c$ sufficiently large (and how large $c$ should be does not depend  on $\tilde{\beta}$). Hence, we can restrict the argmax on a bounded set for $\beta$.

\textit{Step 2}.  Next we show that, there exists $M>0$ such that it suffices to search for the maximizer among $\Lambda$ that are bounded by $M$. Let  $\Lambda$ be such that $\Lambda(\bar{\tau}_0)>M$. We can construct $\tilde{\Lambda}(t)=c\Lambda(t)$ with  $c=M/\Lambda(\bar{\tau}_0)\in(0,1)$. We have $\tilde{\Lambda}(\bar{\tau}_0)=M$ and $\tilde{\lambda}=c\lambda$. We show that 
\[
\E_{\bar\p}[l(\Delta,{Y},{X},{Z};\Lambda,\theta)]< \E_{\bar\p}[l(\Delta,{Y},{X},{Z};\tilde\Lambda,\theta)].
\] 
Indeed we have 
\[
\begin{split}
	&\E_{\bar\p}[l(\Delta,{Y},{X},{Z};\Lambda,\theta)]- \E_{\bar\p}[l(\Delta,{Y},{X},{Z};\tilde\Lambda,\theta)]\\
	&=\E_{\bar\p}\left[-{\Delta}\log c{\1_{\{{Y}<\bar\tau_0\}}} -{\Delta}\Lambda({Y})e^{\beta^T{Z}}+{\Delta}c\Lambda({Y})e^{\beta^T{Z}}\right.\\
	&\qquad\left.+(1-{\Delta})\log\frac{1-\phi(\gamma^T{X})+\phi(\gamma^T{X})S_u({Y}|{Z};\Lambda,\beta)}{1-\phi(\gamma^T{X})+\phi(\gamma^T{X})S_u({Y}|{Z};\tilde\Lambda,\beta)}\right]
\end{split}
\]		 
Since $	S_u({Y}|{Z};\tilde\Lambda,\beta)>S_u({Y}|{Z};\Lambda,\beta)$ the ratio is smaller than 1 and as a result the $(1-{\Delta})$ term in the expectation is negative. Hence
\[
\begin{split}
	&\E_{\bar\p}[l(\Delta,{Y},{X},{Z};\Lambda,\theta)]- \E_{\bar\p}[l(\Delta,{Y},{X},{Z};\tilde\Lambda,\theta)]\\
	&<\E_{\bar\p}\left[-{\Delta}\log c{\1_{\{{Y}<\bar\tau_0\}}} -{\Delta}\Lambda({Y})e^{\beta^T{Z}}+{\Delta}c\Lambda({Y})e^{\beta^T{Z}}\right]\\
	&=-\bar\p({\Delta}=1,{{Y}<\bar\tau_0})\log c-(1-c)\E_{\bar\p}\left[{\Delta}\Lambda({Y})e^{\beta^T{Z}}\right]\\
	&\leq-\bar\p({\Delta}=1,{{Y}<\bar\tau_0}) \log c-(1-c)\E_{\bar\p}\left[{\Delta}\Lambda({Y})e^{\beta^T{Z}}\1_{\{{Y}=\bar\tau_0\}}\right]\\
	&=- \bar\p({\Delta}=1,{{Y}<\bar\tau_0})\log c-(1-c)\Lambda(\bar\tau_0)\E_{\bar\p}\left[{\Delta}e^{\beta^T{Z}}\1_{\{{Y}=\bar\tau_0\}}\right].	
\end{split}
\]
Since we are restricting $\beta$ on a compact and $Z$ is assumed to have bounded support, there exist $c_2>0$ such that $	e^{\beta^T{Z}}>c_2$ a.s.. It follows that
\[
\begin{split}
&\E_{\bar\p}[l(\Delta,{Y},{X},{Z};\Lambda,\theta)]- \E_{\bar\p}[l(\Delta,{Y},{X},{Z};\tilde\Lambda,\theta)]\\
	&<-\bar\p({\Delta}=1,{{Y}<\bar\tau_0})\log c-(1-c)\Lambda(\bar\tau_0)c_2\bar\p({\Delta}=1,{Y}=\bar\tau_0)\\
	&=\bar\p({\Delta}=1,{{Y}<\bar\tau_0})(\log\Lambda(\bar\tau_0)-\log M) -(\Lambda(\bar\tau_0)-M)c_2\bar\p({\Delta}=1,{Y}=\bar\tau_0)\\
	&<\bar\p({\Delta}=1,{{Y}<\bar\tau_0})(\Lambda(\bar\tau_0)-M)\frac{1}{M} -(\Lambda(\bar\tau_0)-M)c_2\bar\p({\Delta}=1,{Y}=\bar\tau_0)\\
	&=(\Lambda(\bar\tau_0)-M)\left\{\frac{1}{M}\bar\p({\Delta}=1,{{Y}<\bar\tau_0}) -c_2\bar\p({\Delta}=1,{Y}=\bar\tau_0)\right\}<0
\end{split}
\]	 
for large enough $M$ since  $	\bar\p({\Delta}=1,{Y}=\bar\tau)>0$ by assumption. Hence we conclude that, there exists $M$ such that it is sufficient to search for the maximizer among $\Lambda$'s bounded by $M$.

\textit{Step 3.} We can also restrict the argmax on a bounded set for $\gamma$ because as $\Vert\gamma\Vert\to \infty$, for fixed values of $\beta$ and $\Lambda$, the expectation converges to $-\infty$. Indeed we have 
\[
\begin{split}
	&\E_{\bar\p}[l(\Delta,{Y},{X},{Z};\Lambda,\theta)]\\
	&=\E_{\bar\p}\left[{\Delta}\log\phi(\gamma^T{X})\1_{\{\gamma^T{X}>0\}}\right]+\E_{\bar\p}\left[{\Delta}\log\phi(\gamma^T{X})\1_{\{\gamma^T{X}\leq0\}}\right]+R_3,
\end{split}
\]
where $R_3$ denotes terms bounded in absolute value.
The first term is bounded and using the same reasoning as with $\beta$ it can be shown that the second term converges to $-\infty$. 

We conclude that we can restrict the argmax on a bounded set, from which the existance of the argmax follows as the criteria is continuous with respect to the parameters. 
\end{proof}
	
 In what follows, we assume that the maximizer $(\bar\Lambda,\bar\theta)$ is unique. It will also be useful to characterize it as the solution of the score equation defined similarly to \cite{Lu2008}. As in the proof of Theorem \ref{theo:MST_cov}, consider $\mathcal{H}_m=\{h=(h_1,h_2)\in BV[0,\bar\tau_{0}]\times\R^{p+q}: \Vert h_1\Vert_v+\Vert h_2\Vert_1\leq m\}$, where $m<\infty$ and $\Vert h_1\Vert_v$ is the absolute value of $h_1(0)$ plus the total variation of $h_1$ on the interval $[0, \bar\tau_{0} ]$. Define the functions
 \[
		\begin{split}
			&\psi_{(\Lambda,\theta),h}(\delta,y,x,z)
			=\delta \left[h_1(y)+h_{21}^Tx+h^T_{22}z\right]-\left\{\phi(\gamma^Tx)-(1-\delta)g(y,\Lambda,\theta)\right\}h_{21}^Tx\\
			&\qquad\qquad\qquad- \left\{\delta+(1-\delta)g(y,\Lambda,\theta)\right\}\left\{e^{\beta^Tz}\int_0^{y}h_1(s)\dd\Lambda(s)+e^{\beta^Tz}\Lambda(y)h^T_{22}z
			\right\},
		\end{split}
		\]
  where 
  \[
g(t,\Lambda,\theta)=\frac{\phi(\gamma^Tx)\exp\left(-\Lambda(t) \exp\left(\beta^T z\right)\right)}{1-\phi(\gamma^Tx)+\phi(\gamma^Tx)\exp\left(-\Lambda(t) \exp\left(\beta^Tz\right)\right)}.
  \]
We will denote by  $\p^\pi_{n_i,i}\psi_{(\Lambda,\theta),h}$ the score function for the i-th permuted sample 
		\[
		\begin{split}
			\p^\pi_{n_i,i}\psi_{(\Lambda,\theta),h}
			&=\frac{1}{n_i}\sum_{j=1}^{n_i} \Delta^\pi_{ij} \left[h_1({Y}_{j}^\pi)+h_{21}^TX_{ij}^\pi+h^T_{22}Z_{ij}^\pi\right]\\
			&\quad-\frac{1}{n_i}\sum_{j=1}^{n_i}\left\{\phi(\gamma^TX_{ij}^\pi)-(1-\Delta_{ij}^\pi)g_{ij}^\pi(Y_{ij}^\pi,\Lambda,\theta)\right\}h_{21}^TX_{ij}^\pi\\
			&\quad-\frac{1}{n_i}\sum_{j=1}^{n_i} \left\{\Delta_{ij}^\pi+(1-\Delta_{ij}^\pi)g_{ij}^\pi(Y_{ij}^\pi,\Lambda,\theta)\right\}\\
			&\quad\qquad\qquad\quad\times\left\{e^{\beta^TZ_{ij}^\pi}\int_0^{Y_{ij}^\pi}h_1(s)\dd\Lambda(s)+e^{\beta^TZ_{ij}^\pi}\Lambda(Y_{ij}^\pi)h^T_{22}Z_{ij}^\pi
			\right\},
		\end{split}
		\]
		where $h=(h_1,h_2)=(h_1,h_{21},h_{22})\in\mathcal{H}_m$ and
		\[
		g_{ij}^\pi(t,\Lambda,\theta)=\frac{\phi(\gamma^TX_{ij}^\pi)\exp\left(-\Lambda(t) \exp\left(\beta^T Z_{ij}^\pi\right)\right)}{1-\phi(\gamma^TX_{ij}^\pi)+\phi(\gamma^TX_{ij}^\pi)\exp\left(-\Lambda(t) \exp\left(\beta^TZ_{ij}^\pi\right)\right)}.
		\]
		Similarly, $\p_{n_1+n_2}\psi_{(\Lambda,\theta),h}$ and  $\bar\p\psi_{(\Lambda,\theta),h}$  are defined using the empirical distribution of the pooled sample or the true distribution of the pooled sample $\bar{\p}=\kappa\p_1+(1-\kappa)\p_2$ respectively.   
		By definition  $(\bar{\Lambda},\bar{\theta})$ is the solution of $ \bar\p\psi_{(\Lambda,\theta),h}\overset{!}{=}0$. 
		
\begin{lemma}
\label{lemma:consistency_pooled}
Assume the maximizer $ (\bar\Lambda,\bar\theta)$ defined in \eqref{eqn:bar_parameters} is unique. The pooled maximum likelihood estimator $(\bar\Lambda_{n_1+n_2},\bar\theta_{n_1+n_2})$ is a (weakly) consistent estimator of $(\bar\Lambda,\bar\theta)$.
\end{lemma}
\begin{proof}
		We will pursue similar ideas as \cite{Lu2008} in the proofs of his Lemma~2 and Theorem~2.
		Comparing the arguments in \cite{Lu2008} that lead to the maximum likelihood estimator in Display~(12) of that paper,
		it is evident that the pooled estimator $\bar\Lambda_{n_1+n_2}$ must exhibit a similar structure.
		In particular,
		\[
			\label{repr:mle_pooled}
			\bar\Lambda_{n_1+n_2}(t) = \int_0^t \frac{d  N (s)}{ \sum_{j=1}^{n_1+n_2}  R_j(s)\exp(\bar \beta_{n_1+n_2}^T Z_j) \{ \Delta_j + (1- \Delta_j) \bar{g}_j( Y_j; \bar \theta_{n_1+n_2}, \bar \Lambda_{n_1+n_2})\}},
		\]
		where $ N = N_1 + N_2$ is the pooled counting process, $ R_j(s) = \1_{\{  Y_j \geq s\}}$ denotes the at-risk process of the $j$-th pooled individual and 
  $$\bar g_i(t; \Lambda, \theta) = \frac{\phi( \gamma^T  X_i) \exp(- \Lambda(t) \exp(\beta^T  Z_i))}{1 - \phi( \gamma^T  X_i) + \phi( \gamma^T  X_i) \exp(- \Lambda(t) \exp(\beta^T  Z_i))}.$$
  		Similarly, we define 
		\[
			 \tilde\Lambda_{n_1+n_2}(t) = \int_0^t \frac{d  N (s)}{ \sum_{j=1}^{n_1+n_2}  R_j(s)\exp(\bar\beta^T Z_j) \{ \Delta_j + (1- \Delta_j) \bar{g}_j( Y_j; \bar \theta, \bar \Lambda)\}}.
		\]
		We have already noticed that $\sup_{n_1,n_2} \tilde \Lambda_{n_1+n_2}(\tau) < \infty$ a.s.; 
		also, following the lines of Lemma~2 (ii) in \cite{Lu2008}, there exists a non-negative and integrable function $\eta: [0, \bar\tau_0] \to (0, \infty)$ bounded away from 0 such that, for each $\omega \in \Omega$, there exists a subsequence $(n_{1,k_1}(\omega)+n_{2,k_2}(\omega))$ such that $\sup_{t \in (0,\bar \tau_0]} | \frac{\dd \bar \Lambda_{n_{1,k_1}+n_{2,k_2}}}{\dd \tilde \Lambda_{n_{1,k_1}+n_{2,k_2}}}(t) - \eta(t) | \to 0$.
		For notational convenience, we will write $\bar n_k = n_{1,k_1}+n_{2,k_2}$ from now on.
		Arguing for a fixed $\omega$ and along subsequences had similarly been done by \cite{murphy94} and \cite{scharfstein98} based on Helly's theorem.
		
		To prove the desired consistency of the pooled estimators, we will show that the difference of the log-likelihoods, say $\bar \ell_{\bar n_k} (\bar \Lambda_{\bar n_k}, \bar \theta_{\bar n_k})$ and $ \bar \ell_{\bar n_k}(\tilde \Lambda_{\bar n_k}, \bar \theta) $ converges to zero.
		Clearly,
		\begin{align*}
			0 \leq & \ \bar \ell_{\bar n_k} (\bar \Lambda_{\bar n_k}, \bar \theta_{\bar n_k}) - \bar \ell_{\bar n_k}(\tilde \Lambda_{\bar n_k}, \bar \theta) \\ 
			= &  \frac1{\bar n_k} \sum_{i=1}^{\bar n_k} \Big[  \Delta_{i} \log{\frac{\bar g_i( Y_i; \bar \Lambda_{\bar n_k}, \bar \theta_{\bar n_k})}{\bar g_i( Y_i, \tilde \Lambda_{\bar n_k}, \bar \theta)}} +  \Delta_i \log \frac{\Delta \bar \Lambda_{\bar n_k}( Y_i)}{\Delta \tilde \Lambda_{\bar n_k}( Y_i)} + \Delta_i(\bar \beta_{n_k} - \bar \beta)^T  Z_i \\
   &\qquad\qquad+\log \frac{\bar S_i( Y_i; \bar \Lambda_{\bar n_k}, \bar \theta_{\bar n_k})}{\bar S_i( Y_i; \tilde \Lambda_{\bar n_k}, \bar \theta)} \Big]
			\\
			= &  \frac1{\bar n_k} \sum_{i=1}^{\bar n_k} \Big[  \Delta_{i} \log\frac{\phi(\bar \gamma^T_{n_k} X_i)}{\phi(\bar \gamma^T X_i)}  
			-  \Delta_i \bar \Lambda_{n_k} ( Y_i)\exp(\bar \beta^T_{n_k} X_i) 
			+  \Delta_i \tilde \Lambda_{n_k} ( Y_i)\exp(\bar \beta^T X_i) \\
			& \qquad +  \Delta_i \log \frac{\Delta \bar \Lambda_{\bar n_k}( Y_i)}{\Delta \tilde \Lambda_{\bar n_k}( Y_i)} + \Delta_i(\bar \beta_{n_k} - \bar \beta)^T  Z_i + (1- \Delta_i)\log \frac{\bar S_i( Y_i; \bar \Lambda_{\bar n_k}, \bar \theta_{\bar n_k})}{\bar S_i( Y_i; \tilde \Lambda_{\bar n_k}, \bar \theta)} \Big]
		\end{align*} 
		where 
		and $\bar S_i(t; \Lambda,  \theta) = 1 - \phi( \gamma^T  X_i) + \phi( \gamma^T  X_i) \exp(- \Lambda(t) \exp(\beta^T  Z_i))$.
		
		The space of bounded, increasing functions with discontinuities only at $\tau_{0,1}$ and $\tau_{0,2}$ is separable with respect to the supremum norm.
		Also, Euclidean spaces are separable.
		Denote by $(\Lambda_l,\theta_l)_{l\in\mathbb{N}}$ a countable subset that is dense in the product of the just-described spaces.
		For each $l \in \mathbb{N}$, by the strong law of large numbers,
		\[
  \begin{split}
      	\frac1{\bar n_k} \sum_{i=1}^{\bar n_k} &\Big[  \Delta_{i} \log \phi( \gamma_l^{T}  X_i)- \Delta_i\Lambda_l(t)\exp(\beta^{T}_l  Z_i)+  \Delta_i \log \dd\Lambda_l( Y_i)+ \Delta_i\beta_l^T  Z_i \\
       &\quad+(1- \Delta_i)\log \bar S_i( Y_i; \Lambda_l,  \theta_l) \Big]
  \end{split}
		\]
		converges a.s.\ to its expectation, i.e., for all $\omega \in \Omega_l$ with $\bar{\mathbb{P}}(\Omega_l)=1$.
		From now on, we restrict $\omega$ to be in the intersection $\bigcap_{l \in \mathbb{N}} \Omega_l $ which also has probability 1.
		Consequently, also due to the continuity of the likelihoods in $\Lambda$ and $\theta$,  
		\begin{align*}
			0 \leq & \ \bar \ell_{\bar n_k} (\bar \Lambda_{\bar n_k}, \bar \theta_{\bar n_k}) - \bar \ell_{\bar n_k}(\tilde \Lambda_{\bar n_k}, \bar \theta) \\ 
			= &\mathbb E_{\bar{\mathbb P}} \Big[   \Delta_{i} \log \frac{\phi( \gamma^{*^T}  X_i)}{\phi( \bar\gamma^T  X_i)}-\{\Lambda^*(t)\exp(\beta^{*^T}  Z_i)-\bar{\Lambda}(t)\exp(\bar\beta^T  Z_i)\}  +  \Delta_i \log \frac{\dd\Lambda^*( Y_i)}{\dd\bar\Lambda( Y_i)}\\
			&\qquad\qquad+ \Delta_i(\beta^* - \bar \beta)^T  Z_i + \log \frac{\bar S_i( Y_i; \Lambda^*,  \theta^*)}{\bar S_i( Y_i; \bar \Lambda, \bar \theta)} \Big]+o(1)
		\end{align*}
		where the expectation is taken with respect to ${X},{Y},{Z}$ and $\Lambda^*$, $\beta^*$, $\gamma^*$ are fixed (depending on $\omega$).
		For a.e.\ $\omega$, the conditional expectation in the previous display represents a negative KL-divergence of the logistic-Cox model specified by $(\bar \Lambda, \bar \theta)$ from the model specified by $( \Lambda^*(\omega),  \theta^*(\omega))$.
		As a consequence, it must be 0, i.e., $\bar \ell (\Lambda^*, \theta^*) = \bar \ell (\bar \Lambda, \bar \theta)$ $\bar{\mathbb{P}}$-a.e..
		We use this fact to identify all model components, one by one; every equality below is to be understood $\bar{\mathbb{P}}$-a.s..
		
		We first consider ${\Delta}=0$ and ${Y}\geq \bar{\tau}_0$, for which $\bar{S}( Y; \Lambda^*, \theta^*) =\phi( \gamma^{*T}  X) $ and $\bar S( Y; \bar \Lambda, \bar \theta)=\phi(\bar \gamma^{T}  X)$.
		From this we can identify $\gamma^*=\bar\gamma$ a.s.\ for the logistic model. 
		Next, for $ \Delta=0$ and $ Y < \bar\tau_0$, we obtain 
		$\bar S( Y; \Lambda^*,  \theta^*) = \bar S( Y; \bar \Lambda, \bar \theta)$, hence
		\begin{align*}
			\exp (-\Lambda^*( Y) \exp( \beta^{*T} Z))
			= \exp (-\bar\Lambda( Y) \exp(\bar \beta^{T} Z))
		\end{align*}
		Upon inserting different combinations of $ Y$ and $ Z$, 
		we conclude that $\beta^* = \bar \beta$ and $\Lambda^* = \bar \Lambda$ a.s.. 
		
	\end{proof}

The following lemma establishes the consistency of randomly permuted Z-estimators.
Since the proof does not make use of the specific underlying model structure, it is clear that this result holds more generally, i.e., also beyond logistic-Cox cure models.
 
 \begin{lemma}
 \label{lemma:consistency_permutation}
Assume the maximizer $ (\bar\Lambda,\bar\theta)$ defined in \eqref{eqn:bar_parameters} is unique.  The permutation  estimators $ (\hat \Lambda_{n_1,1}^\pi, \hat \theta_{n_1,1}^\pi)$ and $ (\hat \Lambda_{n_2,2}^\pi, \hat \theta_{n_2,2}^\pi)$ converge in probability to $(\bar\Lambda,\bar\theta)$.     
 \end{lemma}
	\begin{proof}
	    		First, we would like to point out that conditional convergence in probability (given a $\sigma$-algebra) is equivalent to the unconditional convergence in probability; a variant of Fact~1 in the Supporting Information of \cite{dobler_et_al_19} similarly holds for the present setting.
		That is why we do not distinguish between conditional and unconditional consistency.

		To prove the consistency of the permuted estimators, we are going to employ the permutation version of the score equations, i.e., 
		$\mathbb{P}_{n_i,i}^\pi \psi_{(\Lambda, \theta),h} \stackrel!=0$ for all indexing $h$, $i=1,2$.
		So far, we know by definition that  $\mathbb{P}_{n_i,i}^\pi \psi_{(\hat \Lambda^\pi_{n_i,i}, \hat\theta^\pi_{n_i,i}),h} = 0$
		and that 
		$\mathbb{P}_{n_1+n_2} \psi_{(\bar \Lambda_{n_1+n_2}, \bar\theta_{n_1+n_2}),h} = 0$.
		
		Also, since 
		$ n_1\mathbb{P}_{n_1,1}^\pi \psi_{( \Lambda, \theta),h} + n_2\mathbb{P}_{n_2,2}^\pi \psi_{( \Lambda, \theta),h}
		= (n_1+n_2)\mathbb{P}_{n_1+n_2} \psi_{( \Lambda, \theta),h}$, 
		we have that 
		$n_1\mathbb{P}_{n_1,1}^\pi \psi_{(\bar \Lambda_{n_1+n_2}, \bar\theta_{n_1+n_2}),h} 
		= - n_2\mathbb{P}_{n_2,2}^\pi \psi_{(\bar \Lambda_{n_1+n_2}, \bar\theta_{n_1+n_2}),h}$.
		Thus, because both of these permuted expressions are connected, we will only focus on the index $i=1$ from now on.
		
		Furthermore, upon integrating out all permutations, it is easy to see that the (conditional) expectation is 
		$\E[\mathbb{P}_{n_1,1}^\pi \psi_{(\bar \Lambda_{n_1+n_2}, \bar\theta_{n_1+n_2}),h} \  | \  Y_i,  \Delta_i,  X_i,  Z_i : i = 1, \dots, n_1+n_2] = 0. $
		Additionally, straightforward and standard algebra for permuted linear statistics for the conditional variance leads to
		$$Var[\mathbb{P}_{n_1,1}^\pi \psi_{(\bar \Lambda_{n_1+n_2}, \bar\theta_{n_1+n_2}),h} \  | \  Y_i,  \Delta_i,  X_i,  Z_i : i = 1, \dots, n_1+n_2] = O_p((n_1+n_2)^{-1}). $$
		
		Consequently, Chebychev's inequality (applied to the conditional distribution) verifies that the permutation-based score equations evaluated at the point $(\bar \Lambda_{n_1+n_2}, \bar\theta_{n_1+n_2})$ all converge to 0 in probability.
		Similar convergences in probability (not necessarily to zero) also hold for other evaluation points.
		
		Hence, the pooled estimator is asymptotically a solution to the permutation-based score equations.
		Now, since $(\hat \Lambda^\pi_{n_1,1}, \hat\theta^\pi_{n_1,1})$ is another (finite sample) solution and the ``true'' solution $(\bar \Lambda, \bar \theta)$ is assumed to be unique,
		the permuted estimator must approach the pooled estimator in probability as the sample size goes in infinity.
		Anything else would contradict the continuity of the map $(\Lambda, \theta) \mapsto \bar{\mathbb{E}}_{\bar{\mathbb{P}}}( \bar \ell (\Lambda, \theta)) $.

\end{proof}

\begin{lemma}
\label{lemma:convergence_permutation}
Assume the maximizer $ (\bar\Lambda,\bar\theta)$ defined in \eqref{eqn:bar_parameters} is unique.    Conditionally on the observations, the process
		\[
		\langle n_i^{1/2}(\hat{\Lambda}^\pi_{n_i,i}-\bar{\Lambda}_{n_1+n_2}),n_i^{1/2}(\hat{\theta}_{n_i,i}^\pi-\bar{\theta}_{n_1+n_2})\rangle,\qquad i=1,2
		\]
		defined as in \eqref{eqn:process_h} and	indexed by $h\in\mathcal{H}_m$ converges weakly in $l^\infty(\mathcal{H}_m)$ to a tight Gaussian process $G^*_i$ in   $l^\infty(\mathcal{H}_m)$, in outer probability.
\end{lemma}
\begin{proof}
    We will apply Theorem \ref{thm:perm_Z_est}. 
		We  need to show that the sample specific estimators $(\hat\Lambda^\pi_{n_i,i},\hat\theta^\pi_{n_i,i})$ satisfy the conditions of Theorem 3.3.1 in \cite{VW96}. Consistency of the estimators was shown in Lemmas~\ref{lemma:consistency_pooled} and \ref{lemma:consistency_permutation} above. 
		Verification of the other conditions  can be done as in \cite{Lu2008}. We omit the other details here since the proof goes along the same lines as points a)-c) below.  
		
		a)				To verify condition \eqref{eq:approx_1} of Theorem \ref{thm:perm_Z_est}, it suffices to show that 
		for any sequence $\epsilon_{n_i}\to  0$,
		\begin{equation}
			\label{eqn:cond1}
			\sup_{\substack{\Vert\Lambda-\bar\Lambda_{n_1+n_2}\Vert_\infty\leq\epsilon_{n_i},\\ \,\Vert\beta-\bar\beta_{n_1+n_2}\Vert\leq\epsilon_{n_i},\\ \Vert\gamma-\bar\gamma_{n_1+n_2}\Vert\leq\epsilon_{n_i}}}\frac{\left|(\p^\pi_{n_i,i}-\p_{n_1+n_2})\psi_{(\Lambda,\theta),h}-(\p^\pi_{n_i,i}-\p_{n_1+n_2})\psi_{(\bar\Lambda_{n_1+n_2},\bar\theta_{n_1+n_2}),h}\right|}{n_i^{-1/2}\vee \Vert\beta-\bar\beta_{n_1+n_2}\Vert\vee\Vert\gamma-\bar\gamma_{n_1+n_2}\Vert\vee \Vert\Lambda-\bar\Lambda_{n_1+n_2}\Vert_\infty }
		\end{equation}
  converges to zero in probability given the data. 
		For simplicity, we can write 
		\[
		\begin{aligned}
			&(\p^\pi_{n_i,i}-\p_{n_1+n_2})\psi_{(\Lambda,\theta),h}-(\p^\pi_{n_i,i}-\p_{n_1+n_2})\psi_{(\bar\Lambda_{n_1+n_2},\bar\theta_{n_1+n_2}),h}\\
   &=\sum_{j=1}^6(\p^\pi_{n_i,i}-\p_{n_1+n_2})a_{j,h}
		\end{aligned}
		\]
		where
		\begin{equation}
			\label{eqn:a_j}
			\begin{aligned}
				a_{1,h}(\delta,y,x,z)&=-h_{21}^Tx\left\{\phi(\gamma^Tx)-\phi(\bar\gamma^T_{n_1+n_2}x)\right\}\\
				a_{2,h}(\delta,y,x,z)&= (1-\delta) h_{21}^Tx\left\{g(y,\Lambda,\theta)-g(y,\bar\Lambda_{n_1+n_2},\bar\theta_{n_1+n_2})\right\}\\
				a_{3,h}(\delta,y,x,z)&= \delta \left\{e^{\beta^Tz}\int_0^{y}h_1(s)\dd\Lambda(s)-e^{\bar\beta^T_{n_1+n_2}z}\int_0^{y}h_1(s)\dd\bar\Lambda_{n_1+n_2}(s)\right\}\\
				a_{4,h}(\delta,y,x,z)&= \delta h^T_{22}z \left\{e^{\beta^Tz}\Lambda(y)-e^{\bar\beta^T_{n_1+n_2}z}\bar\Lambda_{n_1+n_2}(y)\right\}\\
				a_{5,h}(\delta,y,x,z)&= (1-\delta) \left\{g(y,\Lambda,\theta)e^{\beta^Tz}\int_0^{y}h_1(s)\dd\Lambda(s)\right.\\
				&\quad\qquad\qquad-\left.g(y,\bar\Lambda_{n_1+n_2},\bar\theta_{n_1+n_2})e^{\bar\beta^T_{n_1+n_2}z}\int_0^{y}h_1(s)\dd\bar\Lambda_{n_1+n_2}(s)\right\}\\
				a_{6,h}(\delta,y,x,z)&= (1-\delta) h^T_{22}z \left\{g(y,\Lambda,\theta)e^{\beta^Tz}\Lambda(y)\right.\\
    &\quad\qquad\qquad\qquad-\left.g(y,\bar\Lambda_{n_1+n_2},\bar\theta_{n_1+n_2})e^{\bar\beta^T_{n_1+n_2}z}\bar\Lambda_{n_1+n_2}(y)\right\}.
			\end{aligned}
		\end{equation}
		Next we consider the third term. The other terms can be handled similarly. First note that, by Lemma~\ref{lemma:consistency_pooled}, $\bar{\Lambda}_{n_1+n_2}$ and $\bar{\theta}_{n_1+n_2}$ are  consistent estimates of  $\bar{\Lambda}$ and $\bar{\theta}$, and as a result they are bounded on a set of probability converging to one.  From a Taylor expansion we have
		\begin{equation*}
			\label{eqn:a3}
			\begin{split}
				&(\p^\pi_{n_i,i}-\p_{n_1+n_2})a_{3,h}\\
    &=(\beta-\bar\beta_{n_1+n_2})^T\int z\delta e^{\bar\beta^T_{n_1+n_2}z}\int_0^{y}h_1(s)\dd\Lambda(s)\dd 	(\p^\pi_{n_i,i}-\p_{n_1+n_2})(\delta,y,x,z)\\
				&\quad+\int \delta e^{\bar\beta^T_{n_1+n_2}z}\int_0^{y}h_1(s)\dd(\Lambda-\bar\Lambda_{n_1+n_2})(s)\dd 	(\p^\pi_{n_i,i}-\p_{n_1+n_2})(\delta,y,x,z)\\
				&\quad+o^*_{P}\left(\int \delta e^{\bar\beta^T_{n_1+n_2}z}\int_0^{y}h_1(s)\dd(\Lambda-\bar\Lambda_{n_1+n_2})(s)\dd 	(\p^\pi_{n_i,i}-\p_{n_1+n_2})(\delta,y,x,z)\right)\\
				&\quad+o^*_P(\Vert\beta-\bar\beta_{n_1+n_2}\Vert)\\
			\end{split}
		\end{equation*}
		For the first term, since the class of functions that we are integrating is Donsker and uniformly bounded, by Theorem 3.7.2 in \cite{VW96} it follows that, conditionally on the observations
		\begin{equation}
			\label{eqn:a3_term1}
			\sup_{\Vert\Lambda-\bar\Lambda_{n_1+n_2}\Vert_\infty\leq\epsilon_{n_i}}	\int z\delta e^{\bar\beta^Tz}\int_0^{y}h_1(s)\dd\Lambda(s)\dd 	(\p^\pi_{n_i,i}-\p_{n_1+n_2})(\delta,y,x,z)=o^*_P(1).
		\end{equation}
		The second term can be rewritten as 
		\[
		\int_0^{\bar\tau_0} D_n(s)h_1(s)\dd(\Lambda-\bar\Lambda_{n_1+n_2})(s)
		\]
		where
		\[
		D_n(s)=\int\delta \1_{\{y>s\}}e^{\bar\beta^Tz} \dd(\p^\pi_{n_i,i}-\p_{n_1+n_2})(\delta,y,x,z). 
		\]
		By integration by parts and the chain rule we have
		\[
		\begin{split}
			&\int_0^{\bar\tau_0} D_n(s)h_1(s)\dd(\Lambda-\bar\Lambda_{n_1+n_2})(s)\\
			&=  D_n(\bar\tau_0)h_1(\bar\tau_0)(\Lambda-\bar\Lambda_{n_1+n_2})(\bar\tau_0)- \int_0^{\bar\tau_0} (\Lambda-\bar\Lambda_{n_1+n_2})\dd \left[D_n(s)h_1(s)\right]\\
			&=D_n(\bar\tau_0)h_1(\bar\tau_0)(\Lambda-\bar\Lambda_{n_1+n_2})(\bar\tau_0)-\int_0^{\bar\tau_0} (\Lambda-\bar\Lambda_{n_1+n_2})(s) D_n(s)\dd h_1(s)\\
			&\quad +\int \delta (\Lambda-\bar\Lambda)(y)h_1(y) e^{\bar\beta^Tz}\dd 	(\p^\pi_{n_i,i}-\p_{n_1+n_2})(\delta,y,x,z)\\
			&\quad+	\int \delta (\bar\Lambda-\bar\Lambda_{n_1+n_2})(y)h_1(y) e^{\bar\beta^Tz}\dd 	(\p^\pi_{n_i,i}-\p_{n_1+n_2})(\delta,y,x,z)
		\end{split} 
		\]
		Again, by Theorem 3.7.2 in \cite{VW96}, it follows that, conditionally on the observations $D_n=o^*_P(1)$. Since $h_1$ is bounded, it follows 
		\[
		\sup_{\Vert\Lambda-\bar\Lambda_{n_1+n_2}\Vert_\infty\leq\epsilon_{n_i}}	\frac{\left|D_n(\bar\tau_0)h_1(\bar\tau_0)(\Lambda-\bar\Lambda_{n_1+n_2})(\bar\tau_0)\right|}{\Vert\Lambda-\bar\Lambda_{n_1+n_2}\Vert_{\infty}}=o^*_P(1).
		\]
		In addition we also have that conditionally on the observations $\sup_{s\in[0,\tau_0]}|D_n(s)|=o^*_P(1)$ and since $h_1$ is of bounded variation 
		\[
  \begin{split}
     &\sup_{\Vert\Lambda-\bar\Lambda_{n_1+n_2}\Vert_\infty\leq\epsilon_{n_i}}	\frac{ \left|\int_0^{\bar\tau_0}(\Lambda-\bar\Lambda_{n_1+n_2})(s) D_n(s)\dd h_1(s)\right|}{\Vert\Lambda-\bar\Lambda_{n_1+n_2}\Vert_{\infty}}\\
     &\quad\leq \sup_{t\in[0,\bar\tau_0]}|D_n(s)| \int_0^{\bar\tau_0}\,\left|\dd h(s)\right|=o_P^*(1). 
  \end{split}
		\] 
		Since $\Vert\Lambda-\bar\Lambda_{n_1+n_2}\Vert_{\infty}\leq \epsilon_{n_i}$ implies $\Vert\Lambda-\bar\Lambda\Vert_{\infty}\leq \tilde\epsilon_{n_i}$ for some $\tilde{\epsilon}_{n_i}\to 0$, the class  $\{g_{\Lambda}(y,\delta,z)=\delta(\Lambda-\bar\Lambda)(y)h_1(y)e^{\bar\beta^Tz}:\,\Vert\Lambda-\bar\Lambda\Vert_{\infty}\leq \tilde\epsilon_{n_i}\}$ is a Donsker class (product of bounded variation functions, uniformly bounded) and 
		\[
		\E_{\bar\p}\left[\Delta(\Lambda-\bar\Lambda)(Y)^2h_1(Y)^2 e^{2\bar\beta^TZ}\right]=O(\tilde\epsilon_n^2)=o(1)
		,
		\] 
		we have that, conditionally on the data,
		\[
		\sup_{\Vert\Lambda-\bar\Lambda_{n_1+n_2}\Vert_\infty\leq\epsilon_{n_i}}	\sqrt{n}\int \delta (\Lambda-\bar\Lambda)(y)h_1(y) e^{\bar\beta^Tz}\dd 	(\p^\pi_{n_i,i}-\p_{n_1+n_2})(\delta,y,x,z)=o^*_P(1).
		\]
		Finally, since 	$\Vert\bar\Lambda_{n_1+n_2}-\bar\Lambda\Vert_\infty\to 0$ a.s., by Proposition A.5.3 in \cite{VW96} it follows that, conditionally on the data, 
		\[
		\int \delta (\bar\Lambda-\bar\Lambda_{n_1+n_2})(y)h_1(y) e^{\bar\beta^Tz}\dd 	(\p^\pi_{n_i,i}-\p_{n_1+n_2})(\delta,y,x,z)=o^*_P(1).
		\]
		Combining all the results we obtain that
		\[
		\sup_{\substack{\Vert\Lambda-\bar\Lambda_{n_1+n_2}\Vert_\infty\leq\epsilon_{n_i},\\ \,\Vert\beta-\bar\beta_{n_1+n_2}\Vert\leq\epsilon_{n_i},\\ \Vert\gamma-\bar\gamma_{n_1+n_2}\Vert\leq\epsilon_{n_i}}}\frac{\left|(\p^\pi_{n_i,i}-\p_{n_1+n_2})a_{3,h}\right|}{n_i^{-1/2}\vee \Vert\beta-\bar\beta_{n_1+n_2}\Vert\vee\Vert\gamma-\bar\gamma_{n_1+n_2}\Vert\vee \Vert\Lambda-\bar\Lambda_{n_1+n_2}\Vert_\infty }
		\]
  converges to zero in probability, given the data.
		The terms related to the other $a_j$ can be treated similarly.

		b)		Next we check condition~\eqref{eq:cond_weak_conv}. From \eqref{eqn:cond1} it follows in particular that, conditionally on the data, 
		\[
		\sqrt{n_i}(\p^\pi_{n_i,i}-\p_{n_1+n_2})\psi_{(\bar\Lambda,\bar\theta),h}-\sqrt{n_i}(\p^\pi_{n_i,i}-\p_{n_1+n_2})\psi_{(\bar\Lambda_{n_1+n_2},\bar\theta_{n_1+n_2}),h}=o_P^*(1)
		\]	
		almost surely. 
		Hence, it is sufficient to show that 
		\begin{equation}
			\label{eqn:cond2}
			\sqrt{n_i}(\p^\pi_{n_i,i}-\p_{n_1+n_2})\psi_{(\bar\Lambda,\bar\theta),h}\rightsquigarrow Z_1
		\end{equation}			
		on $l^\infty(\mathcal{H}_m)$ in outer probability, where $Z_1$ is a tight random element (actually  a Gaussian process). This follows from Theorem 3.7.1. in \cite{VW96} since the class of functions $\{\psi_{(\bar\Lambda,\bar\theta),h}: h\in\mathcal{H}_m\}$ is Donsker and bounded. This is already shown in step 1 of the proof of Theorem 3 in \cite{Lu2008} (the class of functions is the same, just evaluated at a different point $(\bar\Lambda,\bar\theta)$).
		
		c)			Since the functions $\psi_{(\Lambda,\theta),h}$, $ h\in\mathcal{H}_m$ are the same as in \cite{Lu2008}, it can be proved in the same way that $\bar{\p}{\psi_{(\bar\Lambda,\bar\theta),h}}$ is Fr\'echet-differentiable at $(\bar\Lambda,\bar\theta)$ and the derivative is given by 
		\[
		(\bar{\p}{\dot{\psi}_{(\bar\Lambda,\bar\theta)}})((\Lambda,\theta)-(\bar\Lambda,\bar\theta))(h)=\int_{0}^{\bar\tau_{0}}\bar\sigma_{(1)}(h)\,\dd(\Lambda-\bar\Lambda)(t)+(\theta-\bar\theta)^T\bar\sigma_{(2)}(h)
		\]
		where $\bar{\sigma}_{(1)}$, $\bar\sigma_{(2)}$ are defined as in \eqref{eqn:sigma_1}-\eqref{eqn:sigma_2} respectively, with $\E_{\p_i}$ replaced by $\E_{\bar{\p}}$ and evaluated at $(\bar{\Lambda},\bar\theta)$ instead of $(\Lambda_i,\theta_i)$. Also the proof that the derivative is continuously invertible remains the same as in \cite{Lu2008}. 
		
		d)		For condition \eqref{eq:approx_3}, consider a sequence $(\Lambda_{n_1+n_2},\theta_{n_1+n_2})$ converging to $(\bar\Lambda,\bar\theta)$ as $n_1+n_2\to\infty$. We have
		\[
		\begin{split}
			&(\mathbb{P}_{n_1+n_2} \psi_{\theta_{n_1+n_2},h} - \mathbb{P}_{n_1+n_2} \psi_{\bar\theta,h}) - (\bar{\mathbb{P}}\psi_{\theta_{n_1+n_2},h} - \bar{\mathbb{P}}\psi_{\bar\theta,h})\\
			&=(\mathbb{P}_{n_1+n_2}-\bar{\mathbb{P}})  \psi_{\theta_{n_1+n_2},h} - (\mathbb{P}_{n_1+n_2}-\bar{\mathbb{P}}) \psi_{\bar\theta,h}=\sum_{j=1}^6(\mathbb{P}_{n_1+n_2}-\bar{\mathbb{P}}) \bar{a}_{j,h}
		\end{split}		
		\]
		where $\bar{a}_{j,h}$ are defined as in \eqref{eqn:a_j} replacing ($\Lambda$,$\theta$) and $(\bar{\Lambda}_{n_1+n_2},\bar\theta_{n_1+n_2})$	by 
		$(\Lambda_{n_1+n_2},\theta_{n_1+n_2})$ and $(\bar\Lambda,\bar\theta)$ respectively. We can deal with this similarly to what we did to show \eqref{eqn:cond1}. The difference is that now 	$(\Lambda_{n_1+n_2},\theta_{n_1+n_2})$ and $(\bar\Lambda,\bar\theta)$ are fixed and we consider the class of functions with respect to $h\in\mathcal{H}_m$.  For example, for the term corresponding to $\bar{a}_{3,h}$  we have
		\begin{equation}
			\label{eqn:bar_a3}
			\begin{split}
				&(\mathbb{P}_{n_1+n_2}-\bar{\mathbb{P}}) \bar{a}_{3,h}	\\
    &=(\beta_{n_1+n_2}-\bar\beta)^T\int z\delta e^{\bar\beta^Tz}\int_0^{y}h_1(s)\dd\bar\Lambda(s)\dd 		(\mathbb{P}_{n_1+n_2}-\bar{\mathbb{P}})(\delta,y,x,z)\\
				&\quad+\int \delta e^{\bar\beta^Tz}\int_0^{y}h_1(s)\dd(\Lambda_{n_1+n_2}-\bar\Lambda)(s)\dd 	(\mathbb{P}_{n_1+n_2}-\bar{\mathbb{P}})(\delta,y,x,z)\\
				&\quad+o^*_{P}\left(\int \delta e^{\bar\beta^Tz}\int_0^{y}h_1(s)\dd(\Lambda_{n_1+n_2}-\bar\Lambda)(s)\dd 	(\mathbb{P}_{n_1+n_2}-\bar{\mathbb{P}})(\delta,y,x,z)\right)\\
				&\quad+o^*_P(\Vert\beta_{n_1+n_2}-\bar\beta\Vert).\\
			\end{split}
		\end{equation}
		The integral in the first term converges to zero since the class of functions that we are integrating is Donsker and uniformly bounded.
		The second term can be rewritten as 
		\[
		\int_0^{\bar\tau_0} \bar{D}_{n_1+n_2}(s)h_1(s)\dd(\Lambda_{n_1+n_2}-\bar\Lambda)(s)
		\]
		where
		\[
		\bar{D}_{n_1+n_2}(s)=\int\delta \1_{\{y>s\}}e^{\bar\beta^Tz} \dd(\mathbb{P}_{n_1+n_2}-\bar{\mathbb{P}})(\delta,y,x,z). 
		\]
		By integration by parts and the chain rule we again have 
		\[
		\begin{split}
			&\int_0^{\bar\tau_0} \bar{D}_{n_1+n_2}(s)h_1(s)\dd(\Lambda_{n_1+n_2}-\bar\Lambda)(s)\\
			&=\bar{D}_{n_1+n_2}(\bar\tau_0)h_1(\bar\tau_0)(\Lambda_{n_1+n_2}-\bar\Lambda)(\bar\tau_0)-\int_0^{\bar\tau_0} (\Lambda_{n_1+n_2}-\bar\Lambda)(s) \bar{D}_{n_1+n_2}(s)\dd h_1(s)\\
			&\quad +\int \delta (\Lambda_{n_1+n_2}-\bar\Lambda)(y)h_1(y) e^{\bar\beta^Tz}\dd 	(\mathbb{P}_{n_1+n_2}-\bar{\mathbb{P}})(\delta,y,x,z).
		\end{split} 
		\]
		By the Glivenko-Cantelli theorem $\sup_{s\in[0,\bar\tau_0]}\bar{D}_{n_1+n_2}(s)=o_P(1)$. Since $h_1\in\mathcal{H}_m $ are uniformly bounded, it follows that
		\[
		\sup_{h_1\in\mathcal{H}_m}	\frac{\left|\bar{D}_{n_1+n_2}(\bar\tau_0)h_1(\bar\tau_0)(\Lambda_{n_1+n_2}-\bar\Lambda)(\bar\tau_0)\right|}{\Vert\Lambda_{n_1+n_2}-\bar\Lambda\Vert_{\infty}}=o_P(1).
		\]
		In addition, since $h_1$ are functions of bounded variation and uniformly bounded norm, 
		\[
		\begin{split}
&\sup_{h_1\in\mathcal{H}_m}\frac{ \left|\int_0^{\bar\tau_0} (\Lambda_{n_1+n_2}-\bar\Lambda)(s) \bar{D}_{n_1+n_2}(s)\dd h_1(s)\right|}{\Vert\Lambda_{n_1+n_2}-\bar\Lambda\Vert_{\infty}}\\
&\leq \sup_{s\in[0,\bar\tau_0]}|\bar{D}_{n_1+n_2}(s)| \int_0^{\bar\tau_0}\,\left|\dd h_1(s)\right|=o_P(1).	    
		\end{split}
		\] 
		Finally, from Theorem 2.11.23 in \cite{VW96} it follows that 	\[
		\sup_{h_1\in\mathcal{H}_m}\frac{\sqrt{n_1+n_2}\int \delta (\Lambda_{n_1+n_2}-\bar\Lambda)(y)h_1(y) e^{\bar\beta^Tz}\dd 	(\mathbb{P}_{n_1+n_2}-\bar{\mathbb{P}})(\delta,y,x,z)}{\Vert\Lambda_{n_1+n_2}-\bar\Lambda\Vert_{\infty}}
		\]
  is bounded in probability.
		Combining all the results we obtain 
		\[
		\sup_{h_1\in\mathcal{H}_m}{\left\|(\mathbb{P}_{n_1+n_2}-\bar{\mathbb{P}})a_{3,h}\right\|}=o_P\left( \Vert\theta_{n_1+n_2}-\bar\theta\Vert\vee \Vert\Lambda_{n_1+n_2}-\bar\Lambda\Vert_\infty \right).
		\]
		The other terms can be handled similarly obtaining
		\[
		\begin{split}
			&	\Vert(\mathbb{P}_{n_1+n_2} \psi_{\theta_{n_1+n_2},h} - \mathbb{P}_{n_1+n_2} \psi_{\bar\theta,h}) - (\bar{\mathbb{P}}\psi_{\theta_{n_1+n_2},h} - \bar{\mathbb{P}}\psi_{\bar\theta,h})\Vert\\
			&=o_P\left({ \Vert\theta_{n_1+n_2}-\bar\theta\Vert\vee \Vert\Lambda_{n_1+n_2}-\bar\Lambda\Vert_\infty }\right).	
		\end{split}
		\]
\end{proof}
	\begin{proof}[Proof of Theorem \ref{thm:perm_cov}]
		We proceed similarly to the proof of Theorem~\ref{theo:MST_cov}. From Lemma~\ref{lemma:convergence_permutation}, it follows that the process
  \[
		\langle n_i^{1/2}(\hat{\Lambda}^\pi_{n_i,i}-\bar{\Lambda}_{n_1+n_2}),n_i^{1/2}(\hat{\theta}^\pi_{n_i,i}-\bar{\theta}_{n_1+n_2})\rangle_{i=1,2}(h)
		\]
		converges to a Gaussian process $G^*=(G_1^*,G_2^*)$
  with $G_2^*=-\sqrt{\kappa/(1-\kappa)}G_1^*$. We start by deriving the weak convergence of the process $\sqrt{a_n}(\hat{S}^\pi_{n_1,1}-\bar{S}_{n_1+n_2},\hat{S}^\pi_{n_2,2}-\bar{S}_{n_1+n_2})$, where $a_n=\sqrt{n_1n_2/(n_1+n_2)}$.
		By a series of Taylor expansions we can write
		\[
		\begin{split}
			&\sqrt{a_n}\{\hat{S}^\pi_{n_1,1}(t|z)-\bar{S}_{n_1+n_2}(t|z)\}\\
			&=\sqrt{a_n}\left\{\exp\left(-\hat{\Lambda}^\pi_{n_1,1}(t)e^{\hat\beta^{\pi^T}_{n_1,1}z}\right)-\exp\left(-\bar{\Lambda}_{n_1+n_2}(t)e^{\bar\beta^T_{n_1+n_2}z}\right)\right\}\\
			&=-\sqrt{a_n}\left\{\hat{\Lambda}^\pi_{n_1,1}(t)e^{\hat\beta^{\pi^T}_{n_1,1}z}-\bar{\Lambda}_{n_1+n_2}(t)e^{\bar\beta^T_{n_1+n_2}z}\right\}\exp\left(-\bar{\Lambda}_{n_1+n_2}(t)e^{\bar\beta^T_{n_1+n_2}z}\right)+R_1\\
			&=-\sqrt{a_n}\left\{\hat{\Lambda}^\pi_{n_1,1}(t)-\bar{\Lambda}_{n_1+n_2}(t)\right\}e^{\bar\beta^T_{n_1+n_2}z}\exp\left(-\bar{\Lambda}_{n_1+n_2}(t)e^{\bar\beta^T_{n_1+n_2}z}\right)\\
			&\quad- \sqrt{a_n}\left\{e^{\hat\beta^{\pi^T}_{n_1,1}z}-e^{\bar\beta^T_{n_1+n_2}z}\right\}\bar\Lambda_{n_1+n_2}(t)\exp\left(-\bar{\Lambda}_{n_1+n_2}(t)e^{\bar\beta^T_{n_1+n_2}z}\right)+R_1+R_2\\
			&=-\sqrt{a_n}\left\{\hat{\Lambda}^\pi_{n_1,1}(t)-\bar{\Lambda}_{n_1+n_2}(t)\right\}e^{\bar\beta^Tz}\exp\left(-\bar{\Lambda}(t)e^{\bar\beta^Tz}\right)\\
			&\quad- \sqrt{a_n}\left\{{\hat\beta^\pi_{n_1,1}}-{\bar\beta_{n_1+n_2}}\right\}^Tze^{\bar\beta^Tz}\bar\Lambda(t)\exp\left(-{\bar\Lambda}(t)e^{\bar\beta^Tz}\right)+R_1+R_2+R_3\\
		\end{split}
		\]
		where the remainder terms converge to zero in probability. Considering functions $h\in\mathcal{H}_m$ of the form
		\[
  \begin{split}
		h_{t,z}&=(h_{1;t,z},h_{2;t,z})\\
  &=\left(\1_{[0,t]}(\cdot)e^{\bar\beta^Tz} \exp\left(-\bar{\Lambda}(t)e^{\bar\beta^Tz}\right),\left(\mathbf{0}_p,ze^{\bar\beta^Tz}\bar\Lambda(t)\exp\left(-\bar{\Lambda}(t)e^{\bar\beta^Tz}\right)\right)\right),
  \end{split}
		\] 
		where $\mathbf{0}_p$ denotes a zero vector in $\R^p$, we have
		\[
  \begin{split}
     &\sqrt{a_n}\{\hat{S}^\pi_{n_1,1}(t|z)-\bar{S}_{n_1+n_2}(t|z)\}\\
     &=\sqrt{1-\kappa}\langle n_1^{1/2}(\hat{\Lambda}^\pi_{n_1,1}-\bar{\Lambda}_{n_1+n_2}),n_1^{1/2}(\hat{\theta}^\pi_{n_1,1}-\bar{\theta}_{n_1+n_2})\rangle(h_{t,z})+o_{P^*}(1) 
  \end{split}
			\]
		from which we conclude that, given the data, the process $\sqrt{a_n}\{\hat{S}^\pi_{n_1,1}(t|z)-\bar{S}_{n_1+n_2}(t|z)\}$ converges weakly in $D^\infty([0,\bar\tau_{0}])$ to a mean zero Gaussian process $\bar{G}_{1,z}$ with covariance 
		\begin{equation}
			\label{eqn:rho2}
			\begin{split}
				\rho_{1,z}(s,t)&=Cov(\bar{G}_{1,z}(s),\bar{G}_{1,z}(t))\\
				&=(1-\kappa)\left\{\int_0^{\tau_{0,1}}h_{1;s,z}(u)\bar\sigma_{(1)}^{-1}(h_{t,z})(u)\dd\bar\Lambda(u)+h^T_{2;s,z}\bar\sigma_{(2)}^{-1}(h_{t,z})\right\},
			\end{split}
		\end{equation}
  where $\bar{\sigma}_{(1)}$, $\bar\sigma_{(2)}$ are defined as in \eqref{eqn:sigma_1}-\eqref{eqn:sigma_2} respectively, with $\E_{\p_i}$ replaced by $\E_{\bar{\p}}$ and evaluated at $(\bar{\Lambda},\bar\theta)$ instead of $(\Lambda_i,\theta_i)$
		Defining $\overline{MST}_{u,z}$ as the conditional expected lifetime of the uncured in the pooled sample, we have
		\begin{equation}
			\label{eqn:hat_E-E2}
			\begin{split}
			&	\sqrt{a_n}(\widehat{MST}^\pi_{u,1,z}-\overline{MST}_{u,z})\\
   &=\sqrt{a_n}\int_0^{\bar\tau_{0}}\{\hat{S}^\pi_{n_1,1}(t|z)-\bar{S}_{n_1+n_2}(t|z)\}\dd t\\
				&\quad-\sqrt{a_n}\{\bar{\tau}_{0}-Y^\pi_{1,(m_1)}\}\hat{S}^\pi_{n_1,1}(\bar\tau_{0})-\sqrt{a_n}\{\bar\tau_{0}-{Y}_{(m)}\}\bar{S}_{n_1+n_2}(\bar\tau_{0}),
			\end{split}
 		\end{equation}
  where $ Y^\pi_{1,(m_1)}$ and ${Y}_{(m)}$ denote the largest uncensored observation in the first permuted sample and the pooled sample respectively.
		As in the proof of Theorem~\ref{theo:MST_cov}, the second term and third term in the right hand side of the previous equation can be shown to converge to zero in probability. Considering the map 
		\[
		\psi:{D}[0,\bar\tau_{0}]\to \R\qquad \psi(\xi)=\int_0^{\bar\tau_{0}}\xi(u)\,\dd u
		\] 
		which is Hadamard-differentiable, it follows that, given the data,  $\sqrt{a_n}(\widehat{MST}^\pi_{u,1,z}-\overline{MST}_{u,z},\widehat{MST}^\pi_{u,2,z}-\overline{MST}_{u,z})$ converges weakly to a two dimensional Gaussian random vector
		\[
		(N_1,N_2)=\left(\int_0^{\bar\tau_0}\bar{G}_{1,z}(u)\,\dd u,\int_0^{\bar\tau_0}\bar{G}_{2,z}(u)\,\dd u\right)=(N_1,-\frac{\kappa}{1-\kappa}N_1).
		\]
		Taking the difference of both entries of the pair, we conclude that, given the data, $\sqrt{a_n}(\widehat{MST}^\pi_{u,1,z}-\widehat{MST}^\pi_{u,2,z})$ converges weakly (in probability) to a mean-zero Gaussian random variable with variance 
		\begin{equation}
			\label{eqn:sigma_cov2}
			\sigma^{\pi2}_{z}=\frac{1}{(1-\kappa)^2}\int_0^{\bar\tau_{0}}\int_0^{\bar\tau_{0}}\rho_{1,z}(s,t)\,\dd s\, \dd t,
		\end{equation}
		where $\rho_{1,z}$ is defined in \eqref{eqn:rho2}.
		This concludes the proof.
	\end{proof}
	
	\renewcommand{\P}{\mathbb{P}}
	\newcommand{\bP}{\bar{\mathbb{P}}}
	\section{Permutation of $Z$-estimators in a two-sample set-up}
	\label{sec:app_Z-est}
	In this appendix, we discuss the asymptotic properties of randomly permuted $Z$-estimators.
	For this, we consider a two independent samples set-up with $n_1$ and $n_2$ i.i.d.\ random vectors $W_{11}, \dots, W_{1n_1} \sim \P_1$ and $W_{21}, \dots, W_{2n_2} \sim \P_2$, respectively.
	Let $\Theta$ be a subset of a Banach space, 
	$\Psi_{n_1,1}, \Psi_{n_2,2} : \Theta \to \mathbb{L}$ be random maps, and $\Psi: \Theta \to \mathbb{L}$ be a deterministic map.
	Solutions (or approximate solutions) $\hat \theta_{n_i,i}$ to the equations $\Psi_{n_i,i}(\theta)\stackrel!=0$ will be called \emph{Z-estimators}.
	Due to the i.i.d.\ set-up, we assume the structure
	$ \Psi_{n_i,i}(\theta)h = \mathbb{P}_{n_i,i} \psi_{\theta,h} $, for given measurable functions $\psi_{\theta,h}$ indexed by $\Theta$ and $h \in \mathcal{H}$ for some index set $\mathcal{H}$, where $\mathbb{P}_{n_i,i}$ denotes the $i$-th empirical process.
	Thus, we understand the equation system in the space $\mathbb{L} = \ell^\infty(\mathcal H) $.

	For the random permutation approach, we randomly re-assign the $n_1+n_2$ observations of the pooled sample $(W_{11},\dots, W_{1n_1}, W_{21}, \dots, W_{2 n_2})=:(W_1, \dots, W_{n_1+n_2})$ to the groups 1 and 2 without changing the original sample sizes.
	For a random permutation $\pi$ of the numbers $1, \dots, n_1+n_2$, the permuted samples can be expressed as 
	$W_{\pi(1)}, \dots, W_{\pi(n_1)}$ and 
	$W_{\pi(n_1+1)}, \dots, W_{\pi(n_1+n_2)}$.
	For notational convenience, we denote the permuted samples by 
	$W^\pi_{i1}, \dots, W^\pi_{in_i}$, for sample group $i=1,2$,
	and the corresponding $i$-th permutation empirical process by $\mathbb{P}_{n_i,i}^\pi$.
	Let $\Psi_{n_i,i}^\pi(\theta)h = \mathbb{P}_{n_i,i}^\pi \psi_{\theta,h}  \stackrel!=0$ for all $h\in \mathcal{H}$ be the estimating equation corresponding to $\Psi_{n_i,i}(\theta)\stackrel!=0$, just based on the $i$-th permuted sample.
	We denote the (approximate) solution to the $i$-th permuted estimating equation by $\theta^\pi_{n_i,i}$.
	For future uses, let $\mathbb{G}^\pi_{n_i,i} = \sqrt{n_i}(\mathbb{P}_{n_i,i}^\pi - \mathbb{P}_{n_1+n_2})$ be the $i$-th normalized permutation empirical process,
	where $\mathbb{P}_{n_1+n_2}$ denotes the empirical process of the pooled sample.
	The centering at $\mathbb{P}_{n_1+n_2}$ seems reasonable, as this has an interpretation as a conditional expectation:
	$$ E(\mathbb{P}_{n_i,i}^\pi \psi_{\theta,h} \ | \ W_{ij}: i=1,2; j=1,\dots, n_i) 
	= \mathbb{P}_{n_1+n_2} \psi_{\theta,h} . $$
	Let $\bar{\theta}_{n_1+n_2}$ be the (approximate) solution to $\mathbb{P}_{n_1+n_2} \psi_{\theta,h}\stackrel!=0$ for all $h \in \mathcal H$.
	
	The following theorem represents a version of Theorem~3.3.1 of \cite{VW96} for the random permutation-based estimators.
	\begin{thm}
		\label{thm:perm_Z_est}
		Assume that $\frac{n_1}{n_1+n_2} \to \lambda \in (0,1)$ as $n_1+n_2\to\infty$, define $\bar\P = \lambda \P_1 + (1-\lambda) \P_2$, and assume that Theorem~3.3.1 holds for each sample-specific $Z$-estimator $\hat \theta_{n_1,1}$ and $\hat \theta_{n_2,2}$.
		Let the criterion functions $\psi_{\cdot, h}$ be such that
		\begin{align}
			\begin{split}
				\label{eq:approx_1}
				\| \mathbb{G}_{n_i,i}^\pi (\psi_{\theta_{n_i,i}^\pi, h} - \psi_{\bar \theta_{n_1+n_2}, h}) \|_{\mathcal{H}} = o_{P}^*(1 + \sqrt{n_i} \| \theta^\pi_{n_i,i} - \bar{\theta}_{n_1+n_2} \|) .
			\end{split}
		\end{align}
		Conditionally on $W_{11}, W_{21}, W_{12}, W_{22}, \dots$, assume that 
		\begin{align}
			\label{eq:cond_weak_conv}
			(\sqrt{n_i}(\P^\pi_{n_i,i} - \P_{n_1+n_2}) \psi_{\bar{\theta}_{n_1+n_2},h})_{i=1}^2 \rightsquigarrow (Z_1, Z_2)
		\end{align}
		on $(\ell^\infty(\mathcal{H}))^2$ in outer probability, where $(Z_1, Z_2)$ is a tight random element.

		We assume that $\theta \mapsto \bP\psi_{\theta, h}$ is Fr\'echet-differentiable in $\ell^\infty(\mathcal{H})$ at $\bar \theta$ with a continuously invertible derivative $\bP\dot\psi_{\bar\theta, h}$, and that, for any sequence $(\theta_{n_1+n_2})_{n_1,n_2}$ converging to $\bar \theta$,
		\begin{align}
			\label{eq:approx_3}
			\| (\mathbb{P}_{n_1+n_2} \psi_{\theta_{n_1+n_2},h} - \mathbb{P}_{n_1+n_2} \psi_{\bar\theta,h}) - (\bar{\mathbb{P}}\psi_{\theta_{n_1+n_2},h} - \bar{\mathbb{P}}\psi_{\bar\theta,h}) \|_{\mathcal{H}} = o_P^*(\|\theta_{n_1+n_2} - \bar\theta\|) 
		\end{align}
		as $n_1+n_2 \to \infty$.

		If $\theta^\pi_{n_i,i}$ and $\bar{\theta}_{n_1+n_2}$  satisfy 
		$\| \P^\pi_{n_i,i} \psi_{\theta^\pi_{n_i,i},h} \|_{\mathcal{H}}= o_P^*(n^{-1/2})$, $i=1,2$,
		and\\
  \noindent
		$\| \P_{n_1+n_2} \psi_{\bar\theta_{n_1+n_2},h} \|_{\mathcal{H}}= o_P^*(n^{-1/2})$,
		respectively,
		and if all three estimators converge in outer probability to $\bar \theta$,
		then
		\begin{align}
			\label{eq:result1}
			\sqrt{n_i} (\bP\dot\psi_{\bar\theta, h}) (\theta_{n_i,i}^\pi - \bar{\theta}_{n_1+n_2}) = - \sqrt{n_i}(\P_{n_i,i}^\pi - \P_{n_1+n_2}) \psi_{\bar \theta_{n_1+n_2},h} + o_P^*(1) \rightsquigarrow -(Z_1,Z_2)
		\end{align}
		as $n_1+n_2 \to \infty$
		conditionally on $W_{11}, W_{21}, W_{12}, W_{22}, \dots$, in outer probability.
		Finally, $( \sqrt{n_1} (\theta^\pi_{n_1,1} - \bar{\theta}_{n_1+n_2}) , \sqrt{n_2} (\theta^\pi_{n_2,2} - \bar{\theta}_{n_1+n_2}) ) \rightsquigarrow -((\bP\dot\psi_{\bar\theta, h})^{-1} Z_1, (\bP\dot\psi_{\bar\theta, h})^{-1} Z_2) $ conditionally on $W_{11}, W_{21}, W_{12}, W_{22}, \dots$, in outer probability.
	\end{thm}
	Note that, due to the equality
	$ \P^\pi_{n_2,2} - \P_{n_1+n_2} = -  \frac{n_1}{n_2} (\P^\pi_{n_1,1} - \P_{n_1+n_2}) $, $Z_1$ and $Z_2$ are perfectly negatively linearly correlated:
	$Z_2 = - \sqrt{\frac{\lambda}{1-\lambda}} Z_1$.
	\begin{proof}
		The essential steps of this proof are similar to those in the proof of Theorem~3.3.1 of \cite{VW96}.
		But for the sake of completeness, we shall present the whole proof.

		The assumed consistencies of the pooled and the permuted estimators and then assumption~\eqref{eq:approx_1} entail that
		\begin{align}
			\begin{split}
				\label{eq:331_333_perm}
				& \sqrt{n_i}(\P_{n_1+n_2} \psi_{\theta^\pi_{n_i,i},h} - \P_{n_1+n_2} \psi_{\bar \theta_{n_1+n_2},h}) \\
				& = \sqrt{n_i}(\P_{n_1+n_2} \psi_{\theta^\pi_{n_i,i},h} - \P_{n_i,i}^\pi \psi_{\theta_{n_i,i}^\pi,h}) + o_P^*(1) \\		
    & = - \mathbb{G}_{n_i,i}^\pi ( \psi_{\theta^\pi_{n_i,i},h} - \psi_{\theta_{n_i,i}^\pi,h}) + o_P^*(1) \\
    & = - \sqrt{n_i}(\P_{n_i,i}^\pi - \P_{n_1+n_2}) \psi_{\bar \theta_{n_1+n_2},h} + o_P^*(1 + \sqrt{n_i}\| \theta_{n_i,i}^\pi - \bar{\theta}_{n_1+n_2}\|). 
			\end{split}
		\end{align}
		The consistencies of $\theta^\pi_{n_i,i}$ and $\bar \theta_{n_1+n_2}$ for $\bar \theta$ 
		in combination with the 
		approximation~\eqref{eq:approx_3} applied twice
		imply that the norm of the left-hand side of~\eqref{eq:331_333_perm} equals
		\begin{align}
			\label{eq:barP_approx}
			& \sqrt{n_i}( \|\bar{\mathbb{P}}\psi_{\theta^\pi_{n_i,i},h} - \bar{\mathbb{P}}\psi_{\bar \theta_{n_1+n_2},h}  \|_{\mathcal{H}} + o_P^*(\|\theta^\pi_{n_i,i} - \bar\theta_{n_1+n_2}\| ) ).
		\end{align}
		Additionally, the Fr\'echet-differentiability of $\theta \mapsto \bP\psi_{\theta, h}$ at $\bar \theta$ and the continuous invertibility of the derivative $\bP\dot \psi_{\bar\theta, h} $
		respectively imply that 
		\begin{align}
			\label{eq:psi_lin}
			\| \bP\psi_{\theta, h} - \bP\psi_{\bar\theta, h} \|_{\mathcal{H}} = \| (\bP\dot \psi_{\bar\theta, h})(\theta-\bar\theta) \|_{\mathcal{H}} + o(\| \theta-\bar\theta \|)
		\end{align}
		and that the right-hand side in the previous display is bounded below by 
		$$ c \| \theta-\bar\theta \| + o(\| \theta-\bar\theta \|) $$
		for some positive constant $c$.
		Combine this with~\eqref{eq:barP_approx} and~\eqref{eq:331_333_perm} to see that
		\begin{align*}
			\sqrt{n_i} \|\theta^\pi_{n_i,i} - \bar\theta_{n_1+n_2}\| (c + o_P(1)) \leq O_P(1) + o_P(1 + \sqrt{n_i}\| \theta_{n_i,i}^\pi - \bar{\theta}_{n_1+n_2}\|),
		\end{align*}
		conditionally on $W_{11}, W_{21}, W_{12}, W_{22}, \dots$ in probability.
		Thus, in the same manner, $ \theta^\pi_{n_i,i}$ is $\sqrt{n_i}$-consistent for $\bar\theta_{n_1+n_2}$ in norm.

		Next, apply the approximation in~\eqref{eq:approx_3} to the left-hand side of~\eqref{eq:331_333_perm} and use the Fr\'echet-differentiability of  $\theta \mapsto \bP\psi_{\theta, h}$ at $\bar \theta $
		to find that~\eqref{eq:331_333_perm} equals
		$$ \sqrt{n_i} (\bP\dot\psi_{\bar\theta, h}) (\theta_{n_i,i}^\pi - \bar{\theta}_{n_1+n_2})  + o_P^*(\sqrt{n_i}\|\theta_{n_i,i}^\pi - \bar{\theta}_{n_1+n_2}\|). $$
		Since $o_P^*(\sqrt{n_i}\|\theta_{n_i,i}^\pi - \bar{\theta}_{n_1+n_2}\|)$ and also $o_P^*(1+\sqrt{n_i}\|\theta_{n_i,i}^\pi - \bar{\theta}_{n_1+n_2}\|)$ in the right-hand side of~\eqref{eq:331_333_perm} are both $o_P(1)$, 
		the assertion given in~\eqref{eq:result1} follows from assumption~\eqref{eq:cond_weak_conv}.
		
		The continuity of $(\bP\dot\psi_{\bar\theta, h})^{-1}$ together with the continuous mapping theorem and the just established conditional weak convergence~\eqref{eq:result1} in outer probability
		imply the corresponding conditional weak convergence 
		$$( \sqrt{n_1} (\theta^\pi_{n_1,1} - \bar{\theta}_{n_1+n_2}) , \sqrt{n_2} (\theta^\pi_{n_2,2} - \bar{\theta}_{n_1+n_2}) ) \rightsquigarrow -((\bP\dot\psi_{\bar\theta, h})^{-1} Z_1, (\bP\dot\psi_{\bar\theta, h})^{-1} Z_2)  $$ 
		in outer probability.
	\end{proof}

\section{R code for accessing the breast cancer data set}
\label{sec:app_R}

\begin{verbatim}
if (!require("BiocManager", quietly = TRUE))
    install.packages("BiocManager")
BiocManager::install("curatedBreastData")
library(curatedBreastData)

# read the data, study id 2034
data(curatedBreastDataExprSetList)
D=curatedBreastDataExprSetList$study_2034_GPL96_all

# Relapse-free survival
Data=data.frame(Y=D$RFS_months_or_MIN_months_of_RFS,status=1-D$RFS,
    age=D$age,ER=D$ER_preTrt,sizeTum=D$tumor_stage_preTrt)
\end{verbatim}

\end{appendix}

\begin{acks}[Acknowledgments]
The authors would like to thank Sarah Friedrich for useful comments regarding the data set.
\end{acks}
 
	\bibliographystyle{imsart-number}
	\bibliography{literature}

\end{document}